\documentclass[11pt]{article}

\usepackage[absolute,overlay]{textpos}

\usepackage{float}
\usepackage[dvips]{graphicx}
\usepackage{psfrag}

\usepackage{amsmath,amsthm,amsfonts,amssymb,amscd}
\usepackage[table]{xcolor}
\usepackage{mathrsfs}
\usepackage{upgreek}
\usepackage{stackrel}
\usepackage{tipa}
\usepackage{accents}
\usepackage[multiple]{footmisc}

\usepackage{multirow} 
\usepackage{booktabs}

\numberwithin{equation}{section}

\definecolor{SmartBlue}{RGB}{51, 51, 255}

\usepackage{multicol}

\usepackage{hyperref}

\hypersetup{filecolor=magenta}
\hypersetup{colorlinks=true}
\hypersetup{urlcolor=webbrown}
\hypersetup{linkcolor=SmartBlue}
\hypersetup{pdfpagemode=UseNone}
\hypersetup{citecolor=SmartBlue}

\usepackage[absolute,overlay]{textpos}

\usepackage{changepage}

\usepackage{float}

\usepackage{amsmath,amsthm,amsfonts,amssymb,amscd}
\usepackage{mathrsfs}
\usepackage{upgreek}
\usepackage{stackrel}
\usepackage{tipa}
\usepackage{accents}
\usepackage[multiple]{footmisc}

\usepackage{multirow} 
\usepackage{booktabs}

\usepackage{xfrac} 

\newtheorem{theorem}{Theorem}[section]
\newtheorem{proposition}[theorem]{Proposition}
\newtheorem{lemma}[theorem]{Lemma}
\newtheorem{definition}[theorem]{Definition}
\newtheorem{corollary}[theorem]{Corollary}
\newtheorem{remark}[theorem]{Remark}

\newtheorem*{orient*}{Orientation Condition}
\newtheorem*{junction*}{Junction Conditions}

\def\defi{{\stackrel{\mbox{\tiny {\textbf{def}}}}{\,\, = \,\, }}}

\def\bsff{\textbf{\textup{K}}}
\def\btsff{\widetilde{\bsff}{}}

\def\nablao{{\stackrel{\circ}{\nabla}}}

\def\Riemo{{\stackrel{\circ}{R}}}

\def\sone{\hat{s}}
\def\bsone{\bs{\hat{s}}}

\def\metdata{\{\mathcal{N},\gamma,\ellc,\elltwo\}}

\def\normal{\mathfrak{q}}
\def\bnormal{\bs{\normal}}

\def\G{\mathfrak{G}}

\def\fv{\epsilon}

\def\n{\mathfrak{n}}

\def\D{\mathcal D}
\def\G{\mathcal G}
\def\Fcal{\mathcal F}

\def\uD{\underline{\D}}
\def\ugamma{\underline{\gamma}}
\def\uellc{\underline{\bs{\ell}}}

\def\uelltwo{\underline{\ell}^{(2)}}

\def\bY{\textup{\textbf{Y}}}

\def\Y{\textup{Y}}
\def\bU{\textup{\textbf{U}}}
\def\U{\textup{U}}
\def\bk{\textup{\textbf{K}}}

\def\bF{\textup{\textbf{F}}}
\def\F{\textup{F}}

\def\Yn{r}
\def\Q{Q}

\newcommand{\hatD}{\widehat{\D}}
\newcommand{\hatg}{\widehat{\gamma}}
\newcommand{\hatellc}{\widehat{\ellc}}

\newcommand{\hatelltwo}{\widehat{\ell}^{(2)}}
\newcommand{\hatbY}{\widehat{\bY}}
\newcommand{\hatbU}{\widehat{\bU}}

\newcommand{\hattau}{\widetilde{\tau}}

\newcommand{\hatV}{V'}
\newcommand{\hatn}{\widehat{n}}
\newcommand{\hatP}{\widehat{P}}

\def\MHD{\text{metric hypersurface data}}

\newcommand{\nn}{\nonumber}

\newcommand{\bm}[1]{\mbox{\boldmath $#1$}}
\newcommand\ovnabla{\ov{\nabla}}

\newcommand\N{\mathcal N}
\newcommand\M{\mathcal M}

\newcommand\elltwo{\ell^{(2)}}
\newcommand\ntwo{n^{(2)}}
\newcommand\hypdata{\{ \mathcal{N},\gamma,\ellc, \elltwo, \bY\}}
\newcommand\rig{\zeta}

\newcommand\A{\mathcal A}

\def\bmell{\bm{\ell}}
\def\ellc{\bmell}

\def\nablao{{\stackrel{\circ}{\nabla}}}

\def\defi{{\stackrel{\mbox{\tiny \textup{\textbf{def}}}}{\,\, = \,\, }}}

\def\nablao{{\stackrel{\circ}{\nabla}}}

\def\Riemo{{\stackrel{\circ}{R}}}

\def\nabh{\nabla^h}

\def\sone{s}
\def\bsone{\bs{s}}
\def\s{\mathfrak{r}}

\def\metdata{\{\mathcal{N},\gamma,\ellc,\elltwo\}}

\def\sgn{\textup{sign}}

\def\n{\mathfrak{n}}

\def\D{\mathcal D}
\def\G{\mathcal G}
\def\Fcal{\mathcal F}

\def\uD{\underline{\D}}
\def\ugamma{\underline{\gamma}}
\def\uellc{\underline{\bs{\ell}}}

\def\uelltwo{\underline{\ell}^{(2)}}

\def\bY{\textup{\textbf{Y}}}

\def\Y{\textup{Y}}

\def\bUp{\bU_{\parallel}}

\def\ellp{\ell}

\def\bU{\textup{\textbf{U}}}
\def\U{\textup{U}}
\def\bk{\textup{\textbf{K}}}

\def\bF{\textup{\textbf{F}}}
\def\F{\textup{F}}

\def\Yn{r}
\def\Q{\kappa_n}

\def\bmell{\bm{\ell}}
\def\ellc{\bmell}

\def\nablao{{\stackrel{\circ}{\nabla}}}

\newcommand{\bJ}{\textbf{\J}}
\newcommand{\J}{\textup{J}}

\newcommand{\ov}{\overline}
\newcommand{\nfi}{\varphi}
\newcommand{\cp}{\partial}

\newcommand{\bs}{\boldsymbol}
\newcommand{\lp}{\left(}
\newcommand{\rp}{\right)}

\newcommand{\cu}{\mathcal{U}}
\newcommand{\cv}{\mathcal{V}}
\newcommand{\Mp}{\mathcal{M^+}}
\newcommand{\Ml}{\mathcal{M^-}}

\newcommand{\lb}{\left\lbrace}
\newcommand{\rb}{\right\rbrace}
\newcommand{\lc}{\left[}
\newcommand{\rc}{\right]}
\newcommand{\ld}{\left.}
\newcommand{\rd}{\right.}

\newcommand{\rv}{\right\vert}
\newcommand{\la}{\langle}

\newcommand{\rag}{\rangle_{g}}

\newcommand{\Mpm}{\mathcal{M}^{\pm}}
\newcommand{\tdo}{d}
\newcommand{\lieo}{\mathsterling}

\newcommand{\uhat}{\underaccent{\check}}

\newcommand{\wt}{\widetilde}

\newcommand{\tcr}{}
\newcommand{\tcb}{}

\newcommand{\nullhyp}{\widetilde{\N}}

\newcommand{\spc}{\textup{ }}

\newcommand{\ke}{\widetilde{\kappa}}

\usepackage[new]{old-arrows}

\makeatletter
\newsavebox\myboxA
\newsavebox\myboxB
\newlength\mylenA

\newcommand*\xoverline[2][0.75]{%
    \sbox{\myboxA}{$\m@th#2$}%
    \setbox\myboxB\null
    \ht\myboxB=\ht\myboxA%
    \dp\myboxB=\dp\myboxA%
    \wd\myboxB=#1\wd\myboxA
    \sbox\myboxB{$\m@th\overline{\copy\myboxB}$}
    \setlength\mylenA{\the\wd\myboxA}
    \addtolength\mylenA{-\the\wd\myboxB}%
    \ifdim\wd\myboxB<\wd\myboxA%
       \rlap{\hskip 0.5\mylenA\usebox\myboxB}{\usebox\myboxA}%
    \else
        \hskip -0.5\mylenA\rlap{\usebox\myboxA}{\hskip 0.5\mylenA\usebox\myboxB}%
    \fi}
\makeatother

 \newcounter{mnotecount}

 \newcommand{\mnote}[1]
 {\protect{\stepcounter{mnotecount}}$^{\mbox{\tiny
 $\,\bullet$\themnotecount}}$ \marginpar{
 \raggedright\tiny\em
 $\,\bullet$\themnotecount: #1} }

\allowdisplaybreaks

\usepackage[a4paper,top=25.4mm,bottom=25.4mm,left=18.5mm, right=18.5mm,bindingoffset=0mm]{geometry} 

\title{Abstract Formulation of the Spacetime Matching Problem\\ and Null Thin Shells}

\author{Miguel Manzano\footnote{miguelmanzano06@usal.es}\hspace{0.17cm} and Marc Mars\footnote{marc@usal.es}\\ \\
Instituto de F\'{\i}sica Fundamental y Matem\'aticas, IUFFyM\\
Universidad de Salamanca\\
Plaza de la Merced s/n \\
37008 Salamanca, Spain\\
}

\usepackage{tabularx}

\begin{document}

\setlength{\abovedisplayskip}{0.15cm}
\setlength{\belowdisplayskip}{0.15cm}

\maketitle

\begin{abstract}
The formalism of hypersurface data is a framework to study hypersurfaces of any causal
character \textit{abstractly} (i.e.\ without the need of viewing them as embedded in an ambient space). 
In this paper we exploit this formalism to study the general problem of matching two spacetimes in a fully abstract manner, as this turns out to be advantageous over other approaches in several respects.
  We then concentrate on the case when the boundaries are null
  and prove that the whole matching is determined by a diffeomorphism $\nfi$ on the abstract data set. By exploiting the gauge structure  of the formalism
we find \textit{explicit} expressions for the gravitational/matter-energy content of \textit{any} null thin shell. The results hold for arbitrary topology. A particular  case of interest is when more than one matching is allowed. Assuming that one such  matchings has already been solved, we provide
explicit expressions for the gravitational/matter-energy content of any other shell in terms of the known one.  This situation covers, in particular,
all cut-and-paste constructions, where  one can simply take as known matching the trivial re-attachment of the two regions. 
 We include, as an example, the most general 
matching of
 two regions of the (anti-)de Sitter or Minkowski spacetime across a totally geodesic null hypersurface.
\end{abstract}

%
%

\section{Introduction}\label{sec:Intro:Paper}
The question of under which conditions two spacetimes can be matched across a hypersurface and give rise to a new spacetime is a fundamental problem in any metric theory of gravity. In particular, a matching theory is required 
in any physical situation where a substantial amount of gravitational/matter-energy content is located in a thin enough region of the spacetime (with respect to the dimensions of the problem). 
Then, the matter-content can be modelled as concentrated on a hypersurface. This thin 
shell
of gravitational/matter-energy possesses its own gravity and 
hence affects the spacetime geometry, and 
it is worth finding the
relationship between the 
shell's 
content 
and the properties of 
\tcb{the} spacetime. 

Many authors have contributed to the 
matching 
problem 
in General Relativity, 
see e.g.\ the works  
\cite{darmois1927memorial}, 
\cite{obrien1952jump}, 
\cite{le1955theories}, 
\cite{israel1966singular}, 
\cite{bonnor1981junction},  
\cite{clarke1987junction},  
\cite{barrabes1991thin},  
\cite{mars1993geometry},   
\cite{mars2007lorentzian},  
\cite{mars2013constraint}, 
\cite{senovilla2018equations}.     
The standard approach consists of considering two 
spacetimes 
$(\Mpm,g^{\pm})$ with 
boundaries $\nullhyp^{\pm}$.  
%
For the matching to be possible $\nullhyp^{\pm}$ must be diffeomorphic, i.e.\ 
there must exist a 
diffeomorphism $\Phi:\nullhyp^-\longrightarrow\nullhyp^+$, which we call \textit{matching map}.
One then defines 
the resulting 
spacetime $\M$ 
as the union of $\M^+$ and $\M^-$ 
with the corresponding identification 
of 
boundary points (ruled by $\Phi$). 
%
\tcb{
The necessary and sufficient conditions 
for a metric $g$ to exist on  $\M$ are the so-called (preliminary) 
\textit{matching conditions} 
(or \textit{junction conditions}) and require
$(i)$ that 
the first fundamental forms $\gamma^{\pm}$
from both boundaries coincide, 
$(ii)$ that there exists two riggings $\rig^{\pm}$
(i.e.\ vector fields along $\nullhyp^{\pm}$, everywhere transversal to them)
with the same square norm and such that 
the one-forms $g^{\pm}(\rig^{\pm},\cdot)$ coincide and $(iii)$ that $\rig^{\pm}$ are such that one points inwards and the other outwards. 
When these conditions are fulfilled, the matched spacetime exists. In general, this spacetime will contain a thin shell, which is ruled by the jump in the extrinsic geometry of the matching hypersurfaces.}

\tcb{In addition to this standard approach (also called
\textit{\`a la Darmois})},  
one can also construct 
\textit{null} thin shells with
the so-called \textit{cut-and-paste method} (see e.g.\ \cite{penrose1965remarkable}, 
\cite{dewitt1968battelle}, 
\cite{Penrose:1972xrn},  \cite{podolsky1999nonexpanding}, \cite{podolsky1999expanding}, 
\cite{podolsky2002exact}, 
\cite{griffiths2009exact}, 
\cite{podolsky2014gyratonic}, 
\cite{podolsky2017penrose}, 
\cite{podolsky2019cut}, 
\cite{podolsky2022penrose}),  
%
%
%
where 
the shell 
is described 
via  
a metric with a Dirac delta distribution with support on the matching hypersurface. 
The shell is built by taking a 
spacetime $\lp\mathcal{M},g\rp$  
with
a null hypersurface $\nullhyp\subset\M$, then \textit{cutting} $\M$ along $\nullhyp$,     
which leaves two 
spacetimes $\lp\Mpm,g^{\pm}\rp$,    
and finally reattaching (or \textit{pasting}) $\lp\Mpm,g^{\pm}\rp$ by identifying 
the boundary points so that there exists a jump on  the null direction 
on the matching hypersurface. 

Be that as it may, null shells have been widely studied in the literature  (\tcb{for a sample}, see
\cite{barrabes2003singular}, 
\cite{carballo2021inner}, 
\cite{bhattacharjee2020memory},
\cite{nikitin2018stability},
\cite{chapman2018holographic},
\cite{binetruy2018closed},
\cite{kokubu2018energy},
\cite{fairoos2017massless}), usually by imposing additional symmetries (such as spherical symmetry). In particular, the problem of matching two completely general spacetimes $(\Mpm,g^{\pm})$ with null boundaries $\nullhyp^{\pm}$ has been recently addressed in \cite{manzano2021null} under the \tcb{only} assumption that $\nullhyp^{\pm}$ admit a foliation by diffeomorphic spacelike cross-sections. \tcb{One of the main results in \cite{manzano2021null}  is that all the
  information about the matching is codified in}
a diffeomorphism $\Psi$ between the set of null generators of $\nullhyp^{\pm}$ and function $H$, called \textit{step function}, \tcb{which} corresponds to a shift along the null generators.
\tcb{Another result of interest is that,} 
although generically two given spacetimes can be matched  \textit{at most} in one manner, sometimes \tcb{multiple matchings are possible. A relevant case
of the later, studied in detail in 
\cite{manzano2022general}, is the matching across so-called Killing horizons of order zero.}

\tcb{The matching problem is studied in 
  \cite{manzano2021null}, \cite{manzano2022general} }
by means of the so-called \textit{formalism of hypersurface data} \cite{mars2013constraint}, \cite{mars2020hypersurface} (see also
\cite{mars1993geometry},
\cite{manzano2023constraint},  \cite{Mars2023first}, \cite{mars2023covariant}, \cite{manzano2023field}, \cite{manzano2023master}), \tcb{which}
allows one to codify \textit{abstractly} (i.e.\ in a detached way from an ambient manifold) the intrinsic and extrinsic geometric information of a hypersurface in terms of a \textit{data set} $\D\defi\hypdata$. \tcb{The part
  $\metdata$ is called \textit{metric hypersurface data} and 
  codifies at the abstract level
  the would-be components of the full ambient metric $g$ at the hypersurface.
  The tensor $Y$ codifies extrisinc information. The formalism 
  is equipped with a group of gauge transformations that accounts for the fact that, at the embedded level, the choice of a rigging 
is non-unique. \textit{Two data sets are equivalent if they are related by a gauge transformation.}} Each gauge group element $\G_{(z,V)}\subset\G$ is determined by a nowhere-zero function $z$ and a vector field $V$ in $\N$.

In the language of the \tcb{hypersurface data formalism}, the matching can be performed if and only if \cite{mars2013constraint} one can embed a single metric hypersurface data set 
in both spacetimes (and the corresponding matching riggings satisfy the orientation condition $(iii)$ above).  
In that case, the gravitational/matter-energy content of the shell is fully codified by the jump of the tensors $\bY^{\pm}$ of each side, namely $[\bY]\defi \bY^+-\bY^-$ \cite{mars2013constraint}. 
\tcb{The approach in \cite{manzano2021null}, \cite{manzano2022general}, whilst based on this formalism, still analyzes the matching in terms of the embeddings $\phi^{\pm}$ of  the abstract manifold $\N$ in $\M^{\pm}$ and not directly at a detached level. } 
Two questions arise naturally. The first one is whether there is a way of formulating the matching problem in a fully abstract manner (that is, exclusively in terms of objects defined in the abstract manifold $\N$) so that one does not need to make any reference to the actual spacetimes to be matched. The second is whether one can generalize the results in \cite{manzano2021null}, \cite{manzano2022general} to boundaries with arbitrary topology.
\tcb{The aim of this paper is to answer both questions}.

The first question is solved in Theorem \ref{thm:Best:THM:JC}, where we provide a completely abstract version of the (spacetime) matching conditions. The theorem establishes that two given data sets $\D=\{\N,\gamma,\ellc,\elltwo,\bY^-\}$, $\hatD=\{\N,\hatg,\hatellc,\hatelltwo,\hatbY^+\}$ (each of them should be thought of as an abstraction of one of the boundaries)
can be matched provided that there exists a diffeomorphism $\nfi$ of $\N$ onto itself such that the metric hypersurface data sets $\{\N,\gamma,\ellc,\elltwo\}$, $\{\N,\nfi^{\star}\hatg,\nfi^{\star}\hatellc,\nfi^{\star}\hatelltwo\}$ are 
related by a gauge transformation $\G_{(z,V)}$. 
This can be interpreted as follows.
The map $\nfi$ can be understood as an abstract version of the (spacetime) matching map $\Phi$ mentioned before. 
Since the matching requires that one single metric hypersurface data set is embedded in both spacetimes, 
$\D$ and $\hatD$ 
cannot be arbitrarily different. Instead, there must exist a gauge transformation that compensates for the change induced by $\nfi$ so that, even after applying the pull-back $\nfi^{\star}$, \tcb{the metric part of the data} are still equivalent. Theorem \ref{thm:Best:THM:JC} also imposes a restriction on 
the sign of 
$z$. As we shall see, at the embedded level 
this restriction ensures that the orientation of the 
riggings to be identified in the matching process verifies condition $(iii)$ above. 

When the data sets $\D$, $\hatD$ are embedded in two spacetimes, 
Theorem \ref{thm:Best:THM:JC} is equivalent to the standard matching conditions $(i)$-$(iii)$  
above. \tcb{This result is relevant for several reasons.}
First, it applies \tcb{to} (abstract) hypersurfaces of any causality and any topology.  
Secondly,  
the gauges of the data sets $\D$, $\hatD$ are unfixed so that
at the embedded level 
there is full freedom in the \tcb{a priori choice of the riggings on each side. This gives a lot of flexibility to the framework}. 
Finally, having formulated the matching problem abstractly allows one \tcb{analyze in an independent manner} thin shells with specific 
gravitational/matter-energy content 
and, 
\tcb{on a second stage, study} whether they 
can be embedded in a spacetime.  
This is useful e.g.\ for constructing examples of spacetimes containing certain types of shells.


\tcb{We then concentrate on the null case. We intend to
 generalize the works \cite{manzano2021null}, \cite{manzano2022general} so we impose no topological conditions on the boundaries}
We prove that a (null) metric hypersurface data set $\metdata$ is entirely codified by $\gamma$ and that the remaining metric data is pure gauge. In these circumstances, the feasibility of the matching relies on the tensors $\{\gamma,\hatg\}$ satisfying $\nfi^{\star}\hatg=\gamma$ (and $z$ having suitable sign). \tcb{One of our main results  in the paper is  that we  find} explicit expressions for the riggings to be identified in the matching process, \tcb{as well as of}
the gravitational/matter-energy content of the resulting shell (Theorem \ref{thm:Best:THM:matching}). Specifically, we compute \tcb{\textit{explicitly}} the jump $[\bY]$ and the energy-momentum tensor of the shell in terms of $\D$, $\hatD$ and $\nfi$. \tcb{In particular we provide fully geometric definitions of the energy density $\rho$, energy flux $j$ and pressure $p$ of the shell (Remark  \ref{rem:def:eg:flux:pressure}) and find explicit expression for them. We also codify  the purely gravitational content of a null shell in a tensor $\bY^{\mathfrak{G}}$, which we also compute explicitly. We emphazise that all \tcb{these} result hold for   \textit{any} possible null thin shell.}
  The pressure $p$ of the null shell is worth studying in further detail. It turns out that it can be expressed as a difference of the surface gravities (i.e.\ the ``accelerations") of two null generators of $\N$ related by the push-forward map $\nfi_{\star}$. This \tcb{generalizes previous results in} \cite{manzano2021null}, \cite{manzano2022general}, where \tcb{in specific examples} we noticed that $p$ accounts for an effect of compression/stretching of points when crossing the matching hypersurface.

With the abstract matching formalism, 
one also recovers the property from \cite{manzano2021null}, \cite{manzano2022general} that when the data sets $\D$, $\hatD$ define abstract totally geodesic null hypersurfaces, then an infinite number of matchings are feasible. This situation 
is addressed in Section \ref{sec:mult:mathcings} where, 
assuming that all the information about one of the matchings is  \tcb{known}, we prove that the gravitational/matter energy content of \textit{all} the remaining 
matchings \tcb{can be determined easily and explicity in terms of the known one and the map $\varphi$, with no need} of performing additional calculations.  
The \tcb{specific} case when one of the matchings gives rise to no shell \tcb{is of particular interest} because 
it \tcb{includes} all cut-and-paste constructions\footnote{Note that two regions of the same spacetime can always be matched so that the resulting manifold contains no shell.}. In this context, 
we find 
explicit expressions for the gravitational/matter-energy content of \textit{any} null thin shell constructed with the cut-and-paste method. These results are applied in Section \ref{sec:cut-and-paste:abs}, where we study the matching of two regions of a constant-curvature spacetime across a totally geodesic null hypersurface. 

For the sake of consistency, \tcb{we devote Section
  \ref{sec:matching:fol:abs} to showing how the results in \cite{manzano2021null} are recovered as a particular case of the general framework presented here in the specific case when}
$\N$ can be foliated by spacelike diffeomorphic cross-sections.

The structure of the paper is as follows. In Section \ref{sec:Prelim:matching:abs} we \tcb{review the} results on the geometry of embedded null hypersurfaces, formalism of hypersurface data and matching of spacetimes \tcb{that are needed}  later. In Section \ref{sec:matchingABSTRACT} we provide an abstract formulation of the matching problem. The rest of the paper concentrates on the null case. In particular, Section \ref{sec:null:boundaries:abstract} is devoted to studying \tcb{the} properties of completely general null thin shells and finding explicit expressions for their gravitational/matter-energy content, while  
Section \ref{sec:mult:mathcings} addresses the case when multiple matchings are feasible. In Section \ref{sec:matching:fol:abs}, we establish \tcb{the} connection between the results in \cite{manzano2021null} and the abstract matching formalism \tcb{developed here}. The paper concludes with an example \tcb{where we study all possible matchings involving two regions separated by a totally geodesic null hypersurface in the
 (anti-)de Sitter or Minkowski spacetimes} (Section \ref{sec:cut-and-paste:abs}).

\subsection{Notation and conventions}\label{sec:notation}
In this paper manifolds are smooth, connected and, unless otherwise indicated, without boundary. We use $T\mathcal{M}$ to denote the tangent bundle of a manifold $\mathcal{M}$ and $\Gamma\lp T\mathcal{M}\rp$ for its sections (i.e.\ vector fields).
We also let $\mathcal{F}\lp\mathcal{M}\rp\defi C^{\infty}\lp\mathcal{M},\mathbb{R}\rp$ and $\mathcal{F}^{\star}\lp\mathcal{M}\rp\subset\mathcal{F}\lp\mathcal{M}\rp$ its subset of no-where zero functions. We use the symbols $\pounds$, $d$ to denote Lie derivative and exterior derivative respectively. Both tensorial and abstract index notation will be employed depending on convenience. When  index-free notation is used, we shall often use boldface for covariant tensors. In index notation we use standard font (not boldface) in all cases. We work in arbitrary dimension $\mathfrak{n}$, with the following values for different sets of indices:
\begin{equation}
\label{notation}
\alpha,\beta,...=0,1,2,...,\mathfrak{n};\qquad a,b,...=1,2,...,\mathfrak{n};\qquad A,B,...=2,...,\mathfrak{n}.
\end{equation}
As usual, parenthesis (resp. brackets) will denote symmetrization (resp. antisymmetrization) of indices and we also use the notation
$A\otimes_s B\equiv\frac{1}{2}(A\otimes B+B\otimes A)$ for the symmetrized tensor product of two tensors $A$ and $B$. When $B$ is symmetric, $2$-contravariant we write $\text{tr}_BA$ for the trace with respect to $B$ of any
$2$-covariant tensor $A$. 
Given a semi-Riemannian manifold $(\mathcal{M},g)$, the associated
contravariant metric is called  $g^{\sharp}$ and  $\nabla$ is the Levi-Civita derivative. Scalar products of two vectors are denoted indistinctly as 
$g(X,Y)$ or $\la X,Y\rag$. Our convention for Lorentzian signature is 
$(-,+, ... ,+)$. Finally, the abreviation  ``w.r.t.'' for  ``with respect to" will be used sometimes.

\section{Preliminaries}\label{sec:Prelim:matching:abs}

\subsection{Geometry of embedded null hypersurfaces}\label{sec:Geom:null:hyp:paper:matching}

In this subsection we review some facts about embedded null hypersurfaces, see e.g. \cite{galloway2004null}, \cite{gourgoulhon20063+}, \cite{duggal2007null}. This will serve to fix our notation. An embedded null hypersurface in a
spacetime $\lp\mathcal{M},g\rp$ of dimension $\n+1$ is the image
$\nullhyp=\phi\lp \N\rp$ of an embedding  $\phi:\N\longhookrightarrow\mathcal{M}$ of an $\n$-manifold $\N$, such that the first fundamental form $\gamma \defi \phi^{\star}g$ of $\N$ is degenerate.
Any choice of (nowhere zero) normal vector $k$ to $\nullhyp$ defines a null direction tangent to $\nullhyp$ called \textit{null generator} (and viceversa). The integral curves of $k$  are geodesic and the surface gravity $\ke_{k}\in\Fcal(\nullhyp)$ of $k$ is defined by $\nabla_kk = \ke_kk$. The second fundamental form
of $\nullhyp$ w.r.t $k$ is the tensor 
$\btsff^k(X,Y) \defi g(\nabla_Xk,Y)$, $\forall X,Y\in\Gamma(T\nullhyp)$.
Boundaries of manifolds are always two-sided, so (cf. Lemma 1 in \cite{mars2013constraint}) we shall always assume that $\nullhyp$ admits an everywhere transversal vector field $L$, i.e.\ verifying $L \notin T_p\nullhyp$  $\forall p \in \nullhyp$. The vector $L$ 
can always be taken null everywhere (see e.g. \cite{manzano2021null}). 

A \textit{transverse submanifold} of $\nullhyp$ is any $(\n-1)$-dimensional submanifold $S\subset\nullhyp$ to which $k$ is everywhere transverse. When, in addition, every integral curve of $k$ crosses $S$ exactly once  $S$ is called \textit{cross-section} (or simply \textit{section}). The existence of a cross-section entails a strong topological restriction on $\nullhyp$, as in such case there always exist functions $v\in\Fcal(\nullhyp)$, called \textit{foliation functions}, whose level sets
$S_{v_0} \defi \{p\in\nullhyp\spc\vert\spc v(p)=v_0\in\mathbb{R}\}$ are cross-sections of $\nullhyp$ and $\lb S_{v}\rb$ define a foliation of $\nullhyp$. Nevertheless existence of foliation functions is always granted in sufficiently local domains of $\nullhyp$. Note that necessarily $k(v) \neq 0$ so we can always
assume $k(v)=1$ either by rescalling $k$ or by changing $v$.

Given a transverse submanifold $S\subset\nullhyp$, it is useful \cite{manzano2021null}, \cite{manzano2022general} to define the following tensors on $S$
\begin{equation}
\label{somedefs}
\bs{\Theta}^L\lp X, Y\rp\vert_p\defi\la \nabla_XL,Y\rag\vert_p, \qquad \bs{\sigma}_L\lp X\rp\vert_p\defi-\dfrac{1}{g( L,k)}\la \nabla_Xk,L\rag\vert_p,\qquad\forall X,Y\in T_pS.
\end{equation}
When $L$  is chosen null and orthogonal to $S$ then $\bs{\Theta}^L$ and $\bs{\sigma}_L$ are the second fundamental form and torsion one-form 
of $S$ w.r.t.\ $L$. For any choice of $L$, the tensors $\bs{\sigma}_L$, $\bs{\Theta}^L$
encode extrinsic information of $S$. However, $\bs{\Theta}^L$ is {\emph{not}}
symmetric in general.

Assuming that $\nullhyp$ admits a cross-section $S$, one can construct
a foliation function $v\in\Fcal(\nullhyp)$ and (on local patches) a basis $\lb L,k, v_I\rb$ of $\Gamma\lp T\mathcal{M}\rp\vert_{\nullhyp}$ adapted to the foliation  with the following properties:
\begin{equation}
\label{basis}
\begin{array}{cl}
\textup{(A)} & k\textup{ is a future null generator with surface gravity }\ke_k.\\ [\smallskipamount]
\textup{(B)} & v\in\Fcal(\nullhyp)\textup{ is the only foliation function satisfying }v\vert_S=0,\spc k(v)\vert_{\nullhyp}=1.\\ [\smallskipamount]
\textup{(C)} & \textup{Each vector field }v_I\textup{ is tangent to the foliation, i.e.\ }v_I(v)=0.\\ [\smallskipamount]
%
%
%
\textup{(D)} & \textup{The basis vectors }\lb k,v_I\rb\textup{ are such that }\lc k,v_I\rc=0\textup{ and }\lc v_I,v_J\rc=0.\\ [\smallskipamount]
\textup{(E)} & L\textup{ is a past null vector field everywhere transversal to }\nullhyp.
\end{array}
\end{equation}
For any basis $\{L,k,v_I\}$ verifying \eqref{basis}, we also define $\n$ scalar functions $\{\mu_a\}\subset\Fcal(\nullhyp)$ as
\begin{equation}
\label{eqA30}
\mu_1\lp p\rp\defi g( L,k)\vert_p, 
\qquad \mu_I\lp p\rp\defi g (L,v_I)\vert_p\quad\forall p\in\nullhyp .
\end{equation}
Note that necessarily $\mu_1\neq0$ (this has already been used in \eqref{somedefs}).
The vectors $\{v_A\}$ are spacelike by construction  and $\lb k,v_I\rb$ is a basis of $\Gamma(T\nullhyp)$. Conditions (A) and (B) imply that $v$ increases towards the future. We write $h$ for the induced metric on the leaves $\{S_{v}\}$ and $h_{IJ}\defi g(v_I,v_J)$
for its components in the basis $\{v_I\}$. We use $h_{IJ}$ and its inverse
$h^{IJ}$ 
to lower and raise Capital Latin indices irrespectively of whether they are tensorial or
not (e.g. we let $\mu^I\defi h^{IJ}\mu_J$). 
The property $[k,v_I]=0$ 
entails \cite{manzano2021null}
\begin{equation}
\label{eqA28:plus:on:N}k\big(h(v_I,v_J)\big)\stackbin{\nullhyp}=2\btsff^k(v_I,v_J).
\end{equation}

\subsection{Formalism of hypersurface data}\label{sec:FHD:Prelim}

The \textit{formalism of hypersurface data}, which we introduce next, will allow us to analyze the matching of spacetimes at a fully abstract level. We refer to
\cite{mars2013constraint}, \cite{mars2020hypersurface}, 
\cite{Mars2023first}, \cite{mars2023covariant}, 
\cite{manzano2023constraint}, \cite{manzano2023field}, \cite{manzano2023master} for details.

\subsubsection{General hypersurface data}\label{sec:General:HD:Abs:Matching:Paper}

The fundamental notion of the formalism is \textit{metric hypersurface data}, defined to be 
a set $\metdata$ where $\mathcal{N}$ is an $\mathfrak{n}$-dimensional manifold, $ \gamma $ is a $2$-covariant symmetric tensor, $\ellc$  is a covector and $\ell^{(2)}$ is a scalar function subject to the condition that \hspace{-0.14cm} the symmetric $2$-covariant tensor $\bs{\mathcal{A}}\vert_p$
on $T_p\mathcal{N}\times\mathbb{R}$ given by 
\begin{equation}
\label{ambientmetric}
\begin{array}{c}
\ld\bs{\mathcal{A}}\rv_p\lp\lp W,a\rp,\lp Z,b\rp\rp\defi \ld \gamma \rv_p\lp W,Z\rp+a\ld\ellc\rv_p\lp Z\rp+b\ld\ellc\rv_p\lp W\rp+ab \ell^{(2)}\vert_p,\quad W,Z\in T_p\mathcal{N},\quad a,b\in\mathbb{R}
\end{array}
\end{equation}
is non-degenerate at every $p\in\mathcal{N}$. A priori any signature for
$\bs{\mathcal{A}}\vert_p$
is allowed. 
Given metric hypersurface data, one can define unique tensor fields $\{P^{ab},n^a,\ntwo\}$, with $P$ symmetric, by means of 
\cite{mars2013constraint}

\vspace{-0.25cm}

\begin{multicols}{2}
\noindent
\begin{align}
\gamma_{ab} n^b + \ntwo \ell_a & = 0, \label{prod1} \\
\ell_a n^a + \ntwo \elltwo & = 1, \label{prod2}  \\
P^{ab} \ell_b + \elltwo n^a & = 0,  \label{prod3} \\
P^{ab} \gamma_{bc} + n^a \ell_c & = \delta^a_c. \label{prod4}
\end{align}
\end{multicols}

\vspace{-0.4cm}

No restriction is placed on $\gamma$, which in particular is allowed to be degenerate. However,
%
$\bs{\A}$ being non-degenerate forces $\gamma$ to have at most one degeneration direction \cite{mars2020hypersurface}.
%
%
Specifically, the radical of $\gamma$ at  $p\in\N$, 
defined by 
$\textup{Rad}\gamma\vert_p \defi \{X\in T_p \N \spc\vert\spc\gamma(X,\cdot)=0\}$, 
is either zero- or one-dimensional. The latter case occurs if and only if $\ntwo\vert_p=0$, which by \eqref{prod1} means that $\text{Rad}\gamma\vert_p=\la n\vert_p\rangle$. A point $p\in\N$ is called \textit{null} if $\text{dim}(\text{Rad}\gamma\vert_p)=1$ and \textit{non-null} otherwise. 

The second basic notion of the formalism is \textit{hypersurface data} which is
just $\metdata$ equipped with an extra symmetric $2$-covariant tensor $\bY$, namely  $\D\defi\hypdata$.
It is useful to define the following tensors (note that $\bF$, $\bsone$ and $\bU$ 
only require metric hypersurface data)

\vspace{-0.25cm}

\begin{multicols}{2}
\noindent
\begin{align}
\label{defF}\bF & \defi \frac{1}{2} d \ellc, \\
  \label{sone} \bm{\sone} &\defi  \bF(n,\cdot),\\
    \bs{\Yn}&\defi\bY(n,\cdot) \nonumber \\
  \label{defUtensor} \bU & \defi  \frac{1}{2}\pounds_{n} \gamma  + \ellc \otimes_s d \ntwo.\\
   \label{defK}\bk&\defi \ntwo \bY +\bU.  \\
   \Q & \defi -\bY(n,n). \label{defY(n,.)andQ} 
\end{align}
\end{multicols}

\vspace{-0.4cm}
(Metric) hypersurface data has a built-in gauge group structure \cite{mars2020hypersurface} with the following properties. 
  \begin{definition}\label{def:gauge_things}
Let $\D=\hypdata$ be hypersurface data, $z\in\mathcal{F}^{\star}\lp\mathcal{N}\rp$ and $V\in\Gamma\lp T\mathcal{N}\rp$. The gauge transformed data 
$\G_{(z,V)}(\D)\defi \lb \mathcal{N},\mathcal{G}_{\lp z,V\rp}\lp \gamma \rp,\mathcal{G}_{\lp z,V\rp}\lp\ellc\rp,\mathcal{G}_{\lp z,V\rp}\big( \ell^{(2)} \big),\mathcal{G}_{\lp z,V\rp}( \bY )\rb$ 
is defined as
\begin{align}
\label{gaugegamma&ell2} \hspace{-0.25cm}\mathcal{G}_{\lp z,V\rp}\lp \gamma \rp& \defi  \gamma ,\qquad\mathcal{G}_{\lp z,V\rp}\lp\ellc\rp \defi z\lp\ellc+ \gamma \lp V,\cdot\rp\rp,\qquad\mathcal{G}_{\lp z,V\rp}\big(\ell^{(2)}\big) \defi z^2\big( \ell^{(2)}+2\ellc\lp V\rp+ \gamma \lp V,V\rp\big),\\
\label{gaugeY}\hspace{-0.25cm}\mathcal{G}_{\lp z,V\rp}\lp \bY\rp  & \defi z\bY+ \ellc\otimes_s \tdo z+\frac{1}{2}\lieo_{zV} \gamma=z\bY+ \lp \ellc+\gamma(V,\cdot)\rp\otimes_s dz+\frac{z}{2}\lieo_{V} \gamma.
\end{align} 
\end{definition}
The set of all possible gauge transformations forms a group $\mathcal{G}=
\Fcal^{\star}(\N)\times\Gamma(T\N)$ 
with composition law $\mathcal{G}_{\left(z_2, V_2\right)} \circ\mathcal{G}_{\left(z_1, V_1\right)}=\mathcal{G}_{\left(z_1 z_2 , V_2+z_2^{-1} V_1\right)}$, identity
$\mathcal{G}_{(1,0)}$ and inverse  $\mathcal{G}^{-1}_{(z,V)} \defi \mathcal{G}_{(z^{-1},-zV)}$. 

All 
considerations 
so far 
make no reference to any ambient space where $\N$ 
is
embedded. The abstract construction and
the usual geometry of embedded hypersurfaces are connected through the notion of embeddedness of the data.
%
Given  a semi-Riemannian $(\mathfrak{n}+1)$-dimensional manifold
$\lp \mathcal{M},g\rp$
we say $\metdata$ is \textit{embedded with embedding $\phi$ and rigging $\zeta$} in $\lp \mathcal{M},g\rp$
provided there exists an embedding $\phi :\mathcal{N}\longhookrightarrow\mathcal{M}$ and a rigging 
$\zeta$ (i.e.\ a vector field along $\phi \lp\mathcal{N}\rp$, everywhere transversal to it) satisfying  
\begin{equation}
\label{emhd}
\phi ^{\star}\lp g\rp= \gamma , \qquad\phi ^{\star}\lp g\lp\zeta,\cdot\rp\rp=\ellc, \qquad\phi ^{\star}\lp g\lp\zeta,\zeta\rp\rp=\elltwo.
\end{equation}
The same notion for hypersurface data $\hypdata$ requires, in addition,
\begin{equation}
\label{YtensorEmbDef}
\dfrac{1}{2}\phi^{\star}\lp \pounds_{\zeta}g\rp=\bY.
\end{equation}
We often simplify the notation and say simply that
the data is ``$\{\phi,\rig\}$-embedded". We also identify scalars and vectors in
 $\N$ with their corresponding images on $\phi(\N)$ when there is no risk of confusion. 
The action of the gauge group in the data corresponds to a change of rigging according to 
\cite{mars2020hypersurface}
\begin{equation}
\label{gaugerig}\mathcal{G}_{(z,V)} (\rig)  \defi  z (\rig + \phi _{\star} V).
\end{equation}
More specifically, it holds that if
$\metdata$ is $\{\phi,\rig\}$-embedded
in $(\mathcal{M},g)$, then
$\mathcal{G}_{(z,V)}(\metdata)$ is
$\{\phi, \mathcal{G}_{(z,V)} (\rig)\}$-embedded in the same space.

The 
hypersurface $\phi(\N)$ admits
a unique normal $\nu$  satisfying $g(\nu,\rig)=1$, which decomposes as 
\cite{mars2013constraint}, \cite{mars2020hypersurface}
\begin{align}
\label{normal}
\nu & = \ntwo \rig + \phi_{\star}n. 
\end{align}
It then turns out that 
$\bk$ (defined in \eqref{defK}) is the second fundamental form of $\phi(\N)$ w.r.t.\ $\nu$ \cite{mars2020hypersurface}, i.e.\ 
\begin{equation}
\label{2FF}\bk=\phi^{\star}(\nabla\bs{\nu}),\qquad \bs{\nu}\defi g(\nu,\cdot).
\end{equation}
Observe that $\bk$ and $\bU$ coincide at null points of $\N$. 
Although generically $\N$ is not a semi-Riemannian manifold, it admits
two useful covariant derivatives.
The \textit{metric hypersurface connection} $\nablao$ 
depends only on the metric part of the data and it is defined uniquely
\cite{mars2013constraint} by the properties of being torsion-free together with the expressions

\vspace{-0.3cm}

\begin{multicols}{2}
\noindent
\begin{align}
\nablao_{a} \gamma_{bc}  =& - \ell_b \U_{ac} - \ell_c \U_{ab}, \label{nablaogamma} \\
  \nablao_a \ell_b  =&\spc \F_{ab} - \elltwo \U_{ab}. \label{nablaoll}
\end{align}
\end{multicols}

\vspace{-0.4cm}

The second connection is called \textit{hypersurface connection} and
denoted by $\ovnabla$. It is also torsion-free and relates to the former by $\ovnabla_{X}Z=\nablao_XZ-\bY(X,Z)n$ for any  $X,Z\in\Gamma(T\N)$.
When $\hypdata$ is $\{\phi,\rig\}$-embedded in  $(\M,g)$, the ambient Levi-Civita connection $\nabla$ 
and the derivatives $\nablao$, $\ovnabla$ satisfy \cite{mars2013constraint} 
\begin{align}
\label{nablaXYnablao}\nabla_{X}Z&=\nablao_{X}Z-\bY(X,Z)\nu - \bU(X,Z)\rig =\ovnabla_{X}Z-\bk(X,Z)\rig, \\
\label{nablaXrig}
\tcr{  \la  \nabla_{X} \rig, Z \rag} & 
  =  \bY (X,Z) + \bF (X,Z)  
\end{align}
for all $X,Z\in\Gamma(T\N)$.
Thus, 
$\ovnabla$ is the connection induced from $\nabla$ along the rigging
\cite{mars1993geometry}. 
Two 
consequences of the definition of $\nablao$  
are
\begin{align}
 \nablao_b n^c  
  =&\spc n^c \lp \sone_b- \ntwo (d \elltwo)_b\rp + P^{ac}\lp \U_{ba}- \ntwo \F_{ba}\rp, \label{nablaon}\\
\label{contrNsym} n^b\lp \nablao_{b}\theta_{d}+\nablao_{d}\theta_{b}\rp=&\spc\pounds_{n}\theta_d+ \nablao_{d}(\bs{\theta}({n}))-2\lp \bs{\theta}(n) \lp \sone_d - \ntwo  \nablao_d \elltwo\rp + P^{ab}\theta_{b}\lp  \U_{da}- \ntwo  \F_{da}\rp \rp.
\end{align}
where $\theta_a$ is an arbitrary one-form. Their explicit proof 
can be found in \cite{mars2020hypersurface} and   
\cite[Lem.\ 2.5]{manzano2023constraint} respectively. We shall also need the following lemma relating Lie and $\nablao$ derivatives.
\begin{lemma} \label{lem:liegamma}
  Let $\metdata$ be metric hypersurface data, $V^a$ any vector field
  and $\mathfrak{w}_a$ any covector field. Define $\uhat{V}_a\defi  \gamma_{ab} V^b$
  and $\hat{\mathfrak{w}}^a\defi P^{ab} \mathfrak{w}_b$. Then the following identities hold 
  
\vspace{-0.3cm}

\begin{multicols}{2}
\noindent
  \begin{align}
     \frac{1}{2} \pounds_{V} \gamma_{ab} & = \ellc(V) \U_{ab} + \nablao_{(a} \uhat{V}_{b)}, \label{derVgamma} \\
     \frac{1}{2} {\pounds}_{\hat{\mathfrak{w}}} \gamma_{ab}   &= \nablao_{(a} \mathfrak{w}_{b)}
      -  \ell_{(a} \nablao_{b)} \mathfrak{w}(n). \label{deromegagamma} 
      \end{align}
\end{multicols}

\vspace{-0.5cm}

\end{lemma}
\begin{proof}
    We first note that
    $     \nablao_c \gamma_{ab} - \nablao_a \gamma_{bc} -
      \nablao_{b} \gamma_{ac} = 2 \ell_{c} \U_{ab}$
    as a direct consequence of \eqref{nablaogamma}. Moreover, since $\nablao$ has no torsion, the Lie derivative of any $p$-covariant tensor $T$ along any direction $V\in\Gamma(T\N)$ reads 
$(\pounds_V T)_{a_1 \cdots a_p}=V^b\nablao_b T_{a_1 \cdots a_p}+\sum_{\mathfrak{i}=1}^{\mathfrak{p}}T_{a_1\cdots a_{\mathfrak{i}-1}ba_{\mathfrak{i}+1}\cdots a_p}\nablao_{a_\mathfrak{i}}V^b$. 
Particularizing this  
for $T= \gamma$ we get 
    \begin{align*}
      \pounds_{V} \gamma_{ab} & = V^c \nablao_c \gamma_{ab}
      +2 \gamma_{c(a}\nablao_{b)}   V^c
       =
      V^c \left (
      \nablao_c \gamma_{ab} - \nablao_a \gamma_{bc} -
      \nablao_{b} \gamma_{ac}\right ) +2\nablao_{(a} \uhat{V}_{b)}
  = 2 \ellc (V) \U_{ab} + 2 \nablao_{(a} \uhat{V}_{b)}
    \end{align*}
    which is \eqref{derVgamma}. To prove the second identity we apply
    \eqref{derVgamma}
    to $V = \hat{\mathfrak{w}}$. Since by \eqref{prod4} we have
$\gamma_{ab} P^{bc} \mathfrak{w}_c = \mathfrak{w}_a -   \mathfrak{w}(n)\ell_a$, identity \eqref{derVgamma} gives $\frac{1}{2} \pounds_{\hat{\mathfrak{w}}} \gamma_{ab} =  \ellc ( \hat{\mathfrak{w}}) \U_{ab}
  + \nablao_{(a} \left ( \mathfrak{w}_{b)} - \mathfrak{w}(n) \ell_{b)} \right ).$
From \eqref{prod3}, we find $\ellc(\hat{\mathfrak{w}})  =  - \elltwo \mathfrak{w}(n)$. Inserting above yields $\frac{1}{2} \pounds_{\hat{\mathfrak{w}}} \gamma_{ab} =  - \elltwo \mathfrak{w}(n)  \U_{ab}
  + \nablao_{(a}  \mathfrak{w}_{b)} 
  - \ell_{(a} \nablao_{b)} \mathfrak{w}(n) 
  -  \mathfrak{w}(n) \nablao_{(a} \ell_{b)}$, 
which simplifies to \eqref{deromegagamma} after taking into account \eqref{nablaoll}. 
\end{proof}
From a covector and a function on $\N$, one can build a unique vector field according to the next lemma.
\begin{lemma}\label{lemBestLemmaMarc} 
\textup{\cite{mars2013constraint}} 
Let $\metdata$ be metric hypersurface data. Given a covector field $\bs{\varrho}\in\Gamma(T^{\star}\N)$ and a scalar function $u_0\in\Fcal(\N)$, there exists a vector field $W\in\Gamma(T\N)$ satisfying $\gamma(W,\cdot)=\bs{\varrho}$, $\ellc(W)=u_0$ if and only if $\bs{\varrho}(n)+\ntwo u_0=0$. Such $W$ is unique and reads $W=P(\bs{\varrho},\cdot)+u_0n$.
\end{lemma}

\subsubsection*{Matter-hypersurface data and abstract thin shells}\label{sec:thinshells}

Hypersurface data encodes (abstractly) the intrinsic and extrinsic information of embedded hypersurfaces.  
In the context of gravity, knowning the matter contents of the spacetime determines part of the curvature, typically by means of the Einstein tensor $\textbf{\textup{Ein}}_g$. Thus, 
to codify matter information abstractly 
we need to supplement the data with additional quantities. For general hypersurfaces, 
only the normal-transverse and the normal-tangential components of 
$\textbf{\textup{Ein}}_g$ can be related exclusively to intrinsic and extrinsic information of the hypersurface \cite{israel1966singular}, \cite{barrabes1991thin},
  \cite{mars1993geometry}, \cite{mars2013constraint}. 
%
%
Hence, the additional (matter) data 
involves a scalar $\rho$ and a covector $\bJ$ that,  
once the data is embedded,  
correspond to 
such components of 
$\textbf{\textup{Ein}}_g$. 
Their relation with the rest of the data needs to be imposed as constraint equations. They are well-known in the spacelike case (see e.g.\ \cite{bartnik2004constraint}), and were generalized to 
arbitrary causal character in \cite{mars2013constraint}. 

Note that although we refer to the variables $\rho$ and $\bJ$ as {\it matter} variables, what we are actually prescribing are certain components of the Einstein tensor. The terminology is justified because in General Relativity (with vanishing cosmological constant) $\rho$ and $\bJ$ indeed correspond to the matter
four-momentum along the normal direction. 
However, we emphasize that 
we are not assuming any field equations and that the geometric approach that we take can be used in any theory of gravity.

The abstract definition of matter-hypersurface data is as follows.
\begin{definition}\label{def:matterHD}
\textup{\cite{mars2013constraint}}
(Matter-Hypersurface data) A tuple $\{\N,\gamma,\ellc,\elltwo,\bY, \rho_{\ell}, \bJ\}$ formed by hypersurface data $\hypdata$, a scalar $\rho_{\ell}\in\Fcal(\N)$ and a one-form $\bJ\in\Gamma(T^{\star}\N)$ is matter-hypersurface data if 
$\G_{(z,V)}(\rho_{\ell}) = \rho_{\ell} + \bJ(V)$, $
\G_{(z,V)}(\bJ) = z^{-1} \bJ$
and the following identities, called constraint equations, hold:
\begin{align}
\nonumber\rho_{\ell} =&\spc 
 \frac{1}{2} \Riemo{}^{c}_{\,\,\,bcd} P^{bd} + \frac{1}{2}
\ell_a \Riemo{}^{a}_{\,\,\,bcd} P^{bd} n^c+
\nablao_d \lp ( P^{bd} n^c - P^{bc} n^d ) \Y_{bc} \rp+  \ntwo P^{bd} P^{ac}  \Y_{b[c} \Y_{d]a}
 \\
\label{EinnlHypData_a_3}  &+\frac{1}{2} (P^{bd} n^c - P^{bc} n^d )
\Big(  \elltwo \nablao_d \U_{bc} +( \U_{bc} + \ntwo \Y_{bc}  ) \nablao_d \elltwo + 2 \Y_{bc} ( \F_{df} - \Y_{df} ) n^f
\Big) ,\\
\J_c  =&\spc 
   \ell_a \Riemo{}^a_{\,\,\,bcd} n^b n^d
- 2\nablao_f \lp  ( \ntwo P^{bd} - n^b n^d ) 
 \Y_{b[c}\delta^f_{d]}   \rp  + 2  ( P^{bd} - \elltwo n^b n^d )  
\nablao_{[c} \U_{d]b} - 2P^{bd} n^f \U_{b[c} \F_{d]f} \nonumber \\ 
& 
- ( \ntwo P^{bd} - n^b n^d ) \lp  
  ( \U_{b[c} + \ntwo \Y_{b[c} ) \nablao_{d]} \elltwo + 2 \Y_{b[c} \F_{d]f}   n^f \rp  - ( P^{bd} n^f - P^{bf} n^d ) \Y_{bd} \U_{cf}
.
\label{EinneHypData_b_3} 
\end{align}
\end{definition}
The next theorem
justifies both 
the gauge behaviour of $\{\rho_{\ell},\bJ\}$ 
and
the explicit form of \eqref{EinnlHypData_a_3}-\eqref{EinneHypData_b_3}.
\begin{theorem}\label{thm:EMBmatterHD}
\textup{\cite{mars2013constraint}}
Let $\{ \N,\gamma,\ellc,\elltwo,\bY, \rho_{\ell}, \bJ \}$ be
matter-hypersurface data and assume that the hypersurface data $\hypdata$ is $\{\phi,\rig\}$-embedded in a semi-Riemannian manifold $(\M,g)$.  
Then, 
\begin{multicols}{2}
\noindent
\begin{align}
\label{rho:IE}-\rho_{\ell} &= \phi^{\star} \left ( \textbf{\textup{Ein}}_g ( \rig, \nu) \right ), \\
\label{J:IE}-\bJ &= \phi^{\star} ( \textbf{\textup{Ein}}_g( \cdot, \nu) ),
\end{align}
\end{multicols}

\vspace{-0.35cm}

where $\textbf{\textup{Ein}}_g$ is the ($2$-covariant) Einstein tensor of $(\M,g)$ and $\nu$ the (unique) normal vector field along $\phi(\N)$ satisfying $g(\rig,\nu)=1$.
\end{theorem}
As we shall see further on, 
the matching problem involves pairs of matter-hypersurface data. However, at this point we simply put forward various definitions and explore some of their consequences. 
%
%
\begin{definition}\label{def:thinshell}
  (Thin shell) A thin shell is a pair of matter-hypersurface data with same metric hypersurface data, i.e.\ of the form $\{\N,\gamma,\ellc,\elltwo,\bY^{\pm},\rho_{\ell}^{\pm},\bJ^{\pm},\fv\}$, where $\fv$ is a sign with gauge behaviour:
\begin{equation}
\label{gauge:sign}\G_{(z,V)}(\fv)=\frac{z}{\vert z\vert}\fv.
\end{equation}
\end{definition}
We write $\mathcal{Q}^{\pm}$ for quantities constructed from $\{\N,\gamma,\ellc,\elltwo,\bY^{\pm}\}$ and let $[\mathcal{Q}] \defi \mathcal{Q}^+-\mathcal{Q}^-$ be its jump.


One of the main properties of thin shells is that one can define an energy-momentum tensor encoding their matter-energy content. In a completely general case, 
this is done as follows. 
\begin{definition}\label{def:energy-momentum_tensor}
(Energy-momentum tensor) For a thin shell $\{\N,\gamma,\ellc,\elltwo,\bY^{\pm},\rho_{\ell}^{\pm},\bJ^{\pm},\fv\}$, 
the energy-momentum tensor is the symmetric $2$-covariant tensor $\tau$ defined by
\begin{align}
\label{tau} \tau^{df} \defi &\spc \fv \Big(  (P^{af} n^d  +  P^{ad} n^f)n^b   
 -(\ntwo P^{af}P^{bd}  + P^{df} n^an^b)+ P^{ab}( \ntwo P^{df}-  n^dn^f)\Big) {[\Y_{ab}]}.
\end{align}
\end{definition}
\begin{remark} \label{rem:sign:mathfrak(v)}
  Definitions \ref{def:thinshell}
  and \ref{def:energy-momentum_tensor}
  are a modification of the previous ones introduced 
  in \textup{\cite{mars2013constraint}}, which involved no $\fv$. The addition
  of the sign $\fv$ is necessary in order for 
  $\tau$ to retain its physical interpretation as energy-momentum tensor (density) in all gauges. Indeed, 
  a change in the orientation of 
  $\ellc$ (or of rigging in the embedded picture) introduces a sign in $[\bY]$ (by \eqref{gaugeY}). The value of
  $\tau$ cannot be sensitive to this, so one needs to introduce a sign
  $\fv$ with gauge behaviour \eqref{gauge:sign} to compensate the change of sign in $[\bY]$ (in fact, one checks easily that the 
  gauge behaviour of
  $\tau$ is $\G_{(z,V)}(\tau)=\vert z\vert^{-1}\tau$). To be more specific, when one deals
  with thin shell data $\{\N,\gamma,\ellc,\elltwo,\bY^{\pm},\rho_{\ell}^{\pm},\bJ^{\pm}\}$ $\{\phi^{\pm},\rig^{\pm}\}$-embedded in $(\Mpm,g^{\pm})$,  
  the sign $\fv$ must be chosen positive if $\rig^-$ points outwards w.r.t.\ $(\Ml,g^-)$ and negative otherwise.
\end{remark}

The tensor field $\tau$ 
has the symmetries of an energy-momentum tensor 
and 
coincides with the Israel energy-momentum tensor of the shell \cite{israel1966singular} whenever it does not contain null points. Moreover, for null thin shells, 
the definition of energy-momentum tensor provided in \cite[Eq. (31)]{barrabes1991thin} by Barrab\'es and Israel yields precisely $\tau$. 
In a spacetime $(\M,g)$ resulting from a matching, given a basis $\{e_a\}$ of $\Gamma(T\nullhyp)$ where $\nullhyp$ is the matching hypersurface, one can also check that the quantity $\tau^{ab}e_a^{\mu}e_b^{\nu}$ gives the singular part of the Einstein tensor of $(\M,g)$, as it is written in \cite[Eq. (71)]{mars1993geometry}. 
The gauge behaviour of $\tau$ is key in the embedded case, as it ensures that the singular part of the Einstein tensor of the matched spacetime remains invariant under rescaling the normal vector $\nu$. All these  reasons justify the Definition \ref{def:energy-momentum_tensor} for the energy-momentum tensor on a thin shell \cite{mars2013constraint}, irrespectively of whether the data is embedded. 

At null points (and only there), $\tau=0$ is compatible with a non-trivial jump of the geometry.
Indeed, in order to get $\tau=0$ when $\ntwo=0$, it suffices to require $[\bY](n,\cdot)=0$ and $\text{tr}_P[\bY]=0$, which does not mean that the whole tensor $[\bY]$ vanishes identically. Physically, this situation corresponds to an
impulsive gravitational wave supported on the shell. This behaviour is possible only at null points. At non-null points $\tau=0$ implies,
in addition, that $P^{af}P^{bd} {[\Y]}_{ab}=0$ which entails
$0=\gamma_{fi}\gamma_{dj}P^{af}P^{bd} {[\Y]}_{ab}=(\delta^{a}_{i}-n^a\ell_i)(\delta_j^b-n^b\ell_j){[\Y]}_{ab}=[\Y]_{ij}$, i.e.\ abscence of jumps in the geometry.
In particular, this  means that non-trivial thin shells with vanishing energy-momentum tensor can only exist on null points.

\subsubsection{Null hypersurface data}\label{sec:Null:HD:cap2}

A particular case of relevance for the matching problem is when the hypersurfaces are null everywhere. It is immediate to translate this notion to the abstract level.
\begin{definition}
  \label{defNHD}
  (Null (metric) hypersurface data) A metric hypersurface data $\hspace{-0.02cm}\metdata\hspace{-0.02cm}$ or a hypersurface data $\hspace{-0.02cm}\hypdata\hspace{-0.02cm}$ is null if the scalar $\ntwo\hspace{-0.08cm}$ given by \eqref{prod1}-\eqref{prod4} is zero everywhere on $\N$. 
\end{definition}

Let us describe the main properties of the formalism in the null case.
We refer to \cite{manzano2023constraint} for proofs and additional results.
We already know that $\ntwo=0$ implies $\textup{Rad}(\gamma)=\langle n\rangle$ and therefore $\gamma(n,\cdot)=0$. Moreover, the tensors $\bsone$ and $\bU\defi\frac{1}{2}\pounds_n\gamma$ defined in \eqref{sone} and \eqref{defUtensor} verify
\begin{align}
\label{s(n)=U(n,-)=0}\bsone(n)=0,\qquad \bU(n,\cdot)=0.
\end{align}
When the data is $\{\phi,\rig\}$-embedded, 
$\bU$ becomes the second fundamental form of $\phi(\N)$ w.r.t.\ the null normal $\nu=\phi_{\star}n$ (recall \eqref{normal}).
Inserting  $\ntwo=0$ and \eqref{s(n)=U(n,-)=0} in the contraction  of 
\eqref{nablaon} with $n^b$ entails
\begin{align}
\label{nablao_nn}\nablao_nn =0,
\end{align}
which together with \eqref{defY(n,.)andQ} and $\nu=\phi_{\star}n$
yields
\begin{equation}
\label{Qmeaning}\nabla_{\nu}\nu\stackbin{\eqref{nablaXYnablao}}=\phi_{\star}\big( \nablao_nn-\bY(n,n)n\big)\stackbin{\eqref{nablao_nn}}=\Q \phi_{\star}n=\Q \nu.
\end{equation}
Since $\nu$ is a null generator of $\Phi(\N)$,  
\eqref{Qmeaning} means that  $\Q$ corresponds (at the abstract level) to
the surface gravity of $\nu$. 
Under the action of the gauge group, the \textit{surface gravity} $\Q$ transforms as follows.
\begin{lemma}
\label{Gauge_trans_UFsone}
\textup{\cite{manzano2023constraint}} Let $\hypdata$ be null hypersurface data and consider 
gauge parameters $\{z,V\}$.
The  gauge behaviour of the scalar function $\Q$ defined in \eqref{defY(n,.)andQ} is
\begin{align}
\mathcal{G}_{\lp z,V\rp}\lp \Q \rp & = \frac{1}{z} \lp \Q - \frac{n(z)}{z}  \rp.
\label{Qprime}
\end{align}
\end{lemma}
We now state and prove a result that will be of particular relevance below, namely that by means of a gauge transformation one can always adapt the one-form $\ellc$ and the scalar $\elltwo$ to whatever pair $\{u\in\Fcal(\N),\bs{\vartheta}\in\Gamma(T^{\star}\N)\}$ one wishes, with the only restriction that
$\bs{\vartheta}(n)\neq0$ everywhere on $\N$. 
\begin{lemma}
\label{gaugefix}
Let $\metdata$ be null metric hypersurface data, $u$ a function on $\N$ and $\bm{\vartheta}\in\Gamma(T^{\star}\N)$ a covector satisfying $\bm{\vartheta} ( n) \neq 0$ everywhere. There exists a unique gauge transformation $\G_{(z,V)}$ satisfying
\begin{align}
\G_{(z,V)} (\ellc) = \bm{\vartheta}, \qquad \G_{(z,V)} (\elltwo) = u.    \label{trans}
\end{align}
Moreover, the gauge group element $\G_{(z,V)}$ is given by
\begin{align}
z = \bm{\vartheta} (n), \qquad V = \frac{1}{\bs{\vartheta}(n)} P(\bm{\vartheta}, \cdot)+ \frac{u - P(\bm{\vartheta},\bm{\vartheta} )}{2 \left ( \bm{\vartheta}(n) \right )^2}      n.\label{gaugeparameter}
\end{align}
\end{lemma}
\begin{remark}
  The condition $\bm{\vartheta}(n) \neq 0$ is necessary because if  $\bm{\vartheta}(n)$ vanishes at any point then 
  $\bs{\vartheta}$ can never correspond to $\ellc$ in any gauge, as
\begin{align*}
1 =   \big( \G_{(z,V)}(\ellc) \big)  \big(\G_{(z,V)} (n) \big)\vert_p = z^{-1} (\G_{(z,V)} (\ellc) )(n)\vert_p.
\end{align*}
which in particular states that $(\G_{(z,V)} (\ellc)) (n) \neq 0$ for all possible gauge parameters.
\end{remark}
\begin{proof}
We first assume that the gauge transformation exists and restrict its form up to a function yet to be determined. We then restrict to group elements of such a form and show that there exists one and only one of them that satisfies \eqref{trans}, namely \eqref{gaugeparameter}. This will prove both the existence and uniqueness claims of the lemma. For the first part we impose \eqref{trans}:
\begin{align}
 z \left ( \ellc + \gamma(V,\cdot) \right ) = \bm{\vartheta} , \qquad
z^2 \left ( \elltwo + 2 \ellc(V) + \gamma(V,V) \right )=u.
\label{trans2}
\end{align}
Contracting the first with $n$ gives $z = \bm{\vartheta}(n)$, so $\bm{w}\defi  \gamma(V,\cdot )  = \frac{1}{\bs{\vartheta}(n)} \bm{\vartheta} - \ellc$. 
Observe that $\bm{w}(n)=0$. Moreover, the vector $V - P(\bm{w}, \cdot)$ lies in the kernel of $\gamma$ because $\gamma_{ab} \left ( V^b - P^{bc} w_c \right )=    w_{a} - \left ( \delta^c_a - n^c \ell_a \right ) w_c =0$. Therefore, there exists $f\in\Fcal(\N)$ such that
$V^a  = P^{ab} w_b + f n^b = (\bs{\vartheta}(n))^{-1} P^{ab} \vartheta_b+ \left ( \elltwo  + f \right ) n^a$.
Thus, it suffices to restrict oneself to gauge parameters in the class
\begin{align}
\left \{ \Big ( z = \bs{\vartheta}(n) , V  = \frac{1}{\bs{\vartheta}(n)} P(\bm{\vartheta}, \cdot)  + q n\Big  ) ,  q \in \Fcal(\N) \right  \}.\label{defV}
\end{align}
We now start anew and prove that there is precisely one function $q$ such that the corresponding $(z,V)$ in  \eqref{defV} fulfills conditions \eqref{trans}. For $V$ as  in \eqref{defV} we get
\begin{align*}
\bm{\vartheta} (V)  & = \frac{1}{\bs{\vartheta}(n)} P(\bm{\vartheta}, \bm{\vartheta} )+ q \bs{\vartheta}(n), & 
\ellc (V) & = - \elltwo + q, \\
\gamma(V,\cdot)&= \frac{1}{\bs{\vartheta}(n)} \gamma( P(\bm{\vartheta}, \cdot),\cdot)
=\frac{1}{\bs{\vartheta}(n)}\bm{\vartheta}-\ellc , & 
\gamma(V,V)  & = \frac{1}{\bs{\vartheta}(n)} \bm{\vartheta}(V) - \ellc(V)  =\frac{P(\bm{\vartheta}, \bm{\vartheta})}{\bs{\vartheta}(n)^2} + \elltwo.
\end{align*}
The first condition in \eqref{trans2} is satisfied for all $q$. The second is satisfied if and only if
\begin{align*}
\bs{\vartheta}(n)^2 \left ( 2 q + \frac{P(\bm{\vartheta}, \bm{\vartheta})}{\bs{\vartheta}(n)^2}\right ) = u \qquad \Longleftrightarrow \qquad
q =  \frac{u - P(\bm{\vartheta}, \bm{\vartheta})}{2 \bs{\vartheta}(n)^2}.
\end{align*}
which ends the proof.
\end{proof}
In particular, Lemma \ref{gaugefix} (together with \eqref{gaugegamma&ell2}) means that two given null metric hypersurface data sets are related by a gauge transformation if and only if they both have the same data tensor $\gamma$. We prove this in the next corollary.
\begin{corollary}\label{cor:2datas1gamma}
Let $\D\defi \{\N,\gamma,\ellc,\elltwo\}$, $\uD\defi \{\N,\ugamma,\uellc,\uelltwo\}$ be two null metric hypersurface data. Then there is a gauge group element $\G_{(z,V)} \in \Fcal^{\star}(\N) \times \Gamma(T^{\star} \N)$ such that $\G_{(z,V)} ( \D ) = \uD$ if and only if $\gamma= \ugamma$. This gauge element is given by
\begin{align}
z = \uellc (n), \qquad V = \frac{1}{\uellc(n)} P(\uellc, \cdot)+ \frac{\uelltwo - P(\uellc,\uellc )}{2 \left ( \uellc(n) \right )^2}  n.\label{gaugeparametercor}
\end{align}
\end{corollary}
\begin{proof}
The necessity is obvious from the fact that $\gamma$ remains unchanged by a gauge transformation. Sufficiency is a direct application of Lemma \ref{gaugefix} to $\bm{\vartheta}= \uellc$ and $u = \uelltwo$. 
\end{proof} 
Lemma \ref{gaugefix} and Corollary \ref{cor:2datas1gamma} state that in the null case one can codify all the metric hypersurface data information exclusively in the tensor $\gamma$, and that $\ellc$ and $\elltwo$ are pure gauge. This fact will be key later in Section \ref{sec:matchingABSTRACT} 
when studying the matching of spacetimes with null boundaries.

We shall also need the decompositions of 
$\{\gamma,P\}$ in a basis $\{n,e_A\}$ of $\Gamma(T\N)$ and its corresponding dual.
\begin{lemma}\label{lem:Pdecom:dual}
\textup{\cite{manzano2023constraint}} Consider null metric hypersurface data $\metdata$. Let $\{n,e_A\}$ be a basis of $\Gamma(T\N)$ and $\{\bnormal,\bs{\theta}^A\}$ be its corresponding dual, i.e.\ 
$\bnormal(n)=1,\bnormal(e_A)=0,\bs{\theta}^A(n)=0,\bs{\theta}^A(e_B)=\delta^A_B$. Define the functions 
$\psi_{A}\defi \ellc(e_A)\in\Fcal(\N)$. Then, the tensor fields $\gamma$ and $P$ decompose as 
\begin{align}
\label{gammadecom:abstract}\gamma&=\mathfrak{h}_{AB}\bs{\theta}^A\otimes\bs{\theta}^B,\\
\label{Pdecom:proof} P&=P(\bs{\theta}^A,\bs{\theta}^B)e_A\otimes e_B+P(\bnormal,\bs{\theta}^A) (n\otimes e_A+e_A\otimes n)+P(\bnormal,\bnormal) n\otimes n\\
\label{Pdecom:abstract} &=\mathfrak{h}^{AB}e_A\otimes e_B-\mathfrak{h}^{AB}\psi_B\lp n\otimes e_A+ e_A\otimes n\rp- \lp \elltwo-\mathfrak{h}^{AB}\psi_A\psi_B \rp n\otimes n,
\end{align}
where $\mathfrak{h}_{AB}\defi \gamma(e_A,e_B)$ is a metric and $\mathfrak{h}^{AB}$ denotes its inverse. 
\end{lemma}
The concept of null thin shell arises naturally from Definitions \ref{def:thinshell}  and \ref{defNHD}. \textit{A thin shell}  
\textit{is said to be null if its metric part $\metdata$ defines null metric hypersurface data}. 
Moreover, as a corollary of Lemma \ref{lem:Pdecom:dual}, one can find 
a very simple form for the components of $\tau$.
\begin{corollary}\label{cor:eg-mom-tensor:abstract}
In the setup of Lemma \ref{lem:Pdecom:dual}, let $\{\N,\gamma,\ellc,\elltwo,\bY^{\pm},\rho^{\pm}_{\ell},\bs{J}^{\pm}, \fv\}$ be a null thin shell. Then, the components of the energy-momentum tensor $\tau$ in the basis $\{\bnormal,\bs{\theta}^A\}$ read
\begin{align}
\label{tau:prev:abs}\tau(\bnormal,\bnormal)=-\fv\mathfrak{h}^{AB} [\bY](e_A ,e_B),\qquad\tau(\bnormal,\bs{\theta}^A)=\fv\mathfrak{h}^{AB} [\bY](n, e_B) ,\qquad\tau(\bs{\theta}^A,\bs{\theta}^B)=-\fv\mathfrak{h}^{AB}[\bY](n,n).
\end{align}
\end{corollary}
\begin{proof}
Inserting the decomposition \eqref{Pdecom:abstract} into Definition \ref{def:energy-momentum_tensor} yields
\begin{align*}
\tau^{df}=-\fv\mathfrak{h}^{AB}\lp  [\bY](e_A,e_B)n^dn^f-[\bY](n,e_A)(n^de_B^f+e_B^dn^f)+[\bY](n,n)e_A^de_B^f\rp
\end{align*}
after a simple but somewhat long computation in which several terms cancel out. Contracting with $\{\bnormal,\bs{\theta}^A\}$ it is immediate to get \eqref{tau:prev:abs}.
\end{proof}
\begin{remark}\label{rem:def:eg:flux:pressure}
In the literature, the different components of the energy-momentum tensor 
of a thin shell $\{\N,\gamma,\ellc,\elltwo,\bY^{\pm}, \rho_{\ell}^{\pm},\bs{\textup{J}}^{\pm},\fv\}$ are interpreted physically as an energy density $\rho$, an energy-flux $j$ and a pressure $p$ (see e.g.\ \textup{\cite{poisson2004relativist}}). However, this is usually done in a context where $\{\N,\gamma,\ellc,\elltwo,\bY^{\pm}, \rho_{\ell}^{\pm},\bs{\textup{J}}^{\pm},\fv\}$ are embedded with
%
riggings $\rig^{\pm}$ that are 
null and 
%
%
orthogonal to the basis vectors $\{e_A\}$. In a completely general framework, we propose the following geometric definitions for the physical quantities $\{\rho,p,j\}$: 
\begin{equation}
\label{def:rho:jA:p:abs}\rho\defi -\fv\textup{tr}_P[\bY],\qquad p\defi -\fv[\bY](n,n),\qquad j\defi \fv\lp P\lp [\bY](n,\cdot),\cdot\rp-\fv \elltwo p n\rp .
\end{equation}
Definitions \eqref{def:rho:jA:p:abs} are justified because in the null case \eqref{tau}  can be written in terms of $\{\rho,p,j\}$ as 
%
\begin{equation}
\label{tau:rho:p:j}\tau=\rho n\otimes n+p\lp P+2\elltwo n\otimes n\rp+2 j\otimes_s n. 
\end{equation} 
For null shells, the vector field $j$ satisfies $\gamma(j,\cdot)=\fv[\bY](n,\cdot)+p\ellc$ and $\ellc(j)=0$, which makes the definitions \eqref{def:rho:jA:p:abs} consistent since the one-form $\bs{j}\defi\gamma(j,\cdot)$ verifies $\bs{j}(n)=0$. 
Moreover, a direct calculation based on \eqref{defY(n,.)andQ} and \eqref{Qprime} proves the following gauge behaviour for the pressure $p$: 
\begin{equation}
\label{gauge:pressure} \G_{(z,V)}(p)=\frac{p}{\vert z\vert}.
\end{equation}
Whenever $\elltwo=0$ and $\psi_A=0$, it is straightforward to check that \eqref{def:rho:jA:p:abs} becomes
\begin{equation}
\label{poisson:rho:p:j:54368652}\rho=-\fv \mathfrak{h}^{AB}[\bY](e_A,e_B),\quad p=-\fv [\bY](n,n),\quad j=\fv \mathfrak{h}^{AB}[\bY](n,e_B)e_A, 
\end{equation}
after using \eqref{gammadecom:abstract} and  \eqref{Pdecom:abstract}. This allows one to recover the standard definitions for $\{\rho,p,j\}$ introduced e.g.\ in \textup{\cite[Eq. (3.99)]{poisson2004relativist}}. Expressions \eqref{poisson:rho:p:j:54368652} coincide with the definitions proposed in \textup{\cite{poisson2004relativist}} whenever $\fv=-1$
which, as mentioned in Remark \ref{rem:sign:mathfrak(v)}, corresponds to the rigging $\rig^-$ pointint inwards.
\end{remark}
We conclude this subsection by recalling several aspects on the geometry of transverse submanifolds embedded in null metric hypersurface data sets. We again refer to \cite{manzano2023constraint} for proofs. Given null metric hypersurface data $\metdata$, a \textit{transverse submanifold $S$ is a codimension one embedded submanifold of $\N$ to which $n$ is everywhere transverse}. 
%
%
%
Letting $\psi:S\longhookrightarrow \mathcal{N}$ be the embedding 
of $S$ in $\N$ we define 
$h\defi \psi^{\star}\gamma$. It is a fact \cite{manzano2023constraint} that $h$ is
a metric on $S$ and we denote by
$\nabh$
its Levi-Civita covariant derivative. When it is clear from the context
we identify vectors and scalars on $S$ with their counterpars on
$\psi(S)$.  
For any $\mathfrak{p}$-covariant tensor $\bs{T}$ along $\Psi(S)$ and given a basis $\{v_A\}$ of $\Gamma(TS)$, we define $\bs{T}_{\parallel}\defi \psi^{\star}\bs{T}$ and write $T_{A_1\dots A_{\mathfrak{p}}}\defi \bs{T}_{\parallel}(v_{A_1},\dots, v_{A_{\mathfrak{p}}})$ (without the parallel symbol). Capital Latin indices are raised with $h_{IJ}$ and its inverse $h^{IJ}$. 
With the definition  $\elltwo_{\parallel} := h^{IJ} \ell_I \ell_J$, the pull-back to $S$ of the $\nablao$ derivative of any $\mathfrak{p}$-covariant tensor field $\mathcal{T}$ along $\psi(S)$ takes the following explicit form \cite[Lem.\ 3.15]{manzano2023constraint}: 
\begin{align}
\nn &v_{A_1}^{a_1}\dots v_{A_{\mathfrak{p}}}^{a_{\mathfrak{p}}}v_B^b\nablao_b\mathcal{T}_{a_1\cdots a_{\mathfrak{p}}} = \nabla_{B}^h\mathcal{T}_{A_1\cdots A_{\mathfrak{p}}}-\sum_{\mathfrak{i}=1}^{\mathfrak{p}}\ellp^J\mathcal{T}_{A_1\cdots A_{\mathfrak{i}-1}JA_{\mathfrak{i}+1}\cdots A_{\mathfrak{p}}} \U_{A_{\mathfrak{i}}B}\\
\label{covderpcovtensoronS}  &-\sum_{\mathfrak{i}=1}^{\mathfrak{p}}\mathcal{T}_{a_1\cdots a_{\mathfrak{p}}}v_{A_1}^{a_1}\dots v_{A_{\mathfrak{i}-1}}^{a_{\mathfrak{i}-1}}n^{a_{\mathfrak{i}}}v_{A_{\mathfrak{i}+1}}^{a_{\mathfrak{i}+1}}\dots v_{A_{\mathfrak{p}}}^{a_{\mathfrak{p}}}\lp \nabla^h_{(A_{\mathfrak{i}}}\ellp_{B)}+(\elltwo-\elltwo_{\parallel})\U_{A_{\mathfrak{i}}B} \rp.
\end{align}

\subsection{Matching of spacetimes and junction conditions}\label{sec:Matching:Paper:Big:Sec}
From now on 
we focus on the problem of matching two spacetimes with boundary. 
In this section we recall known results, first for boundaries of any causality and secondly in the null case. 

Consider two spacetimes $(\Mpm,g^{\pm})$ with boundaries $\nullhyp^{\pm}$ 
of any causal character. It is well-known (see e.g.\ \cite{darmois1927memorial}, \cite{obrien1952jump}, \cite{le1955theories},  \cite{israel1966singular}, \cite{bonnor1981junction}, \cite{clarke1987junction}, \cite{barrabes1991thin},  \cite{mars1993geometry}) that 
the matching of $(\Mpm,g^{\pm})$ across $\nullhyp^{\pm}$ is possible if and only if the so-called \textit{junction conditions} or \textit{matching conditions} are satisfied. In the language of the formalism of hypersurface data, the matching 
requires \cite{mars2013constraint} that 
there exist metric hypersurface data $\metdata$ that can be embedded in both spacetimes $(\Mpm,g^{\pm})$ with embeddings $\phi^{\pm}$ such that $\phi^{\pm}(\N)=\nullhyp^{\pm}$ and riggings $\rig^{\pm}$, i.e.\ there must exist two pairs $\{\phi^{\pm},\rig^{\pm}\}$ 
satisfying 
\begin{align}
\label{junctcondAbs2} \gamma =(\phi^{\pm})^{\star}(g^{\pm}),\qquad
\ellc =(\phi^{\pm})^{\star}(g^{\pm}(\rig^{\pm},\cdot)),\qquad \elltwo =(\phi^{\pm})^{\star}(g^{\pm}(\rig^{\pm},\rig^{\pm})).
\end{align}
In addition, the riggings $\rig^{\pm}$ must fulfil an orientation condition (see item $(ii)$ below). 
In these circumstances, it is always possible to select one of the embeddings 
freely
by adapting 
$\N$ to one of the boundaries. In the following we shall make use of this freedom by fixing $\phi^-$ at our convenience. This entails no loss of generality. Note that making the choice in the minus side is also of no relevance, as one can always switch the names of the spacetimes to be matched.

When the junction conditions are satisfied, the geometry of the shell \cite{mars2013constraint} is determined by the jump of the transverse tensors $\bY^{\pm}$ defined as
\begin{equation}
\label{eqA21}\bY^{\pm}\stackrel[\eqref{YtensorEmbDef}]{\mbox{\tiny {\textbf{def}}}}{\,\, = \,\, } \dfrac{1}{2}(\phi^{\pm})^{\star}\lp \pounds_{\rig^{\pm}}g^{\pm}\rp,\qquad\text{namely}\qquad \lc \bY\rc\defi \bY^+-\bY^-.
\end{equation}
In the literature, however, the matching conditions are not normally presented in terms of a hypersurface data set. Instead, they 
 are usually formulated as follows (see e.g.\ \cite{mars1993geometry}). 
\begin{junction*}
The matching of 
$(\Mpm,g^{\pm})$ across $\nullhyp^{\pm}$ can be performed if and only if 
%
\begin{itemize}
\item[$(i)$] There exist two riggings $\rig^{\pm}$ along $\nullhyp^{\pm}$ and a diffeomorphism ${\Phi}:\N^-\longrightarrow\N^+$ such that 
\begin{align}
\label{junctcondST} {\Phi}^{\star}(g^+)&=g^-,\qquad
 {\Phi}^{\star}(g^+(\rig^+,\cdot))=g^-(\rig^-,\cdot),\qquad  {\Phi}^{\star}(g^+(\rig^+,\rig^+))=g^-(\rig^-,\rig^-).
\end{align}
\item[$(ii)$] One rigging must point inwards w.r.t.\ its boundary and the other outwards. 
\end{itemize}
\end{junction*}
For the rest of the paper, two riggings $\rig^{\pm}$ satisfying $(i)$-$(ii)$ 
for a diffeomorphism $\Phi$ will be called \textit{matching riggings}. The diffeomorphism $\Phi$  will be referred to as \textit{matching map}.

If \eqref{junctcondST} holds for two riggings $\rig^{\pm}$ then, for any other choice of rigging on one of the sides, 
\eqref{junctcondST} 
is 
fulfilled as well\footnote{For any other rigging $\rig'^-=z(\rig^-+V)$ with $\{ z\in\Fcal^{\star}(\nullhyp^-), V\in\Gamma(T\nullhyp^-)\}$, the rigging $\rig'^+\defi \widehat{z}(\rig^++\widehat{V})$ with $\{\widehat{z}\defi ({\Phi}^{-1})^{\star}z,\widehat{V}\defi {\Phi}_{\star}V\}$ also verifies \eqref{junctcondST}. The same logic applies if one changes the rigging on the plus side.}
(although different choices of rigging on one side will correspond to different riggings on the other side). We shall make use of this freedom to fix $\rig^-$ at 
will, again  
with no loss of generality.

As proven in Lemmas 2 and 3 of \cite{mars2007lorentzian}, given a rigging on one side (say $\rig^-$) and a diffeomorphism ${\Phi}:\nullhyp^-\longrightarrow\nullhyp^+$ satisfying $\Phi^{\star}g^+=g^-$, at non-null points the second and third equations of \eqref{junctcondST} yield either no solution for $\rig^+$ (hence the matching is not possible) or two solutions for $\rig^+$ with opposite orientation. At null points, on the other hand, if there exists a solution $\rig^+$ then it is unique. This means that at non-null points one can always make a suitable choice of rigging $\rig^+$ so that the junction condition $(ii)$ is fulfilled, and hence one only needs to care about \eqref{junctcondST}. In the null case, however, this is not so.  
It 
can 
happen that there exists a solution $\rig^+$ of \eqref{junctcondST} but with unsuitable orientation, and then the matching cannot be performed. Thus, at null points conditions \eqref{junctcondST} are necessary but not sufficient to guarantee that the matching is feasible \cite{mars2007lorentzian}.

When the matching 
is possible, the corresponding matching map ${\Phi}$ is the key object upon which the whole matching depends. This is so because once the point-to-point identification of the boundaries $\nullhyp^{\pm}$ (ruled by $\Phi$) is known, one matching rigging can be selected at will (as we have seen) and the other is the unique solution that arises from enforcing both \eqref{junctcondST} and $(ii)$. 
All the information about the matching is therefore codified by $\Phi$, or equivalently by the 
embedding $\phi^+$
(cf. \eqref{junctcondAbs2}).

We now concentrate on the case when the boundaries $\nullhyp^{\pm}$ are null. From a spacetime viewpoint, this problem was addressed in \cite{manzano2021null} (see also \cite{manzano2022general}, where 
the matching across Killing horizons of order zero was studied). 
In the remainder of the section we summarize the main results of \cite{manzano2021null}.

Consider two $(\n+1)$-dimensional spacetimes $(\Mpm,g^{\pm})$ with null boundaries $\nullhyp^{\pm}$ that can be foliated by a family of diffeomorphic spacelike cross-sections. Assume further that one of the boundaries lies in the future of its corresponding spacetime while the other lies in its spacetime past. This entails no loss of generality, as explained in \cite{manzano2021null}.  
%
%
We 
construct foliation functions $v_{\pm}\in\Fcal(\nullhyp^{\pm})$ and basis $\{L^{\pm},k^{\pm},v_I^{\pm}\}$ of $\Gamma(T\Mpm)\vert_{\nullhyp^{\pm}}$ according to \eqref{basis}. The surface gravities of $k^{\pm}$ are $\ke_{k^{\pm}}^{\pm}$. As in Section \ref{sec:Geom:null:hyp:paper:matching}, the leaves of the foliations are denoted by $\{S^{\pm}_{v_{\pm}}\}$, while their corresponding induced metrics are $h^{\pm}$. 
We also let $\btsff_{\pm}^k$ be the second fundamental forms of $\nullhyp^{\pm}$ w.r.t.\ $k^{\pm}$,  
and introduce the tensors 
$\bs{\Theta}_{\pm}^L$, 
$\bs{\sigma}^{\pm}_L$ on the leaves $\{S^{\pm}_{v_{\pm}}\}$ (cf.\ \eqref{somedefs}). The scalar functions $\{\mu^{\pm}_a\}\subset\Fcal(\nullhyp^{\pm})$ are defined by \eqref{eqA30} w.r.t.\ the basis $\{L^{\pm},k^{\pm},v_I^{\pm}\}$.

As we have seen, in order to perform a matching we need to embed a single metric hypersurface data set in both spacetimes. We codify the already described freedom in the choice of $\{\phi^-,\rig^-\}$ as follows. 
We first consider an abstract null hypersurface $\N$ and define coordinates $\{y^1=\lambda,y^A\}$ therein. Then, we construct null embedded metric hypersurface data by enforcing that $(a)$ the push-forwards $\{e_a^-\defi \phi^-_{\star}(\cp_{y^a})\}$ 
coincide with the basis vectors $\{k^-,v^-_I\}$ (since $\{k^-,v^-_I\}$ are chosen at will, with this procedure we ensure that 
$\phi^-$
is built at our convenience) and $(b)$ that the rigging $\rig^-$ coincides with the basis vector $L^-$. 
This amounts to impose
\begin{equation}
\label{eiyl}
e^-_1=k^-,\qquad e^-_I=v^-_I,\qquad \rig^-=L^-. 
\end{equation}
Thus, $\rig^-$ is a null past rigging (recall \eqref{basis}) and 
$\lambda$ is a coordinate along the degenerate direction of $\N$. In fact, the subsets $\{\lambda=\text{const.}\}\subset\N$ are all diffeomorphic \cite{manzano2021null} and define a (spacelike) foliation of 
$\N$. 

For the matching of $(\Mpm,g^{\pm})$ to be possible, there must exist another pair $\{\phi^+,\rig^+\}$ so that \eqref{junctcondAbs2} hold (and the orientations of $\rig^{\pm}$ are suitable). 
In that case, 
we can build another basis $\{e_a^+\defi \phi^+_{\star}(\cp_{y^a})\}$ of $\Gamma(T\nullhyp^+)$ and then determining the matching requires that we find the explicit form of the vectors 
$\{e_a^{+}\}$ (which fully codify 
$\phi^+$). In 
the basis $\{k^+,v^+_I\}$ of $\Gamma(T\nullhyp^+)$, these vectors 
decompose as \cite{manzano2021null} 
\begin{equation}
\label{eqA4}
e^+_1=\mathfrak{f} k^+,\qquad e^+_I=a_Ik^++b_I^Jv^+_J,
\end{equation}
where $\mathfrak{f},a_I,b_I^J\in\Fcal(\nullhyp^+)$ are given by
\begin{align}
\label{bIJ:coef} \mathfrak{f}&=\dfrac{\cp H(\lambda,y^A)}{\cp\lambda},\qquad a_I=\dfrac{\cp H(\lambda,y^A)}{\cp y^I},\qquad b_I^K=\dfrac{\cp h^K( y^A)}{\cp y^I}
\end{align}
in terms of a set of functions $\{H(\lambda,y^B),h^A(y^B)\}$ on $\N$. The functions $\{H,h^A\}$ encode all the matching information  and hence they determine $\phi^+$. In fact, given coordinates $\{ v_+, u^I \}$ on $\nullhyp^+$ such that $v^+_I = \partial_{u^I}$ (i.e.\ $\{u^I_+\}$ are constant along the null generators), the embedding $\phi^+$ is such that 
\cite{manzano2021null}
\begin{align}
\phi^+(\lambda, y^I) =\left (v_+= H(\lambda,y^I), u^I = h^I(y^J) \right ). 
\label{embedPhi+}                                                   
\end{align}
The function $H\lp\lambda,y^A\rp$ is named \textit{step function} because it measures a kind of jump along the null direction when crossing the matching hypersurface. It must satisfy the condition $\cp_{\lambda}H>0$ \cite{manzano2021null}. 
The explicit form of the matching rigging $\rig^+$ was computed in \cite[Cor. 1]{manzano2021null} and reads 
%
\begin{equation}
\label{uniqxi}
\hspace{-0.3cm}\rig^+=\dfrac{\mu_1^-}{\cp_{\lambda}H}\lp \frac{1}{\mu_1^+}
  L^+-  h_+^{AB}\lp (b^{-1})^I_A \lp \cp_{y^I}H
  - \frac{1}{\mu_1^-} (\cp_{\lambda} H) \mu^-_I \rp + \frac{1}{\mu_1^+}
  \mu^+_A \rp Z_B\rp,
\end{equation}
where $(b^{-1})_{I}^J\defi \cp_{h^I}y^J$ and $Z_B\defi   \frac{1}{2} \lp (b^{-1})_B^J \lp
  \cp_{y^J}H -  \frac{1}{\mu_1^{-}} (\cp_{\lambda}H) \mu_J^- \rp -\frac{1}{\mu^+}
  \mu^+_B   \rp k^++v^+_B$.

The solvability of the first junction condition in \eqref{junctcondAbs2} constitutes the core problem for the existence of a matching. In terms of the metrics $h^{\pm}$,  
it can be rewritten as  
\begin{align}
\label{junct1}h^-_{IJ}\vert_p &= b_I^Lb_J^Kh^+_{LK}\vert_{{\Phi}( p)}\quad\forall p\in\nullhyp^-.
\end{align}
Equation \eqref{junct1} is an isometry condition between each submanifold 
$\{v_-=\textup{const.}\}\subset\nullhyp^-$
%
%
and its corresponding image on $\nullhyp^+$. On the other hand, the identification of $\{e^{\pm}_1\}$ requires the existence of a diffeomorphism $\Psi$ (ruled by the coefficients $b^{A}_{B}$ fulfilling \eqref{junct1}) between the set of null generators on both sides.
Moreover, combining 
\eqref{eqA28:plus:on:N},  
\eqref{eiyl}-\eqref{bIJ:coef}, \eqref{junct1} and $\{e_1^{\pm}\defi\phi_{\star}^{\pm}(\cp_{y^a})\}$ yields \cite{manzano2021null}
\begin{equation}
\label{detrmfA}
{\btsff}_-^{k}( v^-_I,v^-_J)=(\cp_{\lambda} H) b_I^Ab_J^B{\btsff}_+^{k}( v^+_A,v^+_B).
\end{equation} 
Thus, for each possible choice of $\Psi$    
(i.e.\ of $\{b^{A}_{B}\}$), 
\eqref{detrmfA} determines a unique value for $\partial_{\lambda} H$ unless the two second fundamental forms vanish simultaneously. In the latter case, the step function $H$ cannot be restricted. \tcr{Consequently, when $\nullhyp^{\pm}$ are totally geodesic, 
if a single matching of $(\Mpm,g^{\pm})$ can be performed then an infinite number of matchings (one for each possible step function $H$) are feasible \cite{manzano2021null}.} 

When the matching is possible, the matter-energy content of the shell is given by the next proposition.  
\begin{proposition}
\label{prop6} 
\textup{\cite{manzano2021null}} 
Assume that the matching of $(\Mpm,g^{\pm})$ across $\nullhyp^{\pm}$ is possible and that it is determined by the functions $\{H(\lambda,y^A),h^B(y^A)\}$. 
Let $h_{IJ}$ be the induced metric on the leaves $\{\lambda=\textup{const.}\}\subset\N$,  $h^{IJ}$ its inverse tensor and $\nabh$ its Levi-Civita covariant derivative.
Define the vector fields $\{W_A\}$, the scalars $\{\ov{\mu}^+_A\}$, the covector $q\in\Gamma(T^{\star}\N)$ and the vector field $X=X^a\cp_{y^a}\in\Gamma(T\N)$ by 
\begin{align}
\label{def:ovmu:and:W_A} \ov{\mu}_I^+\defi &\spc b_I^B\mu^+_B,\qquad W_I\defi b_I^Bv_B^+,\qquad q_I\defi -\mu_1^+\nabh_I H-\ov{\mu}_I^+,\\
\label{X1} X^1\defi &\spc \dfrac{ h^{IJ}}{2\mu_1^+\cp_{\lambda}H}\lp q_I+\dfrac{\mu_1^+\mu^-_I}{\mu_1^-}\cp_{\lambda}H\rp\lp q_J-\dfrac{\mu_1^+\mu^-_J}{\mu_1^-}\cp_{\lambda}H\rp,\qquad  X^A\defi  h^{IA}\lp q_I+\dfrac{\mu_1^+\mu^-_I}{\mu_1^-}\cp_{\lambda}H\rp. 
\end{align}
Then, the components of the tensor $[\bY]\defi \bY^{+}-\bY^-$ are
\begin{align}
\label{Y11}&[\bY](\cp_{\lambda},\cp_{\lambda})= -\mu_1^{-}\lp \ke_{k^+}^+\cp_{\lambda}H-\ke^-_{k^-} + \dfrac{\cp_{\lambda}\cp_{\lambda}H}{\cp_{\lambda}H}\rp,\\
\label{Y1J}  &[\bY](\cp_{\lambda},\cp_{y^J})= -\mu_1^-\Bigg( \ke_{k^+}^+\nabla_{J}^hH-\big(\bs{\sigma}_{L}^+(  W_J)-\bs{\sigma}^-_{L}( v^-_J)\big)+\dfrac{\cp_{\lambda}\cp_{y^J}H}{\cp_{\lambda}H}+\dfrac{X^L{\btsff}_-^{k}( v_J^-,v_L^-)}{\mu_1^+\cp_{\lambda}H}\Bigg),\\
\nonumber &[\bY](\cp_{y^I},\cp_{y^J})=-\mu_1^-\Bigg(\dfrac{\ke_{k^+}^+\nabla_{I}^hH \textup{ }\nabla_{J}^hH}{\cp_{\lambda}H}-\dfrac{\nabla_{{\lp I\rd}}^hH\textup{ }\cp_{\lambda}\ov{\mu}^+_{\ld J\rp}}{\mu_1^+\lp\cp_{\lambda}H\rp^2}-\dfrac{2\nabla_{{\lp I\rd}}^hH\textup{ }\bs{\sigma}_{L}^+( W_{\ld J\rp})}{\cp_{\lambda}H} \\
\label{YIJ}&-\bigg(\dfrac{\bs{\Theta}^{L}_+( W_{\lp I\rd},W_{\ld J\rp})}{\mu_1^+\cp_{\lambda}H}-\frac{\bs{\Theta}^{L}_-( v_{\lp I \rd}^-,v_{\ld J \rp}^-)}{\mu_1^-}\bigg)- \dfrac{X^1{\btsff}_-^{k}( v_I^-,v_J^-)}{\mu_1^+\cp_{\lambda}H}+\dfrac{\nabla_{I}^h\nabla_{J}^hH}{\cp_{\lambda}H}+\dfrac{\nabla_{\lp I \rd}^h\ov{\mu}^+_{\ld J \rp}}{\mu_1^+\cp_{\lambda}H}-\dfrac{\nabla_{\lp I \rd}^h\mu^-_{\ld J \rp}}{\mu_1^-}\Bigg),
\end{align}
while the energy-momentum tensor of the shell is given by (the sign $\fv$ is given by Definition \ref{def:thinshell})
\begin{align*}
\tau(d\lambda,d\lambda)=-\fv\frac{ h^{IJ}[\bY](\cp_{y^I},\cp_{y^J})}{(\mu_1^-)^2},\quad \tau(d\lambda,dy^I)=\fv\frac{ h^{IJ}[\bY](\cp_{\lambda},\cp_{y^J})}{(\mu_1^-)^2},\quad \tau(dy^I,dy^J)=-\fv \frac{ h^{IJ}[\bY](\cp_{\lambda},\cp_{\lambda})}{(\mu_1^-)^2}.
\end{align*} 
\end{proposition}

\section{Abstract formulation of the matching problem}\label{sec:matchingABSTRACT}
%
%
%
In the previous section, we have summarized the main aspects of the matching of two general spacetimes with null boundaries that admit a foliation by diffeomorphic spacelike sections. The matching conditions have been formulated from a spacetime viewpoint, and we have recalled the geometrical objects upon which the matching depends (namely the step function $H$ and the diffeomorphism $\Psi$). We have also recollected the explicit expressions for the gravitational and matter-energy content of the resulting shells 
(Proposition \ref{prop6}).  

The results we have just summarized leave (at least) two interesting problems unaddressed. The first one is
%
whether one can obtain analogous results without the topological assumptions on the boundaries and the second is whether there is a way of formulating the matching problem in a fully abstract manner, namely without making any reference to the actual spacetimes to be matched.
As already explained in the \hyperref[sec:Intro:Paper]{Introduction}, addressing these problems is the key object of this paper.
%
%
%


Let us start with the abstract formulation of the junction conditions. For that purpose, we first consider that the boundaries $\nullhyp^{\pm}$ of the spacetimes $(\Mpm,g^{\pm})$ to be matched have \textit{any topology} and \textit{any causal character}. Since $\nullhyp^-$ is embedded, there exists an abstract manifold $\N$ and an embedding $\iota^-:\N\longhookrightarrow\Ml$ such that $\iota^-(\N)=\nullhyp^-$. From the embedding $\iota^-$, one can construct an infinite number of embeddings simply by applying additional diffeomorphisms within $\N$. To elude this unavoidable redundancy, we henceforth let $\iota^-$ be one specific choice among all possible. As discussed before, two spacetimes $(\Mpm,g^{\pm})$ can be matched if there exists a pair of embeddings $\phi^{\pm}:\N\longhookrightarrow\Mpm$ related to a matching map $\Phi$ by $\phi^+=\Phi\circ\phi^-$. Moreover, the embedding and the rigging on one of the sides (say the minus side) can always be chosen freely. Suppose we enforce $\phi^-=\iota^-$ and take a specific rigging $\rig^-$. Then we can 
build embedded hypersurface data $\D\defi \{\N,\gamma,\ellc,\elltwo,\bY\}$ by requiring \eqref{emhd}-\eqref{YtensorEmbDef}, i.e.\ by defining
\begin{equation}
\label{middleeq14}
\gamma\defi (\iota^-)^*(g^-),\quad\ellc\defi (\iota^-)^*(g^-(\rig^-,\cdot)),\quad\elltwo\defi (\iota^-)^*(g^-(\rig^-,\rig^-)),\quad 
\bY^-\defi \frac{1}{2}(\iota^-)^*(\pounds_{\rig^-}g^-).
\end{equation}
Thus, all the information about the matching is encoded in $\phi^+$ and the junction conditions \tcr{are \eqref{junctcondAbs2}}. 
These conditions, 
although  of a more abstract nature than \eqref{junctcondST}, still codify the matching information in the pair $\{\phi^+,\rig^+\}$, which is not of abstract nature. In order to determine the matching in terms of objects defined at the abstract level, we must take one step further. The following theorem, based on the existence of a diffeomorphism $\nfi$ of the abstract manifold $\N$ onto itself, sets up the corresponding construction.
%
%
%
\begin{theorem}\label{thm:Best:THM:JC}
Consider two hypersurface data 
$\D\defi \{\N,\gamma,\ellc,\elltwo,\bY^-\}$, 
$\hatD\defi \{\N,\hatg,\hatellc,\hatelltwo,\hatbY^+\}$
embedded in two spacetimes $(\M^-,g^-)$, $(\M^+,g^+)$ with embeddings $\iota^-$, $\iota^+$ and riggings $L^-$, $L^+$ respectively. Assume that $\iota^{\pm}(\N)\defi \nullhyp^{\pm}$ are boundaries of $(\Mpm,g^{\pm})$ and
\tcr{let $\epsilon^{+} = +1$ (resp. $\epsilon^+ = -1$) if $L^{+}$ points outwards (resp. inwards) from $\M^+$. Define $\epsilon^-$ in the same way (i.e.\ $\epsilon^- = +1$ if $L^-$ points outwards, $\epsilon^- = -1$ if inwards).}
Then, the matching of $(\Mpm,g^{\pm})$ across $\nullhyp^{\pm}$ is possible if and only if 
\begin{itemize}
\item[(i)] There exist a gauge group element $\G_{(z,V)}$ 
and a diffeomorphism $\nfi$ of $\N$ onto itself such that
\begin{equation}
\label{JC:HD}\G_{(z ,V )}(\nfi^{\star}\hatg)=\gamma,\qquad\G_{(z ,V )}(\nfi^{\star}\hatellc)=\ellc,\qquad\G_{(z ,V )}(\nfi^{\star}\hatelltwo)=\elltwo;
\end{equation}
\item[(ii)] $\textup{sign}(z)=-\textup{sign}(\epsilon^+)\textup{sign}(\epsilon^-)$. 
%
\end{itemize}
\end{theorem}
\begin{proof}
The fact that $\D$, $\hatD$ are embedded on $(\Mpm,g^{\pm})$ respectively means that
\begin{align}
\label{def:D,hatD:1} \gamma &\defi (\iota^-)^{\star}(g^-), & \ellc &\defi (\iota^-)^{\star}(g^-(L^-,\cdot)), & \elltwo &\defi (\iota^-)^{\star}(g^-(L^-,L^-)), & \bY^-&\defi \frac{1}{2}(\iota^-)^{\star}(\pounds_{L^-}g^-),\\
\label{def:D,hatD:2} \hatg&\defi (\iota^+)^{\star}(g^+), &\hatellc &\defi (\iota^+)^{\star}(g^+(L^+,\cdot)), &  \hatelltwo &\defi (\iota^+)^{\star}(g^+(L^+,L^+)), & \hatbY^+&\defi \frac{1}{2}(\iota^+)^{\star}(\pounds_{L^+}g^+).
\end{align}
Since the spacetimes $(\Mpm,g^{\pm})$, the embeddings $\iota^{\pm}$ and the riggings $L^{\pm}$ are all given, the tensor fields in \eqref{def:D,hatD:1}-\eqref{def:D,hatD:2} are known. To prove the first part of the theorem, we start by assuming $(i)$-$(ii)$. Thus, there exist a pair $\{z\in\Fcal^{\star}(\N),V\in\Gamma(T\N)\}$ and a diffeomorphism $\nfi:\N\longrightarrow\N$ so that \eqref{JC:HD} holds. These conditions can be rewritten as (recall  \eqref{gaugegamma&ell2}, 
$\mathcal{G}^{-1}_{(z,V)} = \mathcal{G}_{(z^{-1},-zV)}$)
\begin{align}
\label{JC:HD:inv:elem1}\nfi^{\star}\hatg&=\mathcal{G}^{-1}_{(z,V)}(\gamma)=\mathcal{G}_{(z^{-1},-zV)}(\gamma)=\gamma,\\
\label{JC:HD:inv:elem2}\nfi^{\star}\hatellc&=\mathcal{G}^{-1}_{(z,V)}(\ellc)=\mathcal{G}_{(z^{-1},-zV)}(\ellc)=\frac{\ellc}{z}- \gamma \lp V,\cdot\rp,\\
\label{JC:HD:inv:elem3}\nfi^{\star}\hatelltwo&=\mathcal{G}^{-1}_{(z,V)}(\elltwo)=\mathcal{G}_{(z^{-1},-zV)}(\elltwo)=\frac{\ell^{(2)}}{z^2}-\frac{2\ellc\lp V\rp}{z}+ \gamma \lp V,V\rp.
\end{align} 
Let us define the map $\phi^+\defi \iota^+\circ\nfi$, the vector field $\hatV \defi \iota^+_{\star}(\nfi_{\star} V)$, the function $z'\in\Fcal^{\star}(\nullhyp^+)$ given by $\nfi^{\star}((\iota^+)^{\star}z')\defi z$ and the rigging $\rig^+\defi z'(L^++\hatV)$ along $\nullhyp^+$. By definition of $z'$, it holds that $\text{sign}(z)=\text{sign}(z')$.
%
%
%
%
%
%
On the other hand, combining \eqref{JC:HD:inv:elem1}-\eqref{JC:HD:inv:elem3} with the fact that $\hatD$ is embedded with embedding $\iota^+$ and rigging $L^+$, it follows
\begin{align}
\label{argument:1} \gamma=&\spc\nfi^{\star}\hatg=\nfi^{\star}((\iota^+)^{\star}(g^+))=(\phi^+)^{\star}(g^+),\\
\label{argument:2} \ellc=&\spc 
%
%
z \lp \nfi^{\star}\hatellc+(\nfi^{\star}\hatg)(V ,\cdot) \rp= z \nfi^{\star}\lp (\iota^+)^{\star}\lp g^+(L^+,\cdot)+g^+(\hatV ,\cdot)\rp\rp=(\phi^+)^{\star}(g^+(\rig^+,\cdot)),\\
\nn \elltwo  
%
%
=&\spc z^2\lp \nfi^{\star}\hatelltwo+2(\nfi^{\star}\hatellc)(V )+(\nfi^{\star}\hatg)(V ,V )\rp= z^2\nfi^{\star}\lp(\iota^+)^{\star}\lp g^+(L^+,L^+)+2g^+(L^+,\hatV )+g^+(\hatV ,\hatV )\rp\rp\\
\label{argument:3} =&\spc (\phi^+)^{\star}(g^+(\rig^+,\rig^+)).
\end{align}
The data $\D$ is therefore embedded in $(\M^+,g^+)$ with embedding $\phi^+$ and rigging $\rig^+$. Thus, conditions  \eqref{junctcondAbs2} are satisfied for $\phi^-=\iota^-$, $\phi^+=\iota^+\circ\nfi$ and for the riggings $\rig^-=L^-$, $\rig^+$. 
Moreover, combining $(ii)$ (which holds by assumption), the definition of $\rig^+$ and $\text{sign}(z')=\text{sign}(z)$, it follows 
\begin{equation}
\label{rig+:orient:signs:etc} \rig^+=-\text{sign}(\epsilon^+)\text{sign}(\epsilon^-)\vert z'\vert\lp L^++V'\rp.
\end{equation}
It is straightforward to check that \eqref{rig+:orient:signs:etc} implies that whenever $L^-$ points inwards (resp. outwards) then $\rig^+$ points outwards (resp. inwards) irrespectively of the orientation of $L^+$. 
Thus, $\D$ is embedded in $(\Mpm,g^{\pm})$ and $L^-$, $\rig^+$ are such that one points inwards and the other outwards, which means that the matching of $(\Mpm,g^{\pm})$ is possible. 

To prove the converse, we assume that the matching is possible for two pairs $\{\phi^{\pm},\rig^{\pm}\}$. 
We have already discussed the flexibility of selecting at will the embedding and the rigging on one side (say the minus side). Let us therefore set $\phi^-=\iota^-$, $\rig^-=L^-$. 
Since both $L^+$ and $\rig^+$ are riggings along $\nullhyp^+$, there exists a pair $\{z'\in\Fcal^{\star}(\nullhyp^+), \hatV\in\Gamma(T\nullhyp^+)\}$ such that $\rig^+=z'(L^++\hatV)$. Moreover, one can define a diffeomorphism $\nfi:\N\longrightarrow\N$ by $\phi^+\defi \iota^+\circ\nfi$. But then one can follow the arguments of \eqref{argument:1}-\eqref{argument:3} backwards and prove \eqref{JC:HD} for a function $z\in\Fcal^{\star}(\N)$ defined by $z\defi \nfi^{\star}((\iota^+)^{\star}z')$. As before,  $\sgn(z)=\sgn(z')$ so both $\rig^+=z'(L^++\hatV)$ and $z'L^+=\sgn(z)\vert z'\vert L^+$ have the same orientation (because $\hatV$ is tangent to $\nullhyp^+$). By assumption the matching is possible, hence $L^-$, $\rig^+$ are such that one points inwards and the other outwards. If $L^-$ points inwards (resp. outwards) then $\sgn(z) L^+$ must point outwards (resp. inwards), so $\sgn(z)=\sgn(\epsilon^+)$ ($\sgn(z)=-\sgn(\epsilon^+)$) is forced. This means that $(i)$-$(ii)$ are both fulfilled.
\end{proof}
\begin{remark}
Theorem \ref{thm:Best:THM:JC} does not impose any conditions on the topology of the abstract manifold $\N$, except for the very mild one that hypersurface data sets can be defined on $\N$. 
\end{remark}
\begin{remark}
In Theorem \ref{thm:Best:THM:JC} we have not restricted the gauges of the data sets $\D$, $\hatD$ (we let the two riggings $L^{\pm}$ be given, but no conditions have been imposed on them). Each specific choice of $L^{\pm}$ will fix a particular gauge on $\D$, $\hatD$. Moreover, Theorem \ref{thm:Best:THM:JC} holds for data sets $\D$, $\hatD$ of any causal nature. In particular, $\D$, $\hatD$ are not required to contain non-null or null points exclusively.
\end{remark}
\begin{remark}
  As proven in \textup{\cite[Lem.\ 3.6]{mars2020hypersurface}}, given {\MHD} $\metdata$ and a point $p\in\N$, the gauge group elements leaving $\metdata$ invariant at $p$ are $(i)$ $\G_{(1,0)}\vert_p$ if $p$ is null and $(ii)$ $\{\G_{(1,0)}\vert_p,\G_{(-1,-2\ell)}\vert_p\}$ if $p$ is non-null, where the vector $\ell\vert_p$ is obtained by raising index to $\ellc|_p$ with the inverse metric $\gamma^{\sharp}\vert_p$ (which in that case exists).
  Since gauge parameters $\{z,V\}$ are smooth by definition, it follows that when $\N$ contains a null point, only the identity element of $\G$ leaves the whole metric hypersurface data invariant. On the contrary, when $\N$ consists exclusively of non-null points there exist two gauge elements which do not transform the metric data. In this last case, the rigging $\G_{(-1,-2\ell)}(\rig)$ corresponds \textup{\cite{mars2020hypersurface}} to the reflection of $\rig$ w.r.t.\ the tangent plane $T_q\phi(\N)$ at each point $q\in\phi(\N)$.

In view of the above, when there are no null points on $\N$, condition $(ii)$ can always be fulfilled once (i) is granted. Indeed,
%
  if there exists a gauge group element $\G_{(z,V)}$ satisfying (i) then this also happens for $\G_{(-1,-2\ell)}\circ\G_{(z,V)}=\G_{(-z,-2\ell-V)}$. 
%
%
%
Thus, there always exists a suitable choice of 
gauge parameter $z$ for which $(i)$ and $(ii)$ hold. 

On the contrary, when $\N$ contains null points only the gauge element $\G_{(1,0)}$ leaves the hypersurface data invariant, which means that $(i)$ can be fulfilled for a gauge group element $\G_{(z,V)}$ but $z$ may have the wrong sign. This is the underlying reason why the spacetime conditions \eqref{junctcondST} provide one unique solution for $\rig^+$ for given $\{\rig^-,\Phi\}$ (see the corresponding discussion in Section \ref{sec:Matching:Paper:Big:Sec}).    
\end{remark}
\begin{remark}
In Theorem \ref{thm:Best:THM:JC}, we have expressed the junction conditions as a restriction over two data sets and a requirement on the sign of a gauge parameter. Theorem \ref{thm:Best:THM:JC} therefore constitutes an abstract formulation of the standard matching conditions. In particular, a remarkable advantage of Theorem \ref{thm:Best:THM:JC} is that it allows us to study different matchings in two different levels. At the first level one takes whatever hypersurface data sets $\D$, $\hatD$ satisfying $(i)$ and studies its properties from a fully detached point of view. At this level, the spacetimes need not even exist. The problem can then move on and study whether or not one can construct spacetimes in which these data can be embedded so that condition $(ii)$ holds.
In other words, by Theorem \ref{thm:Best:THM:JC} one can produce a thin shell of any causality with full freedom to prescribe the gravitational and matter-energy content, and then study the problem of constructing the resulting spacetime $(\M,g)$ which contains it. This is of great use, as it provides a framework to build examples of spacetimes with thin shells of any type. 
\end{remark}
In the setup of Theorem \ref{thm:Best:THM:JC}, the matching riggings are $\{L^-,\rig^+\}$, where $\rig^+$ is of the form \eqref{rig+:orient:signs:etc}. This means that the sign $\epsilon^-$ coincides with the sign $\epsilon$ introduced in 
Definitions \ref{def:thinshell} and \ref{def:energy-momentum_tensor}. It is convenient not to fix the signs $\epsilon^{\pm}$ (or the riggings $L^{\pm}$) a priori because it may well occur that transverse vectors $L^{\pm}$ on each spacetime are already privileged or have been chosen for whatever other reason.  
The main point of the construction in Theorem \ref{thm:Best:THM:JC} is firstly that it provides a fully abstract description of the matching and secondly that it keeps maximum flexibility so that one can adapt Theorem \ref{thm:Best:THM:JC} to any particular scenario.

\section{Abstract formulation of the matching problem: null boundaries}\label{sec:null:boundaries:abstract}
For the remainder of the paper, we focus on the case when both $\D$ and $\hatD$ are null hypersurface data. Under these circumstances, by Lemma \ref{gaugefix} we know that there exists a pair $\{z,V\}$ ensuring that the second and third equations in \eqref{JC:HD} are fulfilled. It follows that the only restrictions are therefore condition $(ii)$ in Theorem \ref{thm:Best:THM:JC} and 
the first equality in \eqref{JC:HD}, namely 
\begin{equation}
\label{isom:cond:abs}\nfi^{\star}\hatg=\gamma. 
\end{equation}
Consequently, given two spacetimes $(\Mpm,g^{\pm})$ with null boundaries $\nullhyp^{\pm}$, 
either there exists (at least) one diffeomorphism $\nfi$ satisfying \eqref{isom:cond:abs} or not. In the former case the matching is possible (provided $(ii)$ holds) and, as we shall see next, all information about the matching is codified by $\nfi$. 

From now on and without loss of generality, we again make the harmless assumption that one of the boundaries lies in the future of its corresponding spacetime while the other lies in its spacetime past 
(see the discussion in \cite{manzano2021null}).  
%
%
The following lemma provides the explicit form of the gauge parameters $\{z,V\}$ and of the matching rigging $\rig^+$ 
in terms of the diffeomorphism $\nfi$.
\begin{lemma}\label{lem:z:and:V:and:rig}
Assume that conditions $(i)$-$(ii)$ in Theorem \ref{thm:Best:THM:JC} hold for a pair of embedded null hypersurface data $\D$, $\hatD$. Then, the gauge parameters $\{z ,V \}$ are given by
\begin{equation}
\label{z:and:V:ABS}z=\frac{1}{(\nfi^{\star}\hatellc)(n)},\qquad V =-P(\nfi^{\star}\hatellc,\cdot)+ \frac{ P(\nfi^{\star}\hatellc,\nfi^{\star}\hatellc)-\nfi^{\star}\hatelltwo}{2(\nfi^{\star}\hatellc)(n)} n.
\end{equation}
Moreover, the matching identifies the rigging 
$L^-$ 
with the following rigging in the plus side
\begin{equation}
\label{rig:general:matching}\rig^+=z'\lp L^+ -\iota^+_{\star}\Big(\nfi_{\star}\big( P(\nfi^{\star}\hatellc,\cdot)\big)\Big)+ \mu \iota^+_{\star}(\nfi_{\star}n) \rp,
\end{equation}
where $z'\in\Fcal^{\star}(\nullhyp^+)$, $\mu\in\Fcal(\nullhyp^+)$ are given
explicitly by 
\begin{equation}
\label{z':mu:rig:functs}
\nfi^{\star}((\iota^+)^{\star}(z'))=\frac{1}{(\nfi^{\star}\hatellc)(n)},\qquad \nfi^{\star}((\iota^+)^{\star}(\mu))=\frac{ P(\nfi^{\star}\hatellc,\nfi^{\star}\hatellc)-\nfi^{\star}\hatelltwo}{2(\nfi^{\star}\hatellc)(n)}.
\end{equation}
\end{lemma}
\begin{proof}
The explicit form \eqref{z:and:V:ABS} for the function $z$ follows from contracting \eqref{JC:HD:inv:elem2} with $n$ and using \eqref{prod2}. The vector field $V$ can be partially obtained also from \eqref{JC:HD:inv:elem2} by particularizing Lemma \ref{lemBestLemmaMarc} for $W=V$, $\bs{\varrho}=z^{-1}\ellc-\nfi^{\star}\hatellc$. This gives 
\begin{equation}
\label{V:prov}V=P\lp \frac{\ellc}{z}-\nfi^{\star}\hatellc,\cdot\rp+u_0n\stackbin{\eqref{prod3}}=-P(\nfi^{\star}\hatellc,\cdot)+\lp u_0-\frac{\elltwo}{z}\rp n,
\end{equation}
where $u_0\defi \ellc(V)$ is a function yet to be determined. This is done by substituting \eqref{V:prov} into \eqref{JC:HD:inv:elem3}. First, 
$\gamma(V,V)=\bs{\varrho}(V)=z^{-2}\elltwo +P(\nfi^{\star}\hatellc,\nfi^{\star}\hatellc)$ because of \eqref{prod2}-\eqref{prod3} and $z^{-1}=(\nfi^{\star}\hatellc)(n)$. Thus,
\begin{align}
\label{eq:for:u0:abs:match} \nfi^{\star}\hatelltwo=\frac{2}{z}\lp \frac{\elltwo}{z}-u_0\rp+ P(\nfi^{\star}\hatellc,\nfi^{\star}\hatellc)\quad\Longrightarrow\quad u_0=\frac{\elltwo}{z}+\frac{z}{2}\lp P(\nfi^{\star}\hatellc,\nfi^{\star}\hatellc)-\nfi^{\star}\hatelltwo\rp
\end{align}
so that substituting this into \eqref{V:prov} proves \eqref{z:and:V:ABS}. Equation \eqref{rig:general:matching} is a direct consequence of \eqref{z:and:V:ABS} and the fact that $\rig^+=z'(L^++\iota^+_{\star}(\nfi_{\star}V))$.
\end{proof}
Whenever there exists a diffeomorphism $\nfi$ solving \eqref{isom:cond:abs} and given a basis $\{n,e_A\}$ of $\Gamma(T\N)$, it is possible to obtain specific expressions for the push-forward vector fields $\{\nfi_{\star}n,\nfi_{\star}e_A\}$. This is done in the next corollary. We use a hat for all objects defined in the data set $\hatD$, in particular $\hatP$ and $\hatn$ are constructed in correspondence with \eqref{prod1}-\eqref{prod4}. 
\begin{corollary}\label{cor:nfistar(n,eA)}
Assume that conditions $(i)$-$(ii)$ in Theorem \ref{thm:Best:THM:JC} hold for a pair of embedded null hypersurface data $\D$, $\hatD$. 
Let $\{n,e_A\}$ be a basis of $\Gamma(T\N)$ and define the covectors $\{\bs{W}_A\}$ and the functions $\{\psi_A,\chi_{(A)}\}$ along $\N$ by 
\begin{equation}
\label{def:cov:WA:paper} \nfi^{\star}\bs{W}_A\defi \gamma(e_A,\cdot),\qquad\psi_A\defi \ellc(e_A),\qquad \chi_{(A)}\defi (\nfi^{-1})^{\star}(z^{-1}\psi_A)-\bs{W}_A(\nfi_{\star}V).
\end{equation}
Then, 

\vspace{-0.5cm}

\begin{multicols}{2}
\noindent
\begin{equation}
\label{nfistar(n):general:abs}\nfi_{\star}n=\big((\nfi^{-1})^{\star}z\big)^{-1} \hatn,
\end{equation}
\begin{equation}
\label{nfistar(e_A):general:abs}\nfi_{\star}e_A=\hatP(\bs{W}_A,\cdot)+\chi_{(A)}\hatn,
\end{equation}
\end{multicols}

\vspace{-0.4cm}

Moreover, it holds that $\hatP(\bs{W}_A,\hatellc)=0$ and $\nfi^{\star}\chi_{(A)}=(\nfi^{\star}\hatellc)(e_A)$. 
\end{corollary}
\begin{proof}
Consider any point $p\in\N$. From \eqref{JC:HD:inv:elem1}  it follows that $\hatg(\nfi_{\star}n,\cdot)\vert_{\nfi(p)}=(\nfi^{\star}\hatg)(n,\cdot)\vert_p=\gamma(n,\cdot)\vert_p=0$, so  
$\nfi_{\star}n=\mathfrak{b} \hatn$ for some function $\mathfrak{b}\in\Fcal(\N)$. This, together with \eqref{z:and:V:ABS} and $\hatellc(\hatn)=1$, entails that $z^{-1}\vert_p=(\nfi^{\star}\hatellc)(n)\vert_p=\hatellc(\nfi_{\star}n)\vert_{\nfi(p)}=\mathfrak{b}\vert_{\nfi(p)}=\nfi^{\star}\mathfrak{b}\vert_p$, which proves \eqref{nfistar(n):general:abs}. 
%
%
On the other hand, 
any vector field $X\in\Gamma(T\N)$ 
satisfies
\begin{align}
\nn \hatg(\nfi_{\star}e_A,\nfi_{\star}X)\vert_{\nfi(p)}&=(\nfi^{\star}\hatg)(e_A,X)\vert_p\stackbin{\eqref{JC:HD:inv:elem1} }=\gamma(e_A,X)\vert_p=\nfi^{\star}\bs{W}_A(X)\vert_p,\\
\nn \hatellc(\nfi_{\star}e_A)\vert_{\nfi(p)}&=(\nfi^{\star}\hatellc)(e_A)\vert_p\stackbin{\eqref{JC:HD:inv:elem2} }=\frac{\psi_A}{z}-\gamma(e_A,V)\vert_p=\frac{\psi_A}{z}-\nfi^{\star}\bs{W}_A(V)\vert_p,
\end{align}
%
which means that
$\hatg(\nfi_{\star}e_A,\cdot)=\bs{W}_A$, 
$\spc\hatellc(\nfi_{\star}e_A)=(\nfi^{-1})^{\star}(z^{-1}\psi_A)-\bs{W}_A(\nfi_{\star}V)$. 
Particularizing Lemma \ref{lemBestLemmaMarc} for the data $\hatD$ and for 
$W=\nfi_{\star}e_A$, $\bs{\varrho}=\bs{W}_A$ and $u_0=(\nfi^{-1})^{\star}(z^{-1}\psi_A)-\bs{W}_A(\nfi_{\star}V)$ yields \eqref{nfistar(e_A):general:abs}.  
Finally, $\hatP(\bs{W}_A,\hatellc)=0$ because 
\begin{align*}
\hatP(\bs{W}_A,\hatellc)\vert_{\nfi(p)}&=-\hatelltwo \bs{W}_A(\hatn)\vert_{\nfi(p)}=-\hatelltwo((\nfi^{-1})^{\star}z) \bs{W}_A(\nfi_{\star}n)\vert_{\nfi(p)}=-\hatelltwo((\nfi^{-1})^{\star}z)\vert_{\nfi(p)}  (\nfi^{\star}\bs{W}_A)(n)\vert_p\\
&=-\hatelltwo((\nfi^{-1})^{\star}z)\vert_{\nfi(p)}  \gamma(e_A,n)\vert_p=0,
\end{align*} 
while 
%
%
%
%
%
%
$\chi_{(A)}\circ\nfi=z^{-1}\psi_A-(\nfi^{\star}\bs{W}_A)(V)=z^{-1}\psi_A-\gamma(e_A,V)\stackbin{\eqref{JC:HD:inv:elem2}}=(\nfi^{\star}\hatellc)(e_A)$ yields $\nfi^{\star}\chi_{(A)}=(\nfi^{\star}\hatellc)(e_A)$.
\end{proof}
\begin{remark}\label{rem:W_A}
From \eqref{nfistar(n):general:abs} it follows that $\nfi$ is a diffeormorphism which sends null generators into null generators. Moreover, since the vector fields $\{W_A\defi \hatP(\bs{W}_A,\cdot)\}$ verify $\hatellc(W_A)=0$, it follows that $W_A\notin\textup{Rad}\hatg$. 
This, together with the fact that $\nfi_{\star}$ is necessarily of maximal rank, force the vector fields $\{W_A\}$ to be everywhere non-zero on $\N$. In fact, $\{\hatn,W_A\}$ constitutes a basis of $\Gamma(T\N)$, since $\{W_A\}$ are all linearly independent. We prove this by contradiction, i.e.\ we assume that one such vector field, e.g.\ $W_2$, can be decomposed as 
  $W_2=\sum_{\s=3}^{\n}c_{\s}W_{\s}$. By \eqref{nfistar(e_A):general:abs}, this would mean that $\nfi_{\star}(e_2-\sum_{\s=3}^{\n}c_{\s}e_{\s})=\big( \chi_{(2)}-\sum_{\s=3}^{\n}c_{\s}\chi_{(\s)}\big) \hatn$, which we know it cannot occur, because only null generators can be mapped to null generators. 
\end{remark}
The point of introducing the objects $\{\bs{W}_A,\chi_{(A)}\}$ will become clear later when we study the particular case when the boundaries have product topology $S\times\mathbb{R}$.  
For the moment, let us simply anticipate that in such case the property $\hatP(\bs{W}_A,\hatellc)=0$ will allow us to conclude that the vector fields $\hatP(\bs{W}_A,\cdot)$ are tangent to the leaves of a specific foliation of $\nullhyp^+$ while from $\nfi^{\star}\chi_{(A)}=(\nfi^{\star}\hatellc)(e_A)$ we will conclude that the functions $\{\chi_{(A)}\}$ are actually spatial derivatives of the step function introduced in Section \ref{sec:Matching:Paper:Big:Sec}. 

One of the relevant results recalled in 
Section \ref{sec:Matching:Paper:Big:Sec}   
is the relation \eqref{detrmfA} between the second fundamental forms of each side. It turns out that in this abstract framework with no topological assumptions one can also recover an equation of this form. 
To do that, we first note that   $\pounds_{f\hatn}\hatg=f\pounds_{\hatn}\hatg$ because $\hatn\in\text{Rad}\hatg$. 
By direct computation one gets 
\begin{equation}
\label{bU:hatbU:matching}\nfi^{\star}\hatbU\defi \frac{1}{2}\nfi^{\star}\lp \pounds_{\hatn}\hatg\rp\stackbin{\eqref{nfistar(n):general:abs}}=\frac{z}{2}  \nfi^{\star}\lp \pounds_{\nfi_{\star}n}\hatg\rp\stackbin{\eqref{JC:HD:inv:elem1}}=\frac{z}{2}\pounds_n\gamma=z\bU\quad\Longrightarrow\quad \bU=\frac{\nfi^{\star}\hatbU}{z},
\end{equation}
which connects the second fundamental forms $\bU$, $\hatbU$ corresponding to the hypersurface data sets $\D$, $\hatD$. Equation \eqref{bU:hatbU:matching} generalizes \eqref{detrmfA} to the case of boundaries with any topology, and has several implications that we discuss below. 
%
%
%

In Theorem \ref{thm:Best:THM:JC} we have seen that when the matching 
is possible there exists a diffeomorphism $\nfi$ verifying \eqref{isom:cond:abs}. In such case, Lemma \ref{lem:z:and:V:and:rig} and Corollary \ref{cor:nfistar(n,eA)} provide explicit expressions for the gauge parameters $\{z,V\}$, the matching rigging $\rig^+$ and the push-forwards $\{\nfi_{\star}n,\nfi_{\star}e_A\}$ of any basis vector fields $\{n,e_A\}$ in terms of the map $\nfi$ still to be determined. 

However, as the reader may have noticed, condition \eqref{isom:cond:abs} does not fix $\nfi$ completely, firstly because there can be more than one diffeomorphism $\nfi$ satisfying \eqref{isom:cond:abs} and secondly because the tensor fields $\gamma$ and $\hatg$ are both degenerate. As happened in 
Section \ref{sec:Matching:Paper:Big:Sec}, 
where the step function could not be fixed directly by the isometry condition \eqref{junct1} 
but \tcr{\eqref{detrmfA} was also required} \cite{manzano2021null}, 
here one also needs an extra condition in order to fix $\nfi$ fully. This additional  restriction is precisely \eqref{bU:hatbU:matching}. As in 
%
%
Section \ref{sec:Matching:Paper:Big:Sec}, 
this provides useful information only when $\bU$ and $\hatbU$ are non-zero. If both are zero then $z$ (and hence part of $\nfi$, recall \eqref{z:and:V:ABS}) remains completely free. This means that one can find an infinite number of diffeomorphisms $\nfi$ verifying \eqref{isom:cond:abs}, with which we recover (and extend to arbitrary topology) the property that whenever the boundaries are totally geodesic then 
%
%
the matching can be performed in an infinite number of ways.

One can obtain explicit expressions for the gravitational and matter-energy content of a general null shell 
in terms of the diffeomorphism $\nfi$. This is done in the following theorem. 
\begin{theorem}\label{thm:Best:THM:matching}
Assume that conditions $(i)$-$(ii)$ in Theorem \ref{thm:Best:THM:JC} hold for a pair of embedded null hypersurface data $\D$, $\hatD$ and let $\fv=\fv^-$. Define 
\begin{equation}
\nn \bY^-\defi \frac{1}{2}(\iota^-)^{\star}(\pounds_{L^-}g^-),\quad \hatbY^+\defi \frac{1}{2}(\iota^+)^{\star}(\pounds_{L^+}g^+)\quad\text{and}\quad\bY^+\defi \frac{1}{2}\nfi^{\star}\lp (\iota^+)^{\star}(\pounds_{\rig^+}g^+)\rp,
\end{equation}
where $\rig^+$ is given by \eqref{rig:general:matching}. Then, 
the tensor $[\bY]\defi \bY^+-\bY^-$ reads
\begin{equation}
\label{[Y]:general:matching}[\Y_{ab}]=z\lp (\nfi^{\star}\hatbY^+)_{ab}+\frac{z}{2}\lp P(\nfi^{\star}\hatellc,\nfi^{\star}\hatellc)-\nfi^{\star}\hatelltwo\rp \U_{ab}-\nablao_{(a}(\nfi^{\star}\hatellc)_{b)}\rp-\Y^-_{ab},
\end{equation}
where  $z\in\Fcal^{\star}(\N)$ is given by \eqref{z:and:V:ABS}.
The components of $[\bY]$ in any basis $\{n,e_A\}$ of $\Gamma(T\N)$ are
\begin{align}
\nn \hspace{-0.3cm}[\bY](e_A,e_B)=&\spc \Big(z(\nfi^{\star}\hatbY^+)-\bY^-\Big)(e_A,e_B)+\frac{z^2}{2}\lp P(\nfi^{\star}\hatellc,\nfi^{\star}\hatellc)-\nfi^{\star}\hatelltwo\rp \bU(e_A,e_B)-ze_A^ae_B^b\nablao_{(a}(\nfi^{\star}\hatellc)_{b)},\\
\label{[Y](n,eA):abs} 
\hspace{-0.3cm}[\bY](n,e_A)= &\spc \Big(z(\nfi^{\star}\hatbY^+)-\bY^-\Big)(n,e_A)-\frac{z}{2}( \pounds_{n}\nfi^{\star}\hatellc)(e_A)+ \frac{e_A(z)}{2z}+ \bsone(e_A) +zP(\nfi^{\star}\hatellc,\bU(e_A,\cdot)) ,\\
\nn \hspace{-0.3cm}[\bY](n,n)=&\spc \Big(z(\nfi^{\star}\hatbY^+)-\bY^-\Big)(n,n)+\frac{n(z)}{z}.
\end{align}
The energy-momentum tensor 
$\tau$ is given by \eqref{tau:prev:abs} in terms of the dual basis $\{\bnormal,\bs{\theta}^A\}$ of $\{n,e_A\}$,  
while the purely gravitational content of the shell is ruled by the tensor 
\begin{align}
\label{purely:grav:abs}\hspace{-0.25cm}\bY^{\mathfrak{G}}(e_A,e_B)\defi [\bY](e_A,e_B)+\frac{\fv\ov{\rho}}{\n-1}\gamma(e_A,e_B),
%
%
\quad\text{where}\quad \ov{\rho}\defi \rho+2 P(\bnormal,\bs{j})+p\lp  2\elltwo + P(\bnormal,\bnormal) \rp
\end{align}
and $\{\rho,p,\bs{j}\}$ are defined as in Remark \ref{rem:def:eg:flux:pressure}.
\end{theorem}
\begin{proof}
Applying Lemma \ref{lem:liegamma} for $\uhat{V}\defi \gamma(V,\cdot)=z^{-1}\ellc-\nfi^{\star}\hatellc$ (recall \eqref{JC:HD:inv:elem2}) 
and $u_0\defi \ellc(V)$ (cf.\ \eqref{eq:for:u0:abs:match}) yields
\begin{align}
\nn \frac{z}{2}\pounds_{V}\gamma_{ab}&=\lp \elltwo+\frac{z^2}{2}\lp P(\nfi^{\star}\hatellc,\nfi^{\star}\hatellc)-\nfi^{\star}\hatelltwo\rp\rp \U_{ab}+z\nablao_{(a}\lp \frac{\ell_{b)}}{z}-(\nfi^{\star}\hatellc)_{b)}\rp\\
%
%
\label{pounds_Vgamma:abs}&=\frac{z^2}{2}\lp P(\nfi^{\star}\hatellc,\nfi^{\star}\hatellc)-\nfi^{\star}\hatelltwo\rp \U_{ab}-\frac{1}{z}(\nablao_{(a}z)\ell_{b)}-z\nablao_{(a}(\nfi^{\star}\hatellc)_{b)},
\end{align}
where in the last step we used that $\nablao_{(a}\ell_{b)}=-\elltwo\U_{ab}$ (cf.\ \eqref{nablaoll}). By hypothesis the matching of $(\Mpm,g^{\pm})$ is possible, so the data sets $\{\N,\nfi^{\star}\hatg, \nfi^{\star}\hatellc,\nfi^{\star}\hatelltwo,\nfi^{\star}\hatbY^+\}$, $\{\N,\gamma,\ellc,\elltwo,\bY^+\}$ are embedded in $(\Mp,g^{+})$ with embedding $\iota^+\circ\nfi$ and respective riggings $L^+$, $\rig^+$. This, together with \eqref{JC:HD}, entails that the tensors $\nfi^{\star}\hatbY^+$, $\bY^+$ are related by $\bY^+=\G_{(z,V)}(\nfi^{\star}\hatbY^+)$, where $\{z,V\}$ are given by \eqref{z:and:V:ABS}. Thus (cf.\ \eqref{gaugeY}, \eqref{isom:cond:abs}) 
\begin{align}
\label{Y^+likepaper1:(A)} \hspace{-0.2cm}\bY^+=&\spc z \nfi^{\star}\hatbY^++dz \otimes_s(\nfi^{\star}\hatellc+\gamma(V ,\cdot))+\frac{z}{2} \pounds_{V }\gamma\stackbin{\eqref{JC:HD:inv:elem2}}= z \nfi^{\star}\hatbY^++\frac{dz }{z }\otimes_s\ellc+\frac{z}{2} \pounds_{V }\gamma.
\end{align}
%
%
%
%
%
Inserting \eqref{pounds_Vgamma:abs} into \eqref{Y^+likepaper1:(A)} yields
the explicit form \eqref{[Y]:general:matching}. 

We now obtain the components of $[\bY]$ in the basis $\{n,e_A\}$, for which we recall that $\bU(n,\cdot)=0$ and $\bsone(n)=0$. Particularizing \eqref{contrNsym} for $\bs{\theta}=\nfi^{\star}\hatellc$ and using \eqref{z:and:V:ABS} gives
\begin{align}
n^a \nablao_{(a}(\nfi^{\star}\hatellc)_{b)}=&\spc \frac{1}{2}\pounds_{n}(\nfi^{\star}\hatellc)_{b}+ \frac{1}{2} \nablao_{b}((\nfi^{\star}\hatellc)({n}))-(\nfi^{\star}\hatellc)(n)\sone_b  - P^{ac}\U_{bc}(\nfi^{\star}\hatellc)_{a} , \nn\\
=&\spc \frac{1}{2}\pounds_{n}(\nfi^{\star}\hatellc)_{b}- \frac{\nablao_{b}z}{2z^2} - \frac{\sone_b}{z} -P^{ac}\U_{bc}(\nfi^{\star}\hatellc)_{a} , \label{contrNsym:nfistar:hatellc1}\\
n^an^b \nablao_{(a}(\nfi^{\star}\hatellc)_{b)}=&\spc \frac{1}{2}\pounds_{n}\lp (\nfi^{\star}\hatellc)(n)\rp- \frac{n(z)}{2z^2} =-\frac{n(z)}{z^2}. \label{contrNsym:nfistar:hatellc2}
\end{align}
Combining \eqref{contrNsym:nfistar:hatellc1}-\eqref{contrNsym:nfistar:hatellc2} with \eqref{[Y]:general:matching} yields \eqref{[Y](n,eA):abs}. 
\tcr{The components of the energy-momentum tensor being given by \eqref{tau:prev:abs} is just the contents of Corollary \ref{cor:eg-mom-tensor:abstract}.}
%
%
%
%
%
%
%
%
%
%
%
%
Finally, we prove \eqref{purely:grav:abs} as follows. First, we note that the one-forms $\bs{j}$ (see Remark \ref{rem:def:eg:flux:pressure}) and $\ellc$ decompose in the basis $\{\bnormal,\bs{\theta}^A\}$ as
\begin{equation}
\label{one-forms:j:ellc}\bs{j}=\bs{j}(e_A)\bs{\theta}^A,\qquad \ellc=\bnormal+\ellc(e_A)\bs{\theta}^A
\end{equation}
because $\bs{j}(n)=0$ and $\ellc(n)=1$. Also by Remark \ref{rem:def:eg:flux:pressure}, we know that the one-form $\bs{j}$ verifies $[\bY](n,e_A)=\fv(\bs{j}(e_A)-p\ellc(e_A))$. 
Thus, a direct computation based on the decomposition \eqref{Pdecom:proof} of the tensor field $P$ yields 
\begin{align*}
\text{tr}_P[\bY]&=P^{ab}[\Y_{ab}]=\mathfrak{h}^{AB}[\bY](e_A,e_B)+2 P(\bnormal,\bs{\theta}^A) [\bY](n,e_A)+P(\bnormal,\bnormal)[\bY](n,n)\\
&=\mathfrak{h}^{AB}[\bY](e_A,e_B)+2\fv P(\bnormal,\bs{j}(e_A)\bs{\theta}^A-p\ellc(e_A)\bs{\theta}^A) -\fv p P(\bnormal,\bnormal)\\
&=\mathfrak{h}^{AB}[\bY](e_A,e_B)+2\fv P(\bnormal,\bs{j})+\fv p \big(2  \elltwo + P(\bnormal,\bnormal)\big)
\end{align*}
where we used that $P(\bs{\theta}^A,\bs{\theta}^B)=\mathfrak{h}^{AB}$ (by Lemma \ref{lem:Pdecom:dual}), $P(\ellc,\cdot)=-\elltwo n$ (cf.\ \eqref{prod3}) and \eqref{one-forms:j:ellc} in this order. 
%
%
%
Taking into account the definition of the energy density $\rho$ (see \eqref{def:rho:jA:p:abs}), one finds
\begin{equation}
\label{h-trace:widetilde:rho}\mathfrak{h}^{AB}[\bY](e_A,e_B) =-\fv\Big(\rho+2 P(\bnormal,\bs{j})+p\lp  2\elltwo + P(\bnormal,\bnormal) \rp\Big)\defi -\fv \ov{\rho}.
\end{equation}
Now, from \eqref{tau:prev:abs} it is clear that the only part of $[\bY]$ that does not contribute to the energy-momentum tensor is the $\mathfrak{h}$-traceless part of $[\bY](e_A,e_B)$. By Lemma \ref{lem:Pdecom:dual}, we know that $\mathfrak{h}^{AB}\gamma(e_A,e_B)=\n-1$.
Consequently, $[\bY](e_A,e_B)$ decomposes in a $\mathfrak{h}$-traceless and a $\mathfrak{h}$-trace part as
\begin{equation}
\nn [\bY](e_A,e_B)=\bY^{\mathfrak{G}}(e_A,e_B)+\frac{\mathfrak{h}^{IJ}[\bY](e_I,e_J)}{\n-1}\gamma(e_A,e_B),
\end{equation}
from where \eqref{purely:grav:abs} follows at once after inserting \eqref{h-trace:widetilde:rho}. 
\end{proof} 
\begin{remark}
We emphasize that we have not made any assumption on the topology of the boundaries $\nullhyp^{\pm}$ in Theorems \ref{thm:Best:THM:JC} and \ref{thm:Best:THM:matching} or in Lemma \ref{lem:z:and:V:and:rig}. The results above therefore describe the most general matching of two spacetimes across null hypersurfaces and generalize the results in \textup{\cite{manzano2021null}} and  \textup{\cite{manzano2022general}}, 
where the existence of a foliation on the boundaries played an important role. 

The gravitational/matter-energy content of the resulting null shell is given by Theorem \ref{thm:Best:THM:matching}, 
%
and the associated energy density $\rho$, energy flux $j$ and pressure $p$ are given by \eqref{def:rho:jA:p:abs}. 
The reason why we refer to $\bY^{\mathfrak{G}}(e_A,e_B)$ as the purely gravitational part of the shell is that only the components $[\bY](n,n)$, $[\bY](n,e_A)$ and the trace $P(\bs{\theta}^A,\bs{\theta}^B)[\bY](e_A,e_B)$ contribute to the energy-momentum tensor $\tau$ (cf.\ \eqref{tau:prev:abs}). This means that even if $\tau$ vanishes identically $\bY^{\mathfrak{G}}(e_A,e_B)$ does not need to be zero. Such a case
  corresponds to an impulsive gravitational wave propagating in the spacetime resulting from the matching. 
\end{remark}
\begin{remark}
By Lemma \ref{lem:Pdecom:dual} we know that $P(\bnormal,\cdot)=0$ if and only if $\ellc(e_A)=0$ and $\elltwo=0$. In such case, the scalar $\ov{\rho}$ coincides with the energy density $\rho$ of the shell. In the embedded picture, these restrictions amount to impose that the matching riggings $\rig^{\pm}$ are null and orthogonal to the vector fields $\phi^{\pm}_{\star}e_A$. \tcr{In particular, in the setup of   
%
%
Section \ref{sec:Matching:Paper:Big:Sec} this holds when 
the rigging $\rig^-$ (chosen according to \eqref{eiyl}) is null and orthogonal to the leaves of the foliation on the minus side (hence $\mu_A^-=0$, cf. \eqref{eqA30}).} 
\end{remark}
\begin{remark}
  In Theorems \ref{thm:Best:THM:JC} and \ref{thm:Best:THM:matching} and Lemma \ref{lem:z:and:V:and:rig}, all expressions are fully explicit in terms of the diffeomorphism $\nfi$. The two data sets $\D$, $\hatD$ are completely known (because the embeddings $\iota^{\pm}$ and the spacetimes $(\Mpm,g^{\pm})$ are given) and the rigging $\rig^+$ is determined by the pair $\{z,V\}$ given by \eqref{z:and:V:ABS} in terms of $\nfi$. This is related to the results in
    \textup{\cite{manzano2021null}}, \textup{\cite{manzano2022general}} summarized in Section \ref{sec:Matching:Paper:Big:Sec}, 
%
%
%
%
%
%
%
%
where the whole matching 
depended upon the step function $H$ and the coefficients $b_I^J$, which in turn determined the matching embedding $\phi^+$ (recall \eqref{bIJ:coef} and \eqref{embedPhi+}) and the matching rigging $\rig^+$ (according to \eqref{uniqxi}). 
\end{remark}
Expressions \eqref{[Y](n,eA):abs} 
%
%
involve the pull-back $\nfi^{\star}\hatbY^+$, whose calculation can be cumbersome in general. It is more convenient to rewrite 
%
%
\eqref{[Y](n,eA):abs}  
in terms of pull-backs of scalar functions referred to the data $\hatD$ and objects defined with respect to $\D$. We provide the corresponding expressions in the next lemma. 
%
%
%
%
%
\begin{lemma}\label{lem:tau:point}
  Assume that conditions $(i)$-$(ii)$ in Theorem \ref{thm:Best:THM:JC} hold for a pair of embedded null hypersurface data $\D$, $\hatD$ and let $\fv=\fv^-$.
  Define the tensors $\{\bY^-,\hatbY{}^+,\bY^+\}$ as in Theorem \ref{thm:Best:THM:matching}, the covectors $\{\bs{W}_A\}$ and the functions $\{\chi_{(A)}, \psi_A\}$ along $\N$ according to Corollary \ref{cor:nfistar(n,eA)} and the vector field $W_A\defi \hatP(\bs{W}_A,\cdot)$. Let $z$ be given by \eqref{z:and:V:ABS} and $\{n,e_A\}$ be a basis of $\Gamma(T\N)$ with dual basis $\{\bnormal,\bs{\theta}^A\}$.
Then, equations 
\eqref{[Y](n,eA):abs} 
%
%
can be rewritten as
%
%
%
%
%
%
\begin{align}
\label{[Y](n,n):abs:point}\hspace{-0.5cm}[\bY](n,n)=&\spc \frac{1}{z}\nfi^{\star}\lp \hatbY^+(\hatn,\hatn)\rp-\bY^-(n,n)+\frac{n(z)}{z},\\
\nn \hspace{-0.5cm}[\bY](n,e_A)= &\spc \nfi^{\star}\lp \hatbY^+(\hatn,W_A)+\chi_{(A)}\hatbY^+(\hatn,\hatn)\rp-\bY^-(n,e_A) \\
\label{[Y](n,eA):abs:point}&-\frac{z}{2}\lp \pounds_{n}\nfi^{\star}\hatellc\rp(e_A)+ \frac{e_A(z)}{2z}+ \bsone(e_A) +zP(\nfi^{\star}\hatellc,\bU(e_A,\cdot)),\\
\nn \hspace{-0.5cm}[\bY](e_A,e_B)=&\spc z\nfi^{\star}\bigg( \hatbY^+(W_A,W_B)+\chi_{(A)}\hatbY^+(\hatn,W_B)+ \chi_{(B)}\hatbY^+(\hatn,W_A)+\chi_{(A)}\chi_{(B)}\hatbY^+(\hatn,\hatn)\bigg)\\
\label{[Y](eA,eB):abs:point} &-\bY^-(e_A,e_B)-ze_A^ae_B^b\nablao_{(a}(\nfi^{\star}\hatellc)_{b)} +\frac{z^2}{2}\lp P(\nfi^{\star}\hatellc,\nfi^{\star}\hatellc)-\nfi^{\star}\hatelltwo\rp \bU(e_A,e_B).
\end{align}
The energy-momentum tensor 
$\tau$ is given by \eqref{tau:prev:abs} in terms of the dual basis $\{\bnormal,\bs{\theta}^A\}$ of $\{n,e_A\}$.
\end{lemma}
\begin{proof}
Inserting $(\nfi^{\star}\hatbY^+)(X,Y)\vert_p=\hatbY^+(\nfi_{\star}X,\nfi_{\star}Y)\vert_{\nfi(p)}$ into 
%
%
\eqref{[Y](n,eA):abs} 
and using \eqref{nfistar(n):general:abs}-\eqref{nfistar(e_A):general:abs}, equations
\eqref{[Y](n,n):abs:point}-\eqref{[Y](eA,eB):abs:point}
%
%
follow at once. We already know from Corollary \ref{cor:eg-mom-tensor:abstract} that $\tau$ is given by \eqref{tau:prev:abs}.
\end{proof}
In Section \ref{sec:matching:fol:abs}, we shall recover the results of Proposition \ref{prop6} by particularizing Lemma \ref{lem:tau:point} to the case when the boundaries $\nullhyp^{\pm}$ have product topology. Lemma \ref{lem:tau:point} therefore generalizes Proposition \ref{prop6} to (null) boundaries of any topology, and 
determines the matter-energy content of any null thin shell arising from the matching of two spacetimes.  

\subsection{Pressure of the shell}\label{sec:pressure:abstract}
In 
%
%
\cite{manzano2021null}, \cite{manzano2022general}
we 
discussed the effect and the importance of a non-zero pressure in a null thin shell. This, however, 
was
done in very specific contexts (namely in the matching of two regions of Minkowski across a null hyperplane or for matchings across embedded AKH$_0$s) and by following a non-fully geometric approach (i.e.\ by analyzing the effect of the pressure in some specific coordinates). Our aim in this section is to study the pressure of a completely general null shell at a fully abstract level, providing its explicit expression in terms of well-defined geometric quantities and reinforcing the geometric interpretation of 
\cite{manzano2021null} and \cite{manzano2022general}.

In the following lemma we find explicit expressions for the pressure $p$ in terms of the surface gravities of various null generators of $\N$. 
\begin{lemma}\label{lem:pressure:abs}
  Assume that conditions $(i)$-$(ii)$ in Theorem \ref{thm:Best:THM:JC} hold for a pair of embedded null hypersurface data $\D=\{\N,\gamma,\ellc,\elltwo,\bY^-\}$, $\hatD=\{\N,\hatg,\hatellc,\hatelltwo,\hatbY^+\}$ and a diffeomorphism $\nfi$. Define $z$ by \eqref{z:and:V:ABS} and introduce 
\begin{equation}
\label{defs:kappas:lemma:p}
\kappa_n\defi -\bY^-(n,n),\qquad\widehat{\kappa}_n\defi -\hatbY^+(\hatn,\hatn),\qquad \nfi^{\star}\kappa_{\nfi_{\star}n} \defi  \frac{1}{z}\lp \nfi^{\star}\widehat{\kappa}_n-n(z)\rp.
\end{equation}
Then, the pressure $p$ of the corresponding null shell is given by
\begin{equation}
\label{pressure1:abs}p=-\fv \lp \Q-\G_{(z,\ov{V})}(\nfi^{\star}\widehat{\kappa}_n)\rp\qquad\text{or, equivalently}\qquad p=-\fv \lp \Q-\nfi^{\star}\kappa_{\nfi_{\star}n}\rp 
\end{equation}
where $\fv=\fv^-$ and $\ov{V}\in\Gamma(T\N)$ is a vector field that can be chosen at will. In particular, the pressure vanishes if and only if
\begin{equation}
\label{pressure:abs:vanishing}\G_{(z,\ov{V})}(\nfi^{\star}\widehat{\kappa}_n)=\Q\qquad\text{or, equivalently}\qquad \nfi^{\star}\kappa_{\nfi_{\star}n} =\Q.
\end{equation}
\end{lemma}
\begin{proof}
Recall that $\G^{-1}_{(z,\ov{V})}=\G_{(z^{-1},-z\ov{V})}$. We start by noticing that \eqref{Qprime} implies that $\G^{-1}_{(z,\ov{V})}(\Q)=\mathcal{G}_{\lp z^{-1},-z\ov{V}\rp}\lp \Q\rp =z\lp\Q+\frac{n(z)}{z}\rp$. 
On the other hand, combining \eqref{def:rho:jA:p:abs} and \eqref{[Y](n,n):abs:point}, it follows 
\begin{align}
\label{pressure:midEQ} p=& \spc -\fv\lp -\frac{1}{z}\nfi^{\star}\widehat{\kappa}_n+\kappa_n+\frac{n(z)}{z}\rp= -\frac{\fv}{z}\lp \G_{(z^{-1},-z\ov{V})}(\kappa_n)-\nfi^{\star}\widehat{\kappa}_n\rp=-\frac{\fv}{z}\lp \G^{-1}_{(z,\ov{V})}(\kappa_n)-\nfi^{\star}\widehat{\kappa}_n\rp.
\end{align}
Recalling the transformation law for $\fv$ and $p$ in \eqref{gauge:sign} and \eqref{gauge:pressure} this expression can be written as  $-\G^{-1}_{(z,\ov{V})}(\fv p)= \G^{-1}_{(z,\ov{V})}(\kappa_n)-\nfi^{\star}\widehat{\kappa}_n$. Applying $\G_{(z,\ov{V})}$ on both sides  one obtains the left part of \eqref{pressure1:abs}. The right part of \eqref{pressure1:abs} is an immediate consequence of inserting the definition 
\eqref{defs:kappas:lemma:p}
of $\nfi^{\star}\kappa_{\nfi_{\star}n} $ into the first line of \eqref{pressure:midEQ},
while \eqref{pressure:abs:vanishing} is proven by setting $p=0$ in \eqref{pressure1:abs}.
\end{proof}
\begin{remark}
The last expression in \eqref{defs:kappas:lemma:p} defines a function $\kappa_{\nfi_{\star}n}$
on $\N$. However, we still need to justify this terminology.   It turns out that
$\kappa_{\nfi_{\star}n}$ coincides with the surface gravity of the vector field $\nfi_{\star}n$ w.r.t.\ the hypersurface connection $\widehat{\ovnabla}$ constructed from the data $\hatD$.  
To prove this, we let $\mathfrak{z}\defi (\nfi^{-1})^{\star}z$, so that (cf.\ \eqref{defs:kappas:lemma:p})
\begin{align*}
 \kappa_{\nfi_{\star}n} &=  \frac{1}{\mathfrak{z}}\lp \widehat{\kappa}_n-(\nfi^{-1})^{\star}(n(z))\rp\quad\text{and}\quad  (\nfi_{\star}n)( \mathfrak{z})=(\nfi^{-1})^{\star}(n(z)),
\end{align*}
where the right part follows from  
$(\nfi_{\star}n)( \mathfrak{z})\vert_{\nfi(p)}=(\nfi^{\star}d\mathfrak{z})(n)\vert_{p}=(d\nfi^{\star}\mathfrak{z})(n)\vert_{p}=n(z)\vert_{p}
=(\nfi^{-1})^{\star}(n(z))\vert_{\nfi(p)}$. 
Then, the combination of 
\eqref{nablao_nn} and \eqref{nfistar(n):general:abs} gives\footnote{Recall that the connections $\nablao$, $\ovnabla$ of a data set $\hypdata$  verify $\ovnabla_{X}Z=\nablao_XZ-\bY(X,Z)n$, $\forall X,Z\in\Gamma(T\N)$.}
\begin{align*} 
\widehat{\ovnabla}_{\nfi_{\star}n}\nfi_{\star}n&=\frac{1}{\mathfrak{z}}\widehat{\ovnabla}_{\hatn}\lp \frac{\hatn}{\mathfrak{z}}\rp=\frac{1}{\mathfrak{z}}\lp \frac{1}{\mathfrak{z}}\widehat{\ovnabla}_{\hatn}\hatn-\frac{\hatn(\mathfrak{z})}{\mathfrak{z}^2} \hatn\rp=-\frac{1}{\mathfrak{z}}\lp \hatbY^+(\hatn,\hatn)+\frac{\hatn( \mathfrak{z})}{ \mathfrak{z}}\rp \nfi_{\star}n\\
&=\frac{1}{\mathfrak{z}}\Big(\widehat{\kappa}_n-(\nfi_{\star}n)( \mathfrak{z})\Big) \nfi_{\star}n=\frac{1}{\mathfrak{z}}\Big(\widehat{\kappa}_n-(\nfi^{-1})^{\star}(n(z))\Big) \nfi_{\star}n=\kappa_{\nfi_{\star}n}\nfi_{\star}n.
\end{align*}
\end{remark}
\begin{remark}
The gauge parameter $\ov{V}$ is completely superfluous and plays no role in determining the pressure, which is only influenced by the function $z$ given by \eqref{z:and:V:ABS}. We keep $\ov{V}$ in the expression to emphasize this fact.
  \end{remark}
\begin{remark}
In \textup{\cite{manzano2021null}}, \textup{\cite{manzano2022general}}, 
we have introduced the notion of self-compression and self-stretching on the boundaries of the spacetimes to be matched. We have seen that this effect is completely ruled by the pressure, and that it has to do with the differences in the acceleration along the null generators of both sides. With \eqref{pressure1:abs}, we recover the same result but for the case of boundaries with any topology. Indeed, the surface gravities $\Q$ and $\kappa_{\nfi_{\star}n}$ verify
$\ovnabla_nn=\kappa_nn$ and 
$\widehat{\ovnabla}_{\nfi_{\star}n}\nfi_{\star}n=\kappa_{\nfi_{\star}n}\nfi_{\star}n$,
so 
that the quantity $-\fv p$ is positive when $\kappa_n>\nfi^{\star}\kappa_{\nfi_{\star}n}$ (namely when the ``acceleration" of $n$ is greater than that of $\nfi_{\star}n$) and negative otherwise. 
The only scenario where there exists no pressure occurs when both surface gravities coincide, i.e.\ when the accelerations of $n$ and $\nfi_{\star}n$ are the same. 
\end{remark}

\section{Multiple matchings across null boundaries}\label{sec:mult:mathcings}
We have already seen that generically 
  there exists at most one way of matching two given spacetimes $(\Mpm,g^{\pm})$ 
  (i.e.\ only one matching map $\Phi$ 
or one single diffeomorphism $\nfi$). However, we have also mentioned  
that sometimes multiple (even infinite) matchings can be performed (e.g.\ when both second fundamental forms $\bU$, $\hatbU$ vanish). In the language of \eqref{junctcondAbs2}, this means that given a choice of embedding $\phi^-$ and matching rigging $\rig^-$ on the minus side, there exist several embeddings $\phi^+$ for which the matching conditions hold, and each embedding gives rise to a unique solution for the rigging $\rig^+$ with suitable orientation. 
%

In this section, our aim is to study the scenario of multiple matchings.  
The idea is to assume that \textit{all} information about one of the matchings is known, in particular its corresponding diffeomorphism $\nfi$ 
and hence the gravitational/matter-energy content.  
As we shall see, in these circumstances 
one only needs to consider a single hypersurface data set $\D$ (instead of two) and it is possible to provide explicit expressions for the jump $[\bY]$ and the energy-momentum tensor $\tau$ of any other shell in terms of their counterparts of the known matching. These results can be particularized to the case when the known matching gives rise to no-shell (i.e.\ when it is such that $[\bY]=0$). This precisely happens in all cut-and-paste constructions, where $(\Mpm,g^{\pm})$ are two regions of the same spacetime. 


Our setup will be the following. We make a choice $\{\phi^-,\rig^-\}$ of embedding and rigging on the minus side and consider two matching embeddings $\phi^+$, $\widetilde{\phi}^+$, each of them satisfying \eqref{junctcondAbs2} for two riggings $\rig^+$, $\widetilde{\rig}^+$ respectively. We also assume that the information about one of the matchings is completely known, namely we let $\{\widetilde{\phi}^+,\widetilde{\rig}^+\}$  
be given. 

From the spacetimes $(\Mpm,g^{\pm})$, we can construct two hypersurface data sets $\D=\{\N,\gamma,\ellc,\elltwo,\bY^-\}$, $\hatD=\{\N,\hatg,\hatellc,\hatelltwo,\hatbY^+\}$ and Theorem \ref{thm:Best:THM:JC} ensures that we can find two  diffeomorphisms $\nfi$, $\widetilde{\nfi}$ and two pairs $\{z,V\}$, $\{\widetilde{z},\widetilde{V}\}$ for which $(i)$-$(ii)$ hold. Even more, since the pair $\{\widetilde{\phi}^+,\widetilde{\rig}^+\}$ is known, we can always make the choice $\{\iota^+=\widetilde{\phi}^+,L^+=\widetilde{\rig}^+\}$ so that $\{\hatg,\hatellc,\hatelltwo\}=\{\gamma,\ellc,\elltwo\}$ 
and $\widetilde{\nfi}$ is the identity map, i.e.\ $\widetilde{\nfi}=\mathbb{I}_{\N}$. 
In these circumstances, using \eqref{prod2}-\eqref{prod3} in \eqref{z:and:V:ABS} yields $\widetilde{z}=1$ and $\widetilde{V}=0$. Making the same choice of $\{\iota^+,L^+\}$ for the matching of $\nfi$ transforms \eqref{JC:HD} into
\begin{equation}
\label{JC:HD:multiple}\G_{(z ,V )}(\nfi^{\star}\gamma)=\gamma,\qquad\G_{(z ,V )}(\nfi^{\star}\ellc)=\ellc,\qquad\G_{(z ,V )}(\nfi^{\star}\elltwo)=\elltwo,
\end{equation} 
and forces the embedding $\phi^+$ to be given by
$\widetilde{\phi}^+\circ\nfi\equiv \phi^+$. 
Equations \eqref{JC:HD:inv:elem1}-\eqref{JC:HD:inv:elem3} now read
\begin{align}
\label{JC:HD:inv:elem:mult}\nfi^{\star}\gamma&=\gamma,\qquad\nfi^{\star}\ellc=\frac{\ellc}{z}- \gamma \lp V,\cdot\rp,\qquad\nfi^{\star}\elltwo=\frac{\ell^{(2)}}{z^2}-\frac{2\ellc\lp V\rp}{z}+ \gamma \lp V,V\rp,
\end{align} 
while the expressions \eqref{z:and:V:ABS} for the gauge parameters $\{z,V\}$ become 
\begin{equation}
\label{z:and:V:ABS:mult}z=\frac{1}{(\nfi^{\star}\ellc)(n)},\qquad V =-P(\nfi^{\star}\ellc,\cdot)+ \frac{ P(\nfi^{\star}\ellc,\nfi^{\star}\ellc)-\nfi^{\star}\elltwo}{2(\nfi^{\star}\ellc)(n)} n.
\end{equation}
It is important to emphasize that whereas $\widetilde{\nfi}=\mathbb{I}_{\N}$ forces the metric parts of $\D$, $\hatD$ to be the same, 
the tensors $\bY^-$, $\hatbY^+$ do not coincide in general. We let $[\widetilde{\bY}]\defi \hatbY^+-\bY^-$, $[\bY]\defi \bY^+-\bY^-$ be the jumps codifying the  gravitational/matter-energy content of the null shells associated to $\widetilde{\nfi}$ and $\nfi$ respectively. 
%
%
Then, 
by \eqref{[Y]:general:matching} we know that $[\bY]$ must be given by 
%
\begin{align}
\label{[Y]:general:matching:mult1}[\Y_{ab}]&=z\lp (\nfi^{\star}\hatbY^+)_{ab}+\frac{z}{2}\lp P(\nfi^{\star}\ellc,\nfi^{\star}\ellc)-\nfi^{\star}\elltwo\rp \U_{ab}-\nablao_{(a}(\nfi^{\star}\ellc)_{b)}\rp-\Y^-_{ab}.
\end{align}
The jumps $[\bY]$, $[\widetilde{\bY}]$ can actually be related, as we shall see next. Indeed, by defining the tensor 
\begin{equation}
\label{def:tensor:calig:Y}\bs{\mathcal{Y}}\defi z \nfi^{\star}\hatbY^+-\hatbY^+,
\end{equation}
expression \eqref{[Y]:general:matching:mult1} can be rewritten as 
\begin{align}
\label{[Y]:general:matching:mult2}[\Y_{ab}]&=\mathcal{Y}_{ab}+\frac{z^2}{2}\lp P(\nfi^{\star}\ellc,\nfi^{\star}\ellc)-\nfi^{\star}\elltwo\rp \U_{ab}-z\nablao_{(a}(\nfi^{\star}\ellc)_{b)}+[\widetilde{\Y}_{ab}].
\end{align}
Moreover, a direct calculation shows that the components \eqref{[Y](n,eA):abs} of $[\bY]$ in a basis $\{n,e_A\}$ of $\Gamma(T\N)$ can be expressed in terms of $\bs{\mathcal{Y}}$ as 
\begin{align}
\label{[Y](n,n):abs:mult}\hspace{-0.32cm}[\bY](n,n)=&\spc \bs{\mathcal{Y}}(n,n)+[\widetilde{\bY}](n,n)+\frac{n(z)}{z},\\
\label{[Y](n,eA):abs:mult} \hspace{-0.32cm}[\bY](n,e_A)= &\spc \bs{\mathcal{Y}}(n,e_A)+[\widetilde{\bY}](n,e_A)-\frac{z}{2}( \pounds_{n}\nfi^{\star}\ellc)(e_A) + \frac{e_A(z)}{2z}+ \bsone(e_A) +zP(\nfi^{\star}\ellc,\bU(e_A,\cdot)) ,\\
\label{[Y](eA,eB):abs:mult} \hspace{-0.32cm}[\bY](e_A,e_B)=&\spc \bs{\mathcal{Y}}(e_A,e_B)+[\widetilde{\bY}](e_A,e_B)+\frac{z^2}{2}\lp P(\nfi^{\star}\ellc,\nfi^{\star}\ellc)-\nfi^{\star}\elltwo\rp \bU(e_A,e_B)-ze_A^ae_B^b\nablao_{(a}(\nfi^{\star}\ellc)_{b)}.
\end{align}
Inserting \eqref{[Y](n,n):abs:mult}-\eqref{[Y](eA,eB):abs:mult} 
%
%
into \eqref{tau:prev:abs}
%
%
%
%
%
%
gives us the relation between the energy-momentum tensors $\tau$, $\hattau$ of the two shells. Specifically, for the dual basis $\{\bnormal,\bs{\theta}^A\}$ of $\{n,e_A\}$ one finds (recall that 
$\mathfrak{h}_{AB}\defi \gamma(e_A,e_B)$)
\begin{align}
\nn \hspace{-0.25cm}\tau(\bnormal ,\bnormal )=&\spc  \hattau(\bnormal ,\bnormal )-\fv\mathfrak{h}^{AB}\bigg(  \bs{\mathcal{Y}}(e_A,e_B)+\frac{z^2}{2}\lp P(\nfi^{\star}\ellc,\nfi^{\star}\ellc)-\nfi^{\star}\elltwo\rp \bU(e_A,e_B)-ze_A^ae_B^b\nablao_{(a}(\nfi^{\star}\ellc)_{b)} \bigg),\\
\label{tauellthetaA:abs:mult} \hspace{-0.25cm}\tau(\bnormal ,\bs{\theta}^A)=&\spc \hattau(\bnormal ,\bs{\theta}^A) +\fv\mathfrak{h}^{AB}\bigg(  \bs{\mathcal{Y}}(n,e_B)-\frac{z}{2}( \pounds_{n}\nfi^{\star}\ellc)(e_B)+ \frac{e_B(z)}{2z}+ \bsone(e_B) +zP(\nfi^{\star}\ellc,\bU(e_B,\cdot)) \bigg) ,\\
\nn \hspace{-0.25cm}\tau(\bs{\theta}^A,\bs{\theta}^B)=&\spc \hattau(\bs{\theta}^A,\bs{\theta}^B)-\fv\mathfrak{h}^{AB}\lp \bs{\mathcal{Y}}(n,n)+\frac{n(z)}{z} \rp .
\end{align}
The results \eqref{[Y](n,n):abs:mult}-\eqref{tauellthetaA:abs:mult} 
turn out to be of particular interest when one of the matchings of $(\Mpm,g^{\pm})$ gives rise to no shell. 
In order to see this, let us assume that this is the case and take $\widetilde{\nfi}$ to be the diffeomorphism corresponding to the no-shell matching. Then, $[\widetilde{\bY}]=0$ (i.e.\ $\hatbY^+=\bY^-$) holds, which means that the tensor $\bs{\mathcal{Y}}$ is given by (cf.\ \eqref{def:tensor:calig:Y})
\begin{equation}
\label{calig:Y:when:no:shell}\bs{\mathcal{Y}}= z \nfi^{\star}\bY^--\bY^-.
\end{equation}
Setting $[\widetilde{\bY}]=0$ in equations \eqref{[Y](n,n):abs:mult}-\eqref{[Y](eA,eB):abs:mult} yields
\begin{align}
\label{[Y](n,n):abs:mult:triv}[\bY](n,n)=&\spc \bs{\mathcal{Y}}(n,n)+\frac{n(z)}{z},\\
\label{[Y](n,eA):abs:mult:triv} [\bY](n,e_A)= &\spc \bs{\mathcal{Y}}(n,e_A)-\frac{z}{2}( \pounds_{n}\nfi^{\star}\ellc)(e_A) + \frac{e_A(z)}{2z}+ \bsone(e_A) +zP(\nfi^{\star}\ellc,\bU(e_A,\cdot)) ,\\
\label{[Y](eA,eB):abs:mult:triv} [\bY](e_A,e_B)=&\spc \bs{\mathcal{Y}}(e_A,e_B)+\frac{z^2}{2}\lp P(\nfi^{\star}\ellc,\nfi^{\star}\ellc)-\nfi^{\star}\elltwo\rp \bU(e_A,e_B)-ze_A^ae_B^b\nablao_{(a}(\nfi^{\star}\ellc)_{b)}.
\end{align}
%
%
Consequently, when a no-shell matching is possible, the jump $[\bY]$ corresponding to any other possible matching is given by \eqref{[Y](n,n):abs:mult:triv}-\eqref{[Y](eA,eB):abs:mult:triv} in terms of the data fields $\{\gamma,\ellc,\elltwo,\bY^-\}$ and the diffeomorphism $\nfi$. In other words, knowing the information about the no-shell matching automatically allows one to obtain the gravitational/matter-energy content of the remaining matchings by simply determining $\nfi$. 
In particular, there is no need to compute the new matching rigging $\rig^+$ or the tensor $\bY^+$ to determine the shell properties. 
One simple needs to compute the right-hand sides of \eqref{[Y](n,n):abs:mult:triv}-\eqref{[Y](eA,eB):abs:mult:triv} using \eqref{calig:Y:when:no:shell}.

We emphasize that 
\eqref{[Y](n,n):abs:mult:triv}-\eqref{[Y](eA,eB):abs:mult:triv} apply, in particular, when $(\Mpm,g^{\pm})$ are two regions of the same spacetime $(\M,g)$ and more than one matching can be performed. Then, the existence of a no-shell matching is always guaranteed, as one can always recover the full spacetime $(\M,g)$ from the matching of $(\Mpm,g^{\pm})$. This in fact occurs in all cut-and-paste constructions,  
which means that \textit{\eqref{[Y](n,n):abs:mult:triv}-\eqref{[Y](eA,eB):abs:mult:triv} 
provide the matter content of a null shell generated by \textit{any} cut-and-paste matching procedure, as long as the 
two regions $(\Mpm,g^{\pm})$ of $(\M,g)$ 
can be pasted in more than one way}. 

We conclude this section by discussing a particular situation of interest, namely the case when a null hypersurface data $\D=\{\N,\gamma,\ellc,\elltwo,\bY^-\}$ 
can be embedded in two spacetimes $(\Mpm,g^{\pm})$ with embeddings $\iota^{\pm}$ (such that $\iota^{\pm}(\N)$ are boundaries of $\Mpm$) and riggings $L^{\pm}$
\tcr{with the appropriate orientation}. This means that $(\Mpm,g^{\pm})$ can be matched so that the resulting spacetime contains no shell (because $\bY^-$ is the same for both spacetimes). We assume, in addition, that $\D$ admits a vector field $\ov{\xi}\in\Gamma(T\N)$ with the property $\pounds_{\ov{\xi}}\gamma=0$. 
%
%
%
%
The vector $\ov{\xi}$ defines a \tcr{(local)} one-parameter group of transformations \tcr{$\{\nfi_{t}\}$} of $\N$ satisfying 
\begin{equation}
  \label{nfi(Y^-):equals:Y^-}\nfi^{\star}_t\gamma=\gamma.
  \end{equation}
We now prove that, for each value of $t$, the diffeomorphism $\nfi_t$ gives rise to a matching. 
First, we define gauge parameters $\{z,V\}$ according to \eqref{z:and:V:ABS:mult} for $\nfi=\nfi_t$. Then, it is immediate to check that \eqref{JC:HD:multiple} holds for $\nfi=\nfi_t$ and that $z>0$ (because $\nfi_t$ \tcr{depends continuously on $t$}  and $(\nfi_{t=0}^{\star}\ellc)(n)=\ellc(n)=1$). 
Therefore, conditions $(i)$ and $(ii)$ in Theorem \ref{thm:Best:THM:JC} are both fulfilled (notice that, since $L^{\pm}$ are matching riggings, one points inwards and the other outwards, so $(ii)$ \tcr{is just  $z>0$}) and indeed 
each $\nfi_t$ corresponds to a different matching. The jump $[\bY]\defi \bY^+-\bY^-$ where 
$\bY^+\defi \frac{1}{2}\nfi_t^{\star}\big((\iota^+)^{\star}(\pounds_{\rig^+}g^+)\big)$ (and $\rig^+$ is given by \eqref{rig:general:matching}) 
rules the gravitational/matter-energy content of the resulting shell. The vector field $\ov{\xi}$ generates a multitude of 
new shells. 
The construction is further simplified
  when, in addtion to \eqref{nfi(Y^-):equals:Y^-}, it holds
\begin{equation}
  \label{nfi(Y^-):equals:Y^-two}
    \nfi_t^{\star}\bY^-=\bY^-.
  \end{equation}  
  Then \eqref{calig:Y:when:no:shell} implies
\begin{equation}
\label{calig:Y:z-1:Y}\bs{\mathcal{Y}}=\lp z -1\rp \bY^-,
\end{equation}
which simplifies the expressions \eqref{[Y](n,n):abs:mult:triv}-\eqref{[Y](eA,eB):abs:mult:triv} considerably. One may wonder what is the final result when, in addition, $\ov{\xi}$ is the restriction to $\N$ of a Killing vector field $\xi$ on $\Ml$ (i.e.\ $\iota^-_{\star}\ov{\xi}=\xi$) and $\pounds_{\xi}L^-=0$ is fulfilled
\tcr{(so that
  \eqref{nfi(Y^-):equals:Y^-}  and \eqref{nfi(Y^-):equals:Y^-two} hold)}.
 It is straightforward to see that
\begin{equation}
\nfi_t^{\star}\ellc=\ellc,\qquad\nfi_{t}^{\star}\elltwo=\elltwo,
\end{equation} 
\tcr{which combined with} \eqref{z:and:V:ABS:mult} means that $z=1$, and $V=0$, so $\bs{\mathcal{Y}}=0$ (cf.\ \eqref{calig:Y:z-1:Y}). Moreover, one can easily check that \tcr{the terms in the} right-hand side of \eqref{[Y](n,n):abs:mult:triv}-\eqref{[Y](eA,eB):abs:mult:triv} cancel out.  \tcr{Thus, the procedure gives rise to another no-shell matching, as one would expect because the transformation induced by $\xi$  does not affect in any geometric way the spacetime $(\Ml,g^-)$.}
This constitutes a non-trivial consistency check of equations \eqref{[Y](n,n):abs:mult:triv}-\eqref{[Y](eA,eB):abs:mult:triv}.

\section{Null boundaries with product topology $S\times\mathbb{R}$}\label{sec:matching:fol:abs}
\label{the whole section has been rewritten}

In order to connect the results in this paper with those from \cite{manzano2021null}, \cite{manzano2022general} (see Section \ref{sec:Matching:Paper:Big:Sec}), 
%
%
we now consider the case when the boundaries of the spacetimes to be matched can be foliated by cross-sections. In particular, we shall construct a step function $H$ and provide explicit expressions for the gauge parameters $\{z,V\}$ (cf.\ \eqref{z:and:V:ABS}). 
The results for the jump $[\bY]$ will be then compared with their counterparts from Proposition \ref{prop6}.
%
%


Our setup for the present section is the following. We consider two spacetimes $(\Mpm,g^{\pm})$ with null boundaries $\nullhyp^{\pm}$ and assume that $\nullhyp^{\pm}$ have product topology $S^{\pm}\times\mathbb{R}$, where $S^{\pm}$ are spacelike cross-sections and the null generators are along $\mathbb{R}$. 
We select two future null generators $k^{\pm}\in\Gamma(T\Mpm)\vert_{\nullhyp^{\pm}}$
  of $\nullhyp^{\pm}$ and two cross-sections $S^{\pm}_0 \subset
\nullhyp^{\pm}$. We then construct 
foliation functions $v_{\pm}\in\Fcal(\nullhyp^{\pm})$ by solving
$k^{\pm}(v_{\pm})=1$ with initial values $v_{\pm}\vert_{S_0^{\pm}}=0$. Finally, the riggings $L^{\pm}$  are fixed by the conditions of being  orthogonal to
the respective leaves $\{ v_{\pm} = \mbox{const}\}$, null  and scaled to satisfy
$\mu_1^{\pm}\defi g^{\pm}(L^{\pm},k^{\pm})=1$.

We assume that $(\Mpm,g^{\pm})$ can be matched, so that conditions $(i)$-$(ii)$ in Theorem \ref{thm:Best:THM:JC} are fulfilled for a diffeomorphism $\nfi:\N\longrightarrow\N$ verifying \eqref{isom:cond:abs}. This allows us to take two embeddings $\iota^{\pm}:\N\longhookrightarrow\Mpm$ and construct the hypersurface data sets $\D=\{\N,\gamma,\ellc,\elltwo,\bY^-\}$, $\hatD=\{\N,\hatg,\hatellc,\hatelltwo,\hatbY^+\}$ according to \eqref{def:D,hatD:1}-\eqref{def:D,hatD:2}. 
%
%
%
%
We also introduce the functions
\begin{equation}
\label{def:lambda:v:H:abs}\lambda\defi (\iota^-)^{\star}(v_-),\qquad v\defi (\iota^+)^{\star}(v_+), \qquad\text{and}\qquad H\defi \nfi^{\star}v
\end{equation}
on $\N$. Since by construction $\iota^-_{\star}(n)=k^-$ and $\iota^+_{\star}(\hatn)=k^+$ (recall \eqref{normal}), it is immediate to check that $\{\lambda,v\}$ are foliation functions of $\N$. Note that, \tcb{also by construction}, the data satisfies
\begin{equation}
  \label{gauge:D:hatD:fol}
  \ellc=d\lambda,\qquad\elltwo=0,\qquad \qquad \quad \hatellc=dv,\qquad \hatelltwo=0,
\end{equation}
which has the following immediate consequences
\begin{equation}
  n(\lambda)=1,\quad \bF=0, \quad \bsone=0, \qquad \qquad
\hatn(v) = 1, \quad \widehat{\bF}=0, \quad \widehat{\bsone}=0. 
\label{consequences}
\end{equation}
We now select vector fields $\{e_A\}$ tangent to the leaves
  $\{\lambda=\text{const}.\}$ so that 
  $\{n,e_A\}$ is a basis  of $\Gamma(T\N)$ satisfying $[n,e_A]=0$. As before, we let $h$
  be induced metric on $\{\lambda=\text{const}.\}$
  and $\nabla^h$ for its Levi-Civita derivative. In particular  $h_{AB}\defi \gamma(e_A,e_B)$ and we note that, for any
  $f\in\Fcal(\N)$, we can write $e_A(f)$ also as $\nabla^h_Af$. The pull-back
of $\ellc$ to the leaves of constant $\lambda$ is zero, so \ $\ell_A  = \psi_A=0$. This, together with $\elltwo=0$ and \eqref{Pdecom:abstract}, means that $P=h^{AB}e_A\otimes e_B$. 
Observe also that 
\begin{equation}
\label{nfi:star:hatellc=dH}\nfi^{\star}\hatellc=\nfi^{\star}dv=d(\nfi^{\star}v)=dH,\qquad  \nfi^{\star}\hatelltwo=0,
\end{equation}
which in particular means that 
\begin{equation}
\label{P(dH,-):whatever}P(\nfi^{\star}\hatellc,\cdot)=P(dH,\cdot)=h^{AB}(\nabla^h_AH)e_B.
\end{equation}
Inserting these properties in \eqref{z:and:V:ABS} fixes the matching gauge parameters to be
\begin{equation}
\label{z:and:V:ABS:fol}z=\frac{1}{n(H)},\qquad V =h^{AB}\nabla^h_AH\lp  \frac{ \nabla^h_BH}{2n(H)} n-e_B\rp.
\end{equation}
The push-forward vector fields $\{\nfi_{\star}n,\nfi_{\star}e_A\}$ can also be computed in terms of the function $H$ and the vector fields $W_A\defi \widehat{P}(\bs{W}_A,\cdot)$. The result is an easy consequence of
  Corollary \ref{cor:nfistar(n,eA)} and reads
  
\vspace{-0.2cm}

\begin{multicols}{2}
\noindent
\begin{align}
\label{nfistar(n):general:abs:fol}\nfi_{\star}n&=(\nfi^{-1})^{\star}(n(H))\hatn,\\
  \label{nfistar(e_A):general:abs:fol}
  \nfi_{\star}e_A&=W_A+(\nfi^{-1})^{\star}(e_A(H))\hatn.
\end{align}
\end{multicols}

\vspace{-0.4cm}

Observe that  $\{W_A\}$ 
are tangent to the leaves $\{v=\text{const.}\}$ (because by Corollary \ref{cor:nfistar(n,eA)} we know that $ 0=\hatP(\bs{W}_A,\hatellc)=\hatellc(W_A)=W_A(v)$).
Let us now  prove that $\hatn$ and $W_A$ commute.
  \begin{lemma}
    \label{nW_A}
  The vector fields $\hatn$ and $W_A$ satisfy
  $[\hatn,W_A] =0$.
\end{lemma}
\begin{proof}
Define the  functions $\mathfrak{u}\defi (\nfi^{-1})^{\star}(n(H))$ and $\chi_{(A)}\defi (\nfi^{-1})^{\star}(e_A(H))$,
so that \eqref{nfistar(n):general:abs:fol}-\eqref{nfistar(e_A):general:abs:fol}
can be written as
$ \hatn = \mathfrak{u}^{-1} \nfi_{\star}n$ and
$ W_A = \nfi_{\star}e_A-\chi_{(A)}\hatn$. Thus,
\begin{align}
  \nn [\hatn,W_A]&=[\mathfrak{u}^{-1} \nfi_{\star}n,\nfi_{\star}e_A]-[\hatn,\chi_{(A)}\hatn]=\mathfrak{u}^{-1} \nfi_{\star}([n,e_A])+
\mathfrak{u}^{-2} \nfi_{\star}e_A(\mathfrak{u}) \nfi_{\star}n-\hatn(\chi_{(A)})\hatn 
\\
 & =  \mathfrak{u}^{-2} \big (
                   \nfi_{\star}e_A(\mathfrak{u}) 
- \nfi_{\star} n ( \chi_{(A)} ) \big ) \nfi_{\star} n, 
%
\end{align}
where in the last equality we used $[n,e_A]=0$ and $\hatn = \mathfrak{u}^{-1}
\nfi_{\star}n$. To prove the claim we just need to show that the last parenthesis is zero. Indeed, 
\begin{align*}
\nfi_{\star}e_A(\mathfrak{u})-\nfi_{\star}n(\chi_{(A)})&=(d\mathfrak{u})(\nfi_{\star}e_A)-(d\chi_{(A)})(\nfi_{\star}n)=\nfi^{\star}(d\mathfrak{u})(e_A)-\nfi^{\star}(d\chi_{(A)})(n)\\
&=(d \nfi^{\star}\mathfrak{u})(e_A)-(d \nfi^{\star}\chi_{(A)})(n)=e_A(n(H))-n(e_A(H))=[e_A,n](H)=0.
\end{align*} 
\end{proof}
By Remark \ref{rem:W_A} we also know that $\{\hatn,W_A\}$ constitute a basis of $\Gamma(T\N)$ and hence the vector fields 
\begin{equation}
\label{def:abstract:e_a^+}\{ e_1^-\defi \iota^-_{\star}n, \quad e_A^-\defi \iota^-_{\star}e_A\},\qquad\{ e_1^+\defi \iota^+_{\star}(\nfi_{\star}n), \quad e_A^+\defi \iota^+_{\star}(\nfi_{\star}e_A)\}
\end{equation}
form basis of $\Gamma(T\nullhyp^{\pm})$ respectively. Inserting \eqref{nfistar(n):general:abs:fol}-\eqref{nfistar(e_A):general:abs:fol} into \eqref{def:abstract:e_a^+} and using again that $\iota_{\star}^+(\hatn)=k^+$, one obtains
\begin{equation}
\label{e_a^+:explicit:abstract:762483}e_1^+= n(H) k^+, \qquad e_A^+=e_A(H)k^++\iota^+_{\star}(W_A),
\end{equation}
where for simplicity we have dropped pull-backs affecting functions. 
 Given that
  $\{\iota^+_{\star}W_A\}$ are  linearly independent and tangent 
  to the leaves $\{v_+=\text{const.}\}\subset\nullhyp^+$, they can be decomposed \tcr{in a basis $\{L^+,k^+,v_A^+\}$ of $\Gamma(T\Mp)\vert_{\nullhyp^+}$ satisfying \eqref{basis}}   
%
%
as $\iota_{\star}^+W_A=b_A^Bv_B^+$, with $\{b_A^B\}$ defining an invertible matrix. Moreover, $b_A^B$ are constant along the null generators as a consequence of Lemma \ref{nW_A}:
  \begin{equation*}
      0 =  [\iota_{\star}(\hatn), \iota_{\star}(W_A)] =
[ k^{+}, b^B_A v_A^{+} ] = k^{+} (b^B_A) v^{+}_B \qquad \Longleftrightarrow
 \qquad k^+(b^B_A)=0.
  \end{equation*}

%

\tcr{The matching rigging $\rig^+$, obtained by inserting \eqref{z:and:V:ABS:fol} into \eqref{rig:general:matching} and using \eqref{nfistar(n):general:abs:fol}-\eqref{nfistar(e_A):general:abs:fol},  \eqref{def:abstract:e_a^+}-\eqref{e_a^+:explicit:abstract:762483}, reads} 
\begin{align}
\rig^+
%
%
\label{rig:general:matching:fol} &=\frac{1}{n(H)}\lp L^+ -h^{AB}\nabh_AH\lp \iota^+_{\star}(W_A)+\frac{\nabh_BH}{2} k^+ \rp \rp,
\end{align}
which one easily checks to be the same as \eqref{uniqxi} simply by noting that (in the notation of Section \ref{sec:Matching:Paper:Big:Sec}) our choice of $L^{\pm}$ entails  $\mu_1^{\pm}=1$, $\mu_A^{\pm}=0$ and that
$\iota_{\star}^+W_A=b_A^Bv_B^+$ gives $h^{AB}=h_+^{IJ}(b^{-1})_I^A(b^{-1})_J^B$.

The expressions for $[\bY]$ are obtained as a particular case of Theorem \ref{thm:Best:THM:matching}.
\begin{theorem}\label{thm:matter:content:mult:6145375688}
In  the setup and conditions of  Theorem \ref{thm:Best:THM:matching} 
suppose  further that the boundaries $\nullhyp^{\pm}$ can be foliated by cross-sections and define $\lambda,v,H\in\Fcal(\N)$ as in \eqref{def:lambda:v:H:abs}.  Let $h$ be the induced metric and $\nabh$ the corresponding Levi-Civita covariant derivative  on the leaves $\{\lambda=\text{const.}\}\subset\N$.
Then,
\begin{equation}
\label{[Y]:general:matching:fol}[\Y_{ab}]=\frac{1}{n(H)}\lp (\nfi^{\star}\hatbY^+)_{ab}+\frac{ h^{AB}(\nabla^h_AH)(\nabla^h_BH) }{2 n(H)}\U_{ab}-\nablao_{a}\nablao_{b}H\rp-\Y^-_{ab}.
\end{equation}
Let $\{ e_A\}$ be vector fields in $\N$ such that $\{n , e_A\}$ is a basis
adapted to the foliation $\{ \lambda = \textup{const.}\}$ and define
$W_A$ by means of \eqref{nfistar(e_A):general:abs:fol}. Then 
the components the jump $[\bY]$ can be written as
\begin{align}
\label{[Y](n,n):abs:fol:point}[\bY](n,n)=&\spc n(H)\nfi^{\star}\big( \hatbY^+(\hatn,\hatn)\big)-\bY^-(n,n)-\frac{n(n(H))}{n(H)},\\
\nn [\bY](n,e_A)= &\spc \nfi^{\star}\big(\hatbY^+(\hatn,W_A)\big)+(\nabla^h_AH) \nfi^{\star}\big(\hatbY^+(\hatn,\hatn)\big)-\bY^-(n,e_A)\\
\label{[Y](n,eA):abs:fol:point} & -\frac{\nabla^h_A(n(H))}{n(H)}+\frac{h^{IJ}\nabla_I^hH}{n(H)} \bUp(e_A,e_J), \\
\nn [\bY](e_A,e_B)=&\spc \frac{1}{n(H)}\Bigg({\nfi^{\star}\big(\hatbY^+(W_A,W_B)\big)}+{2(\nabla^h_{(A}H) \nfi^{\star}\big(\hatbY^+(\hatn,W_{B)})\big)}+{(\nabla^h_AH) (\nabla^h_BH) \nfi^{\star}\big(\hatbY^+(\hatn,\hatn)\big)}\\
\label{[Y](eA,eB):abs:fol:point} &-n(H)\bY^-(e_A,e_B)+\frac{h^{IJ}(\nabla_I^hH)(\nabla_J^hH)}{2n(H)} \bUp(e_A,e_B)-{\nabla^h_A\nabla^h_BH}\Bigg).
\end{align}
\end{theorem}
\begin{proof}
  Equation \eqref{[Y]:general:matching:fol} follows at once after inserting \eqref{nfi:star:hatellc=dH}-\eqref{z:and:V:ABS:fol} into \eqref{[Y]:general:matching}.  To obtain \eqref{[Y](n,n):abs:fol:point}-\eqref{[Y](eA,eB):abs:fol:point}, it suffices to particularize \eqref{[Y](n,n):abs:point}-\eqref{[Y](eA,eB):abs:point} for $z^{-1}=n(H)$, $\nfi^{\star}\hatellc=dH$, $\chi_{(A)}=(\nfi^{-1})^{\star}(e_A(H))$, $\nfi^{\star}\hatelltwo=0$, $\bsone=0$ and $P(\nfi^{\star}\hatellc,\cdot)=h^{AB}(\nabla_A^hH)e_B$ and notice that $\pounds_{n}(\nfi^{\star}\hatellc)=\pounds_{n}dH=d(n(H))$, as well as $e_A^ae_B^b\nablao_{a}\nablao_{b}H=\nabh_A\nabh_BH$ (see   \eqref{covderpcovtensoronS}).
\end{proof}

Before establishing the connection between  \eqref{[Y](n,n):abs:fol:point}-\eqref{[Y](eA,eB):abs:fol:point} and the corresponding expressions in Proposition \ref{prop6} we  need to relate hypersurface data quantities with the tensors defined in \eqref{somedefs}.
\begin{lemma}
\label{geometric_quantities_sub}
  Let $\hypdata$ be $\{\phi,\rig\}$-embedded in $(\M,g)$,  $k := \phi_{\star}n$  the corresponding null generator and $\ke_{k}$  its surface gravity. 
  Consider a transverse submanifold $S \subset \N$ and assume that 
the gauge is such that the rigging
  $\rig$ is null and orthogonal to $\phi(S)$. 
  Then, for any basis $\{ e_A\}$ of $\Gamma(TS)$ it holds (we identify scalars and vectors with their images on $\phi(\N)$)
\begin{multicols}{2}
\noindent
\begin{itemize}
\item[$(a)$] $\ke_{k}=-\bY(n,n)$,
\item[$(b)$] $\bs{\sigma}_{\rig}(e_A)=\bY(e_A,n) + \bF(e_A,n)$, 
\item[$(c)$] $\btsff^{k}(e_A,e_B)=\bU(e_A,e_B)$.
\item[$(d)$] $\bs{\Theta}^{\rig}(e_{(A},e_{B)})=\bY(e_A,e_B)$,

\end{itemize}
\end{multicols}
\end{lemma}  
\begin{remark}
  This result is a particular case of a much more general analysis on the geometry of embedded submanifold in a hypersurface data set \tcr{carried out in
  \textup{\cite{mars2023covariant}}}. We include the proof for completeness.
\end{remark}
\begin{proof}
Claim $(a)$ follows at once from \eqref{defY(n,.)andQ} and \eqref{Qmeaning} (note that here $\nu=k$). To prove $(b)$ we compute
\begin{align*}
\bs{\sigma}_{\rig}(e_A)  \stackbin{\eqref{somedefs}}=
-g(\nabla_{e_A}k,\rig)
= g(\nabla_{e_A}\rig,k)  \stackbin{\eqref{nablaXrig}}=
\bY(e_A,n) + \bF(e_A,n).
\end{align*}
Item $(c)$ has already been stated after definition \eqref{2FF} and  $(d)$ 
follows from
\begin{align*}
\bY(e_A,e_B)=&\spc\frac{1}{2} (\pounds_{\rig} g) (e_A,e_B)=
g\big(\nabla_{e_{(A}}\rig,e_{B)}\big)\stackbin{\eqref{somedefs}}=
=\bs{\Theta}^{\rig}\big(e_{(A},e_{B)}\big).
\end{align*}
\end{proof}
We are now in a position where the comparison can be made.
We identify the vector fields $\{v_A^-\}$ \tcr{introduced in Section \ref{sec:Matching:Paper:Big:Sec}} with the push-forwards of $\{e_A\}$, \tcr{hence $\mu_1^-=1$ and $\mu^-_A=0$. On the other hand, $\mu_1^+=1$ and $\mu_A^+=g^{+}(L^+,v_A^+)
=(b^{-1})_A^B\hatellc(W_B)
=(b^{-1})_A^BW_B(v)=0$}, so the covector $q$ 
defined in Proposition \ref{prop6} 
is simply $q_A=-\nabh_AH$. The vector $X^a$ in \eqref{X1} 
is \tcr{in turn} given by
\begin{equation}
X^1=\frac{h^{AB}\nabh_AH\nabh_BH}{2n(H)},\qquad X^A=-h^{AB}\nabh_BH.
\end{equation}
\tcr{Thus,} expressions \eqref{Y11}-\eqref{YIJ} become 
\begin{align}
\label{Y11:abs}[\bY](n,n)=& -n(H) \ke_{k^+}^+ +\ke^-_{k^-} - \dfrac{n(n(H))}{n(H)},\\
\label{Y1J:abs}   [\bY](n,e_J)=&\spc \bs{\sigma}_{L}^+(  W_J)-\bs{\sigma}^-_{L}( v^-_J)-(\nabla_{J}^hH)\ke_{k^+}^+ -\dfrac{\nabh_J(n(H))}{n(H)}+\dfrac{h^{LB}\nabh_BH}{n(H)}{\btsff}_-^{k}( v_J^-,v_L^-),\\
\nonumber [\bY](e_I,e_J)=&\spc\frac{1}{n(H)}\bigg(
2(\nabla_{{\lp I\rd}}^hH)\bs{\sigma}_{L}^+( W_{\ld J\rp})-\ke_{k^+}^+(\nabla_{I}^hH )(\nabla_{J}^hH)
+\bs{\Theta}^{L}_+( W_{\lp I\rd},W_{\ld J\rp})-n(H)\bs{\Theta}^{L}_-( v_{\lp I \rd}^-,v_{\ld J \rp}^-)\\
\label{YIJ:abs}&+ \frac{\gamma^{AB}\nabh_AH\nabh_BH}{2n(H)}{\btsff}_-^{k}( v_I^-,v_J^-)-\nabla_{I}^h\nabla_{J}^hH\bigg).
\end{align}
Particularizing Lemma \ref{geometric_quantities_sub}
to the sections $\{ \lambda= \mbox{const}\}$ of $\D$  (with basis $e_A$)
and the sections $\{v= \mbox{const}\}$ of $\hatD$  (with basis $W_A$), and recalling that $\bF= \widehat{\bF} =0$ (see \eqref{consequences}),  
it is  straightforward to check that \eqref{Y11:abs}-\eqref{YIJ:abs} coincide with \eqref{[Y](n,n):abs:fol:point}-\eqref{[Y](eA,eB):abs:fol:point}.

\section{Cut-and-paste matching: (anti-)de Sitter spacetime}\label{sec:cut-and-paste:abs}
We have already mentioned that \eqref{[Y](n,n):abs:mult:triv}-\eqref{[Y](eA,eB):abs:mult:triv} hold for the specific case when the two spacetimes to be matched are actually two regions of the same spacetime (and more than one matching is allowed). In this section, our aim is to provide an example of a cut-and-paste construction, namely the matching of two regions of a constant-curvature spacetime across a totally geodesic null hypersurface.  
For previous works on the cut-and-paste construction describing non-expanding impulsive gravitational waves in constant curvature backgrounds we refer e.g.\ to \cite{podolsky1999nonexpanding}, \cite{podolsky2002exact}, \cite{griffiths2009exact}
\cite{podolsky2017penrose}, \cite{podolsky2019cut} and references therein. 

%
%
%

In any constant curvature spacetime $(\M,g)$ 
there exists only one totally geodesic null hypersurface up to isometries (see e.g.\  \cite{ferrandez2001geometry}, \cite{navarro2016null}). We denote one such hypersurface by $\nullhyp$. Then, one can always construct coordinates $\{\cu,\cv,x^A\}$ adapted to $\nullhyp$ so that the metric is conformally flat and $\nullhyp\defi\{\cu=0\}$, namely
%
%
\begin{align}
\label{metric:AdSMk}
g=\dfrac{g_{Mk}}{\mu^2},\quad
\text{where}\quad g_{Mk}= -2d\cu d\cv+\delta_{AB}dx^Adx^B,\quad \mu\defi 1+\dfrac{\Lambda}{12}\lp \delta_{AB}x^Ax^B-2\cu\cv\rp.
%
%
\end{align}
Here $\Lambda$ stands for the cosmological constant, so $\Lambda=0$, $\Lambda>0$, $\Lambda<0$ correspond to Minkowski, de Sitter and anti-de Sitter spacetimes respectively. 
When $\Lambda\leq0$, the coordinates $\{\cu,\cv,x^A\}$ cover a whole neighbourhood of $\nullhyp$. However, for the de Sitter case one needs to remove one generator of $\nullhyp$ 
because the topology of $\nullhyp$ is $\mathbb{S}^{\n}\times \mathbb{R}$ while stereographic coordinates only cover the sphere minus one point. 
%
%
%
%
In this section, we will analyze the three cases $\Lambda=0$, $\Lambda<0$ and $\Lambda>0$ at once with the matching formalism introduced before. 

%
%
%

The induced metric on $\nullhyp$ reads 
$ds^2\stackbin{\nullhyp}=\lp 1+\frac{\Lambda}{12} \delta_{AB}x^Ax^B\rp^{-2}\delta_{AB}dx^Adx^B$, 
and obviously the topology of $\nullhyp$ is $S\times\mathbb{R}$, $S$ being a spacelike section and the null generators being along $\mathbb{R}$. Therefore, all results from Section \ref{sec:matching:fol:abs} can be applied.

Let us construct hypersurface data associated to $\nullhyp$.
Since $\nullhyp$ is embedded on $(\M,g)$, there exists an abstract manifold $\N$ and an embedding $\iota$ such that $\iota(\N)=\nullhyp$. We can select $\iota$ to be as trivial as possible by constructing coordinates $\{\lambda,y^A\}$ on $\N$ so that 
\begin{equation}
\begin{array}{clcl}
\iota: & \N & \longhookrightarrow & \nullhyp\\
	   & (\lambda,y^A) & \longmapsto & \iota(\lambda,y^A)\equiv(\cu=0,\cv=\lambda,x^A=y^A). 
\end{array}
\end{equation}
We also need a choice of rigging vector field $\rig$ along $\nullhyp$. For convenience, we set $\rig=-\mu^2\cp_{\cu}$ (observe that $\mu^2\vert_{\nullhyp}\neq0$). The corresponding null metric hypersurface data
   \eqref{emhd} defined by  $\{\N,\gamma,\ellc,\elltwo\}$
  is 
\begin{equation}
\label{data:AdS}\gamma=\frac{\delta_{AB}}{\mu_{\N}^2}dy^A\otimes dy^B,\qquad \ellc=d\lambda,\qquad\elltwo=0,
\end{equation}
where  $\mu_{\N}\defi \iota^{\star}\mu=1+\frac{\Lambda}{12} \delta_{AB}y^Ay^B$. 
Observe that $\cp_{\lambda}\in\text{Rad}\gamma$ and $\ellc(\cp_{\lambda})=1$  
imply that $n=\cp_{\lambda}$. Moreover,  
$\bF=0$ and $\bsone=0$ (cf.\ \eqref{defF}-\eqref{sone})
and $\bU=0$ as a consequence of \eqref{defUtensor}.  
The tensor $\bY$ is obtained from \eqref{YtensorEmbDef}. A simple calculation gives 
\begin{align}
  \label{bY:AdS}  \bY&=-\frac{\Lambda \delta_{AB}}{6\mu_{\N}}\lp \lambda dy^A\otimes dy^B- 2  y^Bdy^A \otimes_s d\lambda \rp.
\end{align}
Cutting the spacetime across the hypersurface $\{ \cu=0\}$ leaves two spacetimes
  %
  $(\Mpm,g^{\pm})$ defined to be the regions $\cu\gtreqless0$ endowed with the metrics
\begin{align}
\hspace{-0.15cm}g^{\pm}=\frac{g^{\pm}_{Mk}}{\mu_{\pm}^2},\quad
\text{where}\quad g^{\pm}_{Mk}\defi -2d\cu_{\pm}d\cv_{\pm}+\delta_{AB}dx_{\pm}^Adx_{\pm}^B,\quad\mu_{\pm}\defi 1+\frac{\Lambda}{12}\lp \delta_{AB}x_{\pm}^Ax_{\pm}^B-2\cu_{\pm}\cv_{\pm}\rp.
\end{align}
Obviously, the boundaries are  $\nullhyp^{\pm}\equiv\{\cu_{\pm}=0\}$. These two regions can clearly be matched so that the original spacetime (containing no shell) is obtained. Moreover, since $\nullhyp^{\pm}$ are totally geodesic
  we know that multiple matchings can be performed. 
We therefore proceed as in Section \ref{sec:mult:mathcings}, i.e.\ we let the two embeddings $\iota^{\pm}$ be given by $\iota^{\pm}=\iota$ and take $\rig^-=-\mu_-^2\cp_{\cu_-}$, $\widetilde{\rig}^+=-\mu_+^2\cp_{\cu_+}$ as the riggings defining the no-shell matching, namely the matching for which $[\wt{\bY}]=0$. Any other possible matching will be ruled by a diffeomorphism $\nfi$ of $\N$ onto itself and it will correspond to a different rigging $\rig^+$ along $\nullhyp^+$. Specifically, the hypersurface data corresponding to the no-shell matching is $\D=\{\N,\gamma,\ellc,\elltwo,\bY\}$, where $\{\gamma,\ellc,\elltwo\}$ and $\bY$ are respectively given by \eqref{data:AdS} and \eqref{bY:AdS}, while the matter/gravitational content of the shell of any other possible matching (ruled by $\nfi$) is given by the the jump $[\bY]\defi \bY^+-\bY$ with
\begin{equation}
\bY^+\defi \frac{1}{2}\nfi^{\star}\Big( \iota^{\star}(\pounds_{\rig^+}g^+)\Big).
\end{equation} 
%
%
%
From Section \ref{sec:mult:mathcings}, we know that there is no need to compute the new rigging $\rig^+$ or its corresponding $\bY^+$ \tcb{to determine the jump $[\bY]$, which} is explicitly given by
  \eqref{[Y](n,n):abs:mult:triv}-\eqref{[Y](eA,eB):abs:mult:triv}. Consequently, we only need to worry about the diffeomorphism $\nfi$. The only restriction that $\nfi$ must satisfy is 
$\nfi^{\star}\gamma=\gamma$, which in coordinates reads
\begin{equation}
\frac{(\cp_{y^a}\nfi^A)(\cp_{y^b}\nfi^B)\delta_{AB}}{(1+\frac{\Lambda}{12}\delta_{IJ}\nfi^I \nfi^J)^{2}}=\frac{\delta_a^A\delta_b^B\delta_{AB}}{(1+\frac{\Lambda}{12}\delta_{IJ}y^Iy^J)^{2}}.
\end{equation}  
It follows that the components $\{\nfi^A\}$ cannot depend on the coordinate $\lambda$. 
In particular, if we let $\{h^A(y^B)\}$ be a set of functions such that $(a)$ the Jacobian matrix $\frac{\cp(h^2,...,h^{\n+1})}{\cp(y^2,...,y^{\n+1})}$ has non-zero determinant and $(b)$ $\{h^A(y^B)\}$ verify $(1+\frac{\Lambda}{12}\delta_{IJ}y^Iy^J)^{-2}\delta_{CD}=(1+\frac{\Lambda}{12}\delta_{IJ}h^I h^J)^{-2}(\cp_{y^C}h^A)(\cp_{y^D}h^B)\delta_{AB}$, any diffeomorphism $\nfi:\N\longrightarrow\N$ of the form
\begin{equation}
\label{nfi:paper:specific}
\begin{array}{clcl}
\nfi: & \N & \longrightarrow & \N\\
	   & (\lambda,y^B) & \longmapsto & \nfi(\lambda,y^B)\equiv(H(\lambda,y^B),h^A(y^B))
\end{array}
\end{equation}
with $\cp_{\lambda}H\neq0$ fulfils $\nfi^{\star}\gamma=\gamma$.
%
%
%
%
%
%
\tcb{A particular simple example is  $\{h^A=y^A\}$, but many more exist. In fact since the metric on any section of $\widetilde{\N}$ is
  of constant curvature, it is also maximally symmetric (and of dimension $\mathfrak{n}-1$) so $h^A(y^B)$ can depend
  on $\mathfrak{n}(\mathfrak{n}-1)/2$ arbitrary parameters.}
\tcb{For any possible choice of $\{h^A(y^B)\}$ and 
  an arbitrary step function $H(\lambda,y^A)$,}
the gauge parameters $z$ and $V$ are given by \eqref{z:and:V:ABS:fol} for $\{n=\cp_{\lambda},e_A=\cp_{y^A}\}$.
In the present case the tensor $\bs{\mathcal{Y}}$ is given by $\bs{\mathcal{Y}}=\frac{1}{n(H)}\nfi^{\star}\bY-\bY$ (cf. \eqref{calig:Y:when:no:shell}), so we \tcb{need to}  compute the pull-back $\nfi^{\star}\bY$. 
Defining $\ov{\mu}_{\N}\defi 1+\frac{\Lambda}{12}\delta_{AB}h^Ah^B$, from \eqref{bY:AdS} and \eqref{nfi:paper:specific} it is straightforward to get 
\begin{align}
\label{nfi:star:Y:AdS1}(\nfi^{\star}\bY)_{\lambda\lambda}&=0,\qquad (\nfi^{\star}\bY)_{\lambda y^B}=\frac{\Lambda\delta_{IJ}h^J}{6\ov{\mu}_{\N}} \frac{\cp h^I}{\cp y^B}\frac{\cp H}{\cp \lambda},\\
\label{nfi:star:Y:AdS2}(\nfi^{\star}\bY)_{y^A y^B}&=\frac{\Lambda\delta_{IJ}}{6\ov{\mu}_{\N}} \lp h^J\lp \frac{\cp H}{\cp y^A}\frac{\cp h^I}{\cp y^B}+\frac{\cp h^{I}}{\cp y^A}\frac{\cp H}{\cp y^B}\rp -H\frac{\cp h^{I}}{\cp y^A}\frac{\cp h^J}{\cp y^B}\rp
\end{align}
so that, multiplying \eqref{nfi:star:Y:AdS1}-\eqref{nfi:star:Y:AdS2} by $\frac{1}{n(H)}$ and subtracting $\bY$ (cf.\ \eqref{bY:AdS}) yields
\begin{align}
\label{Yraro:equation:yoquese}\mathcal{Y}_{\lambda\lambda}&=0,\qquad \mathcal{Y}_{\lambda y^B}=\frac{\Lambda\delta_{IJ}}{6}\lp \frac{h^J}{\ov{\mu}_{\N}}\frac{\cp h^I}{\cp y^B}-\frac{\delta_B^Iy^J}{\mu_{\N}}\rp,\\
\nn \mathcal{Y}_{y^Ay^B}&=\frac{\Lambda\delta_{IJ}}{6n(H)} \lp \frac{h^J}{\ov{\mu}_{\N}}\lp \frac{\cp H}{\cp y^A}\frac{\cp h^I}{\cp y^B}+\frac{\cp h^{I}}{\cp y^A}\frac{\cp H}{\cp y^B}\rp-\frac{H}{\ov{\mu}_{\N}}\frac{\cp h^{I}}{\cp y^A}\frac{\cp h^J}{\cp y^B} +\frac{\delta^{I}_{A}\delta^{J}_{B} \lambda}{\mu_{\N}} n(H) \rp.
\end{align}
Inserting these expressions into \eqref{[Y](n,n):abs:mult:triv}-\eqref{[Y](eA,eB):abs:mult:triv} and using $n=\cp_{\lambda}$, $\bsone=0$, $\bU=0$ together with the identity $\lp\pounds_n\nfi^{\star}\ellc\rp(e_A)=\lp\pounds_n dH\rp(e_A)= d(n(H))(e_A)=e_A(n(H))$ (here $\nfi^{\star}\ellc=dH$ by \eqref{nfi:star:hatellc=dH} and $\hatellc=\ellc$), one finds 
\begin{align}
\label{jumpY:adS:Mk:dS} [\Y_{\lambda\lambda}]= -\frac{ n(n(H))}{n(H)},\qquad
[\Y_{\lambda y^A}]=  \mathcal{Y}_{\lambda y^A} -\frac{\nabla^h_A(n(H))}{ n(H)},\qquad
 [\Y_{y^Ay^B}]= \mathcal{Y}_{y^Ay^B}-\frac{\nabla^h_A\nabla^h_BH}{n(H)},
\end{align}
which can be interpreted as the sum of the jump  corresponding to the matching of two regions of Minkowski across a null hyperplane  
%
%
(see \cite[Eq. (6.6)]{manzano2021null}) 
plus the contribution of the tensor $\bs{\mathcal{Y}}$. Observe that $\Lambda=0$ entails $\bs{\mathcal{Y}}=0$, so in this way we recover expressions (6.6) in \cite{manzano2021null} 
%
%
for the most general \tcb{planar} shell in the spacetime of Minkowski 

A direct computation that combines the definitions \eqref{def:rho:jA:p:abs}, \eqref{data:AdS} and  \eqref{jumpY:adS:Mk:dS} yields energy-density, energy flux and pressure (note that here we need to take $\fv=-1$)
\begin{align}
\hspace{-0.15cm}\rho=\mu^2_{\N}\delta^{AB}\lp \mathcal{Y}_{y^Ay^B}-\frac{\nabla^h_A\nabla^h_BH}{n(H)}\rp,\quad j=\mu^2_{\N}\delta^{AB}\lp \frac{\nabla^h_B(n(H))}{ n(H)} -\mathcal{Y}_{\lambda y^B} \rp \cp_{y^A},\quad
\label{p:dS:Mk:AdS} p=-\frac{ n(n(H))}{n(H)}.
\end{align}
Observe that only the pressure is 
independent of the value of the cosmological constant $\Lambda$ ($\rho$ and $j$ depend on the conformal factor $\mu_{\N}$ and on $\bs{\mathcal{Y}}$). 
The pressure $p$ takes the same value for the matchings of two regions of (anti-)de Sitter or Minkowski (in fact, $p$ coincides with the pressure obtained in \cite[Sect. 6]{manzano2021null}). In particular, in the case $h^A=y^A$ (i.e.\ when the mapping between null generators of both sides is trivial), then $\mathcal{Y}_{\lambda y^B}=0$ (cf.\ \eqref{Yraro:equation:yoquese}) and 
%
%
\eqref{p:dS:Mk:AdS} 
simplifies to
\begin{align}
\label{p:dS:Mk:AdS:2} \rho=\mu^2_{\N}\delta^{AB}\lp \mathcal{Y}_{y^Ay^B}-\frac{\nabla^h_A\nabla^h_BH}{n(H)}\rp,\qquad j= \mu^2_{\N}\delta^{AB}\frac{\nabla^h_B(n(H))}{ n(H)}  \cp_{y^A},\qquad p=-\frac{ n(n(H))}{n(H)}.
\end{align}
%
%
%
%



In the cut-and-paste constructions corresponding to constant-curvature spacetimes, the so-called \textit{Penrose's junction conditions} (see e.g. \cite{podolsky2017penrose}, \cite{podolsky2019cut}) impose a jump in the coordinates
across  the shell. This jump is of the form $\cv_+\vert_{\cu_+=0}=\cv_-+\mathcal{H}\lp x^A_-\rp\vert_{\cu_-=0}$. 
In the present case the matching embeddings $\phi^-=\iota$ and $\phi^+=\iota\circ\nfi$ are given by
\begin{align}
\nn 
\phi^-(\lambda,y^B)=\Big(\cu_-=0,\cv_-=\lambda,x_-^A=y^A\Big),\qquad \phi^+(\lambda,y^B) = \Big(\cu_+=0,\cv_+=H(\lambda,y^B),x_+^A=h^A(y^B)\Big),
\end{align}
so the step function corresponding to Penrose's jump is $H(\lambda,y^A)=\lambda+\mathcal{H}(y^A)$, $\mathcal{H}\in\Fcal(\N)$.
%
%
In order to recover 
such an $H$, one needs that there is no energy flux and no pressure on the shell. Indeed, imposing this in \eqref{p:dS:Mk:AdS:2} and integrating for $H$ yields $H(\lambda,y^A)=a\lambda+\mathcal{H}(y^A)$, \tcb{where $\mathcal{H}\in\Fcal(\N)$ and $a$ is a positive\footnote{In the present case, $\rig^-$ point inwards w.r.t.\ $(\Ml,g^-)$, so $\rig^+$ points outwards. 
  In these circumstances, condition $(ii)$ in Theorem \ref{thm:Best:THM:JC} imposes $\text{sign}(z)=-\text{sign}(\fv^+)\text{sign}(\fv)=-(+1)(-1)=+1$. This, together with \eqref{z:and:V:ABS:fol}, means that $n(H)>0$ necessarily, \tcb{i.e.\
    $a>0$}.} constant}. Thus, in this more general context with arbitrary cosmological constant, the Penrose's jump 
still describes either purely gravitational waves (when $\rho$, $j$ and $p$ are all zero)
or shells of null dust (when $j$ and $p$ vanish but $\rho\neq0$), analogously to what happened in 
%
%
\cite[Sect. 6]{manzano2021null} 
for the Minkowski spacetime.


\section*{Acknowledgements}

The authors acknowledge financial support under \tcr{Grant PID2021-122938NB-I00 funded by MCIN/AEI /10.13039/501100011033 and
by “ERDF A way of making Europe”} and SA096P20 (JCyL). M. Manzano also acknowledges the Ph.D. grant FPU17/03791 (Spanish Ministerio de Ciencia, Innovaci{\'o}n y Universidades).

\begingroup
\let\itshape\upshape
\bibliographystyle{acm}
\bibliographystyle{chronological}

\bibliography{ref}

\end{document}